\documentclass[conference]{IEEEtran}
\IEEEoverridecommandlockouts

\usepackage{amssymb,amsfonts,amsmath,amsthm}
\usepackage{cite}

\usepackage{graphicx}
\usepackage{xcolor}

\usepackage[caption=false,font=footnotesize]{subfig}

\usepackage[shortlabels,inline]{enumitem}

\usepackage{xstring}

\usepackage{hyperref}
\usepackage{xurl}

\usepackage{datetime2}

\usepackage{siunitx}

\usepackage{varwidth}

\usepackage{tikz}
\usetikzlibrary{calc}
\usetikzlibrary{arrows}
\usetikzlibrary{arrows.meta}
\usetikzlibrary{patterns}
\usetikzlibrary{positioning}
\usetikzlibrary{decorations.pathreplacing}
\usetikzlibrary{shapes.misc}
\usetikzlibrary{spy}
\usetikzlibrary{backgrounds}

\usepackage{pgfplots}
\pgfplotsset{compat=1.14}
\usepgfplotslibrary{fillbetween}

\usepackage{algorithmicx}
\usepackage{algpseudocode}
\usepackage{algorithm}

\usepackage{xspace}
\usepackage{ifthen}

\newcommand{\sac}{snap-and-chat\xspace}
\newcommand{\Sac}{Snap-and-chat\xspace}
\newcommand{\SAC}{Snap-and-Chat\xspace}
\newcommand{\eaf}{ebb-and-flow\xspace}

\newcommand{\EAF}{Ebb-and-Flow\xspace}

\newcommand{\theconstruction}{our construction\xspace}

\newcommand{\Chlc}[2]{\ensuremath{\mathsf{ch}_{#1}^{#2}}}
\newcommand{\Chbft}[2]{\ensuremath{\mathsf{Ch}_{#1}^{#2}}}

\newcommand{\PayloadTxs}[0]{\ensuremath{\mathsf{txs}}}
\newcommand{\PayloadCh}[0]{\ensuremath{\mathsf{ch}}}

\newcommand{\blocklc}[0]{LC block\xspace}
\newcommand{\blockslc}[0]{LC blocks\xspace}
\newcommand{\blockbft}[0]{BFT block\xspace}
\newcommand{\blocksbft}[0]{BFT blocks\xspace}

\newcommand{\Blockslc}[0]{\blockslc}

\newcommand{\LOGda}[2]{%
    \ifthenelse{\equal{#1}{}}{%
        \ensuremath{\mathsf{LOG}_{\mathrm{da}}^{#2}}%
    }{%
        \ensuremath{\mathsf{LOG}_{\mathrm{da},#1}^{#2}}%
    }%
}
\newcommand{\LOGfin}[2]{%
    \ifthenelse{\equal{#1}{}}{%
        \ensuremath{\mathsf{LOG}_{\mathrm{fin}}^{#2}}%
    }{%
        \ensuremath{\mathsf{LOG}_{\mathrm{fin},#1}^{#2}}%
    }%
}

\newcommand{\LOGlc}[2]{%
    \ifthenelse{\equal{#1}{}}{%
        \ensuremath{\mathsf{LOG}_{\mathrm{lc}}^{#2}}%
    }{%
        \ensuremath{\mathsf{LOG}_{\mathrm{lc},#1}^{#2}}%
    }%
}
\newcommand{\LOGbft}[2]{%
    \ifthenelse{\equal{#1}{}}{%
        \ensuremath{\mathsf{LOG}_{\mathrm{bft}}^{#2}}%
    }{%
        \ensuremath{\mathsf{LOG}_{\mathrm{bft},#1}^{#2}}%
    }%
}
\newcommand{\LOG}[2]{%
    \ensuremath{\mathsf{LOG}_{#1}^{#2}}
}

\newcommand{\LOGdaBLANKFIX}[0]{\ensuremath{\mathsf{LOG}_{\mathrm{da}}}}

\newcommand{\flatten}[0]{\ensuremath{\mathsf{flatten}}}
\newcommand{\sanitize}[0]{\ensuremath{\mathsf{sanitize}}}

\newcommand{\PI}[0]{\ensuremath{\Pi}}
\newcommand{\PIlc}[0]{\ensuremath{\Pi_{\mathrm{lc}}}}
\newcommand{\PIbft}[0]{\ensuremath{\Pi_{\mathrm{bft}}}}
\newcommand{\PIda}[0]{\ensuremath{\Pi_{\mathrm{da}}}}
\newcommand{\PIfin}[0]{\ensuremath{\Pi_{\mathrm{fin}}}}

\newcommand{\PIanyexample}[0]{\ensuremath{\Pi_{\mathrm{sac}}}}
\newcommand{\PItheexample}[2]{\ensuremath{\PIanyexample}}

\newcommand{\tx}[0]{\ensuremath{\mathsf{tx}}}
\newcommand{\Txs}[2]{\ensuremath{\mathsf{txs}_{#1}^{#2}}}

\newcommand{\Adv}[0]{\ensuremath{\mathcal A}}
\newcommand{\Env}[0]{\ensuremath{\mathcal Z}}

\newcommand{\AdvEnvParameterized}[2]{\ensuremath{(\Adv_{#1}(#2), \Env_{#1})}}
\newcommand{\AdvEnvOneOpt}[0]{\ensuremath{(\Adv_1^*, \Env_1)}}
\newcommand{\AdvEnvTwoOpt}[0]{\ensuremath{(\Adv_2^*, \Env_2)}}

\newcommand{\ie}[0]{\emph{i.e.}\xspace}
\newcommand{\eg}[0]{\emph{e.g.}\xspace}

\newcommand{\cf}[0]{\emph{cf.}\xspace}
\newcommand{\etal}[0]{\emph{et al.}\xspace}

\newcommand{\Pone}[0]{\textbf{P1}\xspace}
\newcommand{\Ptwo}[0]{\textbf{P2}\xspace}

\newcommand{\GST}[0]{\ensuremath{\mathsf{GST}}}
\newcommand{\GOT}[0]{\ensuremath{\mathsf{GAT}}}

\newcommand{\Prob}[1]{\ensuremath{\operatorname{Pr}\!\left[#1\right]}}

\DeclareMathAlphabet{\mathsfit}{T1}{\sfdefault}{\mddefault}{\sldefault}
\SetMathAlphabet{\mathsfit}{bold}{T1}{\sfdefault}{\bfdefault}{\sldefault}
\newcommand{\const}[0]{\ensuremath{\mathsfit{C}}}
\newcommand{\extra}[0]{\ensuremath{N_e}}
\newcommand{\Tconfirm}[0]{\ensuremath{T_{\mathrm{confirm}}}}
\newcommand{\negl}[0]{\ensuremath{\operatorname{negl}}}

\newcommand{\finalp}[0]{final\xspace}

\newcommand{\finalizationp}[0]{finalization\xspace}

\newcommand{\finalda}[0]{confirmed\xspace}

\theoremstyle{plain}
\newtheorem{theorem}{Theorem}
\newtheorem*{theorem*}{Theorem}
\newtheorem{corollary}{Corollary}
\newtheorem{lemma}{Lemma}
\newtheorem{proposition}{Proposition}

\theoremstyle{definition}
\newtheorem{definition}{Definition}

\definecolor{myParula01Blue}{RGB}{0,114,189}
\definecolor{myParula02Orange}{RGB}{217,83,25}
\definecolor{myParula03Yellow}{RGB}{237,177,32}
\definecolor{myParula04Purple}{RGB}{126,47,142}
\definecolor{myParula05Green}{RGB}{119,172,48}
\definecolor{myParula06LightBlue}{RGB}{77,190,238}
\definecolor{myParula07Red}{RGB}{162,20,47}

\tikzset{myparula11/.style={color=myParula01Blue,solid,mark=+,mark options={solid}}}
\tikzset{myparula12/.style={color=myParula01Blue,densely dashed,mark=x,mark options={solid}}}
\tikzset{myparula13/.style={color=myParula01Blue,densely dotted,mark=o,mark options={solid}}}
\tikzset{myparula14/.style={color=myParula01Blue,dashdotted,mark=triangle,mark options={solid}}}
\tikzset{myparula15/.style={color=myParula01Blue,dashdotdotted,mark=square,mark options={solid}}}

\tikzset{myparula21/.style={color=myParula02Orange,solid,mark=+,mark options={solid}}}
\tikzset{myparula22/.style={color=myParula02Orange,densely dashed,mark=x,mark options={solid}}}
\tikzset{myparula23/.style={color=myParula02Orange,densely dotted,mark=o,mark options={solid}}}
\tikzset{myparula24/.style={color=myParula02Orange,dashdotted,mark=triangle,mark options={solid}}}
\tikzset{myparula25/.style={color=myParula02Orange,dashdotdotted,mark=square,mark options={solid}}}

\tikzset{myparula31/.style={color=myParula03Yellow,solid,mark=+,mark options={solid}}}
\tikzset{myparula32/.style={color=myParula03Yellow,densely dashed,mark=x,mark options={solid}}}
\tikzset{myparula33/.style={color=myParula03Yellow,densely dotted,mark=o,mark options={solid}}}
\tikzset{myparula34/.style={color=myParula03Yellow,dashdotted,mark=triangle,mark options={solid}}}
\tikzset{myparula35/.style={color=myParula03Yellow,dashdotdotted,mark=square,mark options={solid}}}

\tikzset{myparula41/.style={color=myParula04Purple,solid,mark=+,mark options={solid}}}
\tikzset{myparula42/.style={color=myParula04Purple,densely dashed,mark=x,mark options={solid}}}
\tikzset{myparula43/.style={color=myParula04Purple,densely dotted,mark=o,mark options={solid}}}
\tikzset{myparula44/.style={color=myParula04Purple,dashdotted,mark=triangle,mark options={solid}}}
\tikzset{myparula45/.style={color=myParula04Purple,dashdotdotted,mark=square,mark options={solid}}}

\tikzset{myparula51/.style={color=myParula05Green,solid,mark=+,mark options={solid}}}
\tikzset{myparula52/.style={color=myParula05Green,densely dashed,mark=x,mark options={solid}}}
\tikzset{myparula53/.style={color=myParula05Green,densely dotted,mark=o,mark options={solid}}}
\tikzset{myparula54/.style={color=myParula05Green,dashdotted,mark=triangle,mark options={solid}}}
\tikzset{myparula55/.style={color=myParula05Green,dashdotdotted,mark=square,mark options={solid}}}

\tikzset{myparula61/.style={color=myParula06LightBlue,solid,mark=+,mark options={solid}}}
\tikzset{myparula62/.style={color=myParula06LightBlue,densely dashed,mark=x,mark options={solid}}}
\tikzset{myparula63/.style={color=myParula06LightBlue,densely dotted,mark=o,mark options={solid}}}
\tikzset{myparula64/.style={color=myParula06LightBlue,dashdotted,mark=triangle,mark options={solid}}}
\tikzset{myparula65/.style={color=myParula06LightBlue,dashdotdotted,mark=square,mark options={solid}}}

\tikzset{myparula71/.style={color=myParula07Red,solid,mark=+,mark options={solid}}}
\tikzset{myparula72/.style={color=myParula07Red,densely dashed,mark=x,mark options={solid}}}
\tikzset{myparula73/.style={color=myParula07Red,densely dotted,mark=o,mark options={solid}}}
\tikzset{myparula74/.style={color=myParula07Red,dashdotted,mark=triangle,mark options={solid}}}
\tikzset{myparula75/.style={color=myParula07Red,dashdotdotted,mark=square,mark options={solid}}}

\pgfplotsset{
    mysimpleplot/.style = {
        every axis plot/.prefix style={thick},
        width=\linewidth,
        height=0.75\linewidth,
        title style={font=\footnotesize,align=center},
        legend cell align=left,
        legend style={font=\footnotesize},
        legend columns=3,
        legend style={
            at={(0.5,1)},
            yshift=0.3em,
            anchor=south,
            draw=none,
            /tikz/every even column/.append style={
                column sep=0.3em
            },
            cells={
                align=left
            }
        },
        grid=both,
        minor tick num=4,
        major grid style={solid,draw=gray!50},
        minor grid style={densely dotted,draw=gray!50},
        label style={font=\footnotesize,align=center},
        tick label style={font=\footnotesize},
    },
}

\usepackage{comment}
\includecomment{onlyonarxiv}
\excludecomment{onlyinproceedings}
\title{\EAF Protocols:\\ A Resolution of the Availability-Finality Dilemma}

\author{%
\IEEEauthorblockN{Joachim Neu}%
\IEEEauthorblockA{%
jneu@stanford.edu%
}%
\and%
\IEEEauthorblockN{Ertem Nusret Tas}%
\IEEEauthorblockA{%
nusret@stanford.edu%
}%
\and%
\IEEEauthorblockN{David Tse}%
\IEEEauthorblockA{%
dntse@stanford.edu%
}%
\thanks{The authors contributed equally and are listed alphabetically.
Contact: DT.%
}%
}%

\begin{onlyonarxiv}
\makeatletter
\def\ps@headings{%
\def\@oddhead{\mbox{}\scriptsize\rightmark \hfil \thepage}%
\def\@evenhead{\scriptsize\thepage \hfil \leftmark\mbox{}}}
\makeatother
\end{onlyonarxiv}

\begin{document}
\begin{onlyonarxiv}
\bstctlcite{IEEEexample:BSTcontrol_arxiv}
\end{onlyonarxiv}
\begin{onlyinproceedings}
\bstctlcite{IEEEexample:BSTcontrol_proceedings}
\end{onlyinproceedings}
\nocite{neu2020ebbandflow}
\maketitle

\begin{abstract}
The CAP theorem says that no blockchain can be live under dynamic participation and safe under temporary network partitions. To resolve this availability-finality dilemma, we formulate a new class of flexible consensus protocols, {\em \eaf protocols}, which support a full dynamically available ledger in conjunction with a finalized prefix ledger.
The finalized ledger falls behind the full ledger when the network partitions but catches up when the network heals.
Gasper, the current candidate protocol for Ethereum 2.0's beacon chain, combines the finality gadget Casper FFG with the LMD GHOST fork choice rule and aims to achieve this property. However, we discovered an attack in the standard synchronous network model,
highlighting a general difficulty with existing finality-gadget-based designs.
We present a construction of provably secure \eaf protocols with optimal resilience. Nodes run an off-the-shelf dynamically available protocol, take snapshots of the growing available ledger, and input them into a separate off-the-shelf BFT protocol to finalize a prefix.
We explore connections with flexible BFT and improve upon the state-of-the-art for that problem.
\end{abstract}

\section{Introduction}

\label{sec:intro}

\subsection{The Availability-Finality Dilemma}

Distributed consensus is a 40-year-old field. 
In its classical state machine replication formulation,
\emph{clients}
(\eg, merchants)
issue \emph{transactions}
(\eg, payments)
to be shared with \emph{nodes}
(\eg, the servers implementing a distributed payment system)
who communicate among each other via
an unreliable network
and seek to reach agreement on
a common \emph{ledger}%
{ }(\eg, sequence of payments)%
.
In the standard permissioned setting, the number of nodes is assumed to be known, fixed and each node is always awake, actively participating in the consensus protocol.  One important novelty blockchains have brought into this field is the notion of {\em dynamically available protocols}: consensus systems that can support an unknown number of nodes each of which can go to sleep and awake dynamically. Dynamic availability is a useful property of a consensus protocol, particularly in a large-scale setting with many nodes not all of which are active at the same time. Nakamoto's Proof-of-Work (PoW) longest chain protocol \cite{nakamoto_paper} is perhaps the first such dynamically available consensus protocol. The amount of mining power is varying in time and the system is live and safe as long as less than $50\%$ of the online hashrate belongs to adversary miners. The longest chain design was subsequently adapted to support dynamic availability in permissioned \cite{sleepy} and Proof-of-Stake (PoS) settings \cite{snowwhite,david2018ouroboros,badertscher2018ouroboros}.
Supporting dynamic availability is more challenging in these settings. Earlier works need to assume all adversary nodes are awake at the beginning \cite{sleepy,badertscher2018ouroboros} or a trusted setup for nodes to join the network \cite{snowwhite,david2018ouroboros}, but recently it has been shown that these restrictions can be removed using verifiable delay functions \cite{posat}.

One limitation of dynamically available protocols is that they are not tolerant to network partition: when the network partitions, honest nodes in a dynamically available protocol will think that many nodes are asleep, continue to confirm transactions, and thus is not safe.\footnote{In this paper, network partition can equally mean a catastrophic physical disconnection among the nodes, or perhaps a less rare situation where many adversary nodes are not communicating with the honest nodes but building a chain in private.} This is in contrast to permissioned BFT protocols designed for partially synchronous networks, such as PBFT \cite{pbft}, Tendermint \cite{tendermint_thesis,tendermint}, Hotstuff \cite{yin2018hotstuff} and Streamlet \cite{streamlet}.
This type of protocols is the basis for permissioned blockchains such as Libra \cite{libraBFT,baudet2018state} and PoS blockchains such as Algorand \cite{chen2016algorand,gilad2017algorand}.
In these protocols, a quorum of two-thirds of the signatures of all the nodes is required to finalize transactions, and hence is safe under network partition. On the other hand, these protocols are not live under dynamic availability: when many nodes are asleep, there is not enough of a quorum for the consensus protocol to proceed and it will get stalled. In fact, it is impossible for {\em any} protocol to be both safe under network partition and live under dynamic participation: individual nodes in the network cannot distinguish between the two scenarios to act differently. This intuition is formalized in \cite{sleepy} and its connection to the CAP theorem \cite{cap} was made precise recently in \cite{lewispye2020resource}.
 In light of this, protocol designers
see themselves faced with an availability-finality dilemma: whether to favor liveness under dynamic participation or safety under network partition. Hence, consensus protocols are typically classified as liveness-favoring or safety-favoring \cite{PS_partition}.

\subsection{\EAF Protocols}

For inspiration on a way to resolve this dilemma, let us revisit another important aspect of Nakamoto's longest chain protocol: the $k$-deep confirmation rule. In this protocol, all miners work on the longest chain, but different clients can choose different values of $k$ to determine how deep a block should be in the longest chain to confirm it. A client who chooses a larger value for $k$ is a more conservative client, believing in a more powerful attacker or wanting more reliability, and its ledger is a prefix of that of a more aggressive client which chooses a smaller value of $k$. Hence, in contrast to classic consensus protocols, Nakamoto's protocol supports {\em multiple} (nested) ledgers rather than only a single one. This concept of {\em flexible consensus} is
formalized and
further developed in \cite{flexibleBFT}, where different clients can make different assumptions about the synchronicity of the network as well as the power of the adversary.

The CAP theorem says no protocol can support clients that simultaneously want availability and finality. Inspired by the idea of flexible consensus, we can instead seek a flexible protocol that supports two types of clients: conservative clients who favor finality and want to be safe under network partition, and more aggressive clients who favor availability and want to be live under dynamic availability. A conservative client will only trust a {\em finalized} ledger, which is a prefix of a longer dynamically available ledger (or, \emph{available} ledger for short) believed by a more aggressive client. The finalized ledger falls behind the available ledger when network partitions, but catches up when the network heals.
This \emph{\eaf} property
avoids a system-wide determination of availability versus finality and instead leaves this decision to the
clients.

\subsection{Understanding Gasper}

Gasper \cite{buterin2020combining} is the current candidate protocol for Ethereum 2.0's beacon chain. The Gasper protocol is complex, combining the finality gadget Casper FFG \cite{buterin2017casper} with the LMD (Latest Message Driven) GHOST fork choice rule in a handcrafted way. One motivation for our work is to understand Gasper's design goals.  As far as we can gather, two of its main goals are:%
\begin{enumerate}
    \item Ability to finalize certain blocks in the blockchain \cite[p. 1]{buterin2020combining}. In addition to network partition tolerance, finalization also allows accountability through slashing of protocol violators.
    \item Support of a highly available distributed ledger which does not halt even when finality is not achieved
    \cite{ethresearch-liveness-requirement-eth2,ryan2020}, \cite[Section 8.7]{buterin2020combining}.
    Availability is a central feature of the existing global Ethereum blockchain.
\end{enumerate}

Although the sense in which Gasper aims to simultaneously achieve these two goals is not specified in \cite{buterin2020combining}, we do know from the CAP theorem that no protocol can finalize all blocks {\em and} be a highly available ledger at the same time. Thus, we believe that the \eaf property is a good formulation of Gasper's design goals. In this context, the role of the finality gadget is to finalize a prefix of the ledger and the role of LMD GHOST is to support availability.

In \cite{buterin2020combining}, Gasper's finalized ledger is shown to be safe. However, it is claimed to be live only under a non-standard
\emph{stochastic}
network delay model. Following the standards advocated by \cite{cachin2017blockchain} for
the design and analysis of
blockchain protocols, we analyzed Gasper under a standard security model, and found it to be insecure. In particular, we discovered a liveness attack on Gasper in the standard synchronous model where messages can be delayed arbitrarily by the adversary up to a known network delay bound. Moreover, because this liveness attack is a balancing attack causing the votes to split between two parallel chains, this attack also denies the safety of the available ledger even when there is no network partition.

\subsection{A Provably Secure Construction with Optimal Resilience}
\label{sec:abstract-eaf-construction}

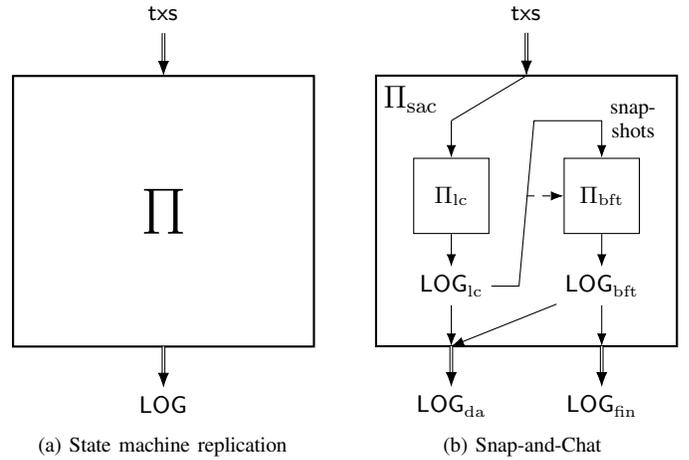
\begin{figure}%
    \centering%
    \vspace{-0.75em}%
    \subfloat[State machine replication\label{fig:protocol-overview-smr}]{%
        \begin{tikzpicture}[
                x=0.2cm,
                y=0.4cm,
            ]
            \small

            \node at (0,-0.5) [draw,minimum height=3.6cm,minimum width=4cm,thick] {\Huge $\Pi$};
            
            \coordinate (INTERFACE_txs_connector) at (0,4);
            \node (INTERFACE_txs) at (0,6) [align=center] {\Txs{}{}};
            
            \draw [-Latex,double] (INTERFACE_txs) -- (INTERFACE_txs_connector);
            
            \coordinate (INTERFACE_LOG_connector) at (0,-5);
            \node (INTERFACE_LOG) at (0,-7) {$\LOG{}{}$};
            
            \draw [-Latex,double] (INTERFACE_LOG_connector) -- (INTERFACE_LOG);

        \end{tikzpicture}%
    }%
    \hfill%
    \subfloat[\SAC\label{fig:protocol-overview-eaf}]{%
        \begin{tikzpicture}[
                x=0.2cm,
                y=0.4cm,
            ]
            \small
            
            \node (box_lc) at (-5,0) [draw,minimum width=1cm,minimum height=1cm,align=center] {$\PIlc$};
            \node (box_bft) at (+5,0) [draw,minimum width=1cm,minimum height=1cm,align=center] {$\PIbft$};
            
            \node (chain_lc) at (-5,-3) [align=center] {$\LOGlc{}{}$};
            \node (chain_bft) at (5,-3) [align=center] {$\LOGbft{}{}$};

            \coordinate (txs) at (-5,2.5);

            \coordinate (ledger_da) at (-5,-5);
            \coordinate (ledger_fin) at (5,-5);
            
            \draw[-Latex] (txs) -- (box_lc);
            \draw[-Latex] (box_lc) -- (chain_lc);
            \draw[-Latex] (box_bft) -- (chain_bft);
            
            \draw[-Latex] (chain_lc) -- (-0.5,-3) -- (0.5,2.5) -- (5,2.5) -- (box_bft)
                node [pos=0,right,align=center,font=\footnotesize] {snap-\\shots};
            
            \draw[-Latex,dashed] (0,0) -- (box_bft);

            \draw[-Latex] (chain_lc) -- (ledger_da);
            \draw[-Latex] (chain_bft) -- (ledger_da);
            \draw[-Latex] (chain_bft) -- (ledger_fin);

            \node (outerbox) at (0,-0.5) [draw,minimum height=3.6cm,minimum width=4cm,thick] {};
            
            \coordinate (INTERFACE_txs_connector) at (0,4);
            \node (INTERFACE_txs) at (0,6) [align=center] {$\Txs{}{}$};
            
            \draw [-Latex,double] (INTERFACE_txs) -- (INTERFACE_txs_connector);
            \draw (INTERFACE_txs_connector) -- (txs);
            
            \coordinate (INTERFACE_LOGda_connector) at (-5,-5);
            \coordinate (INTERFACE_LOGfin_connector) at (5,-5);
            \node (INTERFACE_LOGda) at (-5,-7) {$\LOGda{}{}$};
            \node (INTERFACE_LOGfin) at (5,-7) {$\LOGfin{}{}$};
            
            \draw [-Latex,double] (INTERFACE_LOGda_connector) -- (INTERFACE_LOGda);
            \draw [-Latex,double] (INTERFACE_LOGfin_connector) -- (INTERFACE_LOGfin);

            \node at (outerbox.north west) [anchor=north west,yshift=-1pt] {\large $\PIanyexample$};

        \end{tikzpicture}%
    }%
    \caption[]{%
        \subref{fig:protocol-overview-smr} A consensus protocol $\Pi$ implementing state machine replication receives transactions $\Txs{}{}$ as inputs from the environment and outputs an ever-increasing ordered ledger of transactions $\LOG{}{}$.
        \subref{fig:protocol-overview-eaf} A \sac protocol produced by \theconstruction, $\PIanyexample$, receives transactions $\Txs{}{}$ from the environment and outputs two ever-increasing ledgers $\LOGda{}{}$ and $\LOGfin{}{}$ by running a dynamically available protocol $\PIlc$ and a partially synchronous protocol $\PIbft$ in parallel. The inputs to $\PIlc$ are environment's transactions but the inputs to $\PIbft{}{}$ are snapshots of the output ledger of $\PIlc$ from the nodes' views. 
        The dashed line signifies that
        nodes use the output of $\PIlc$
        as side information in $\PIbft$
        to boycott the finalization of invalid snapshots.
    }
    \label{fig:protocol-overview-abstract} 
\end{figure}

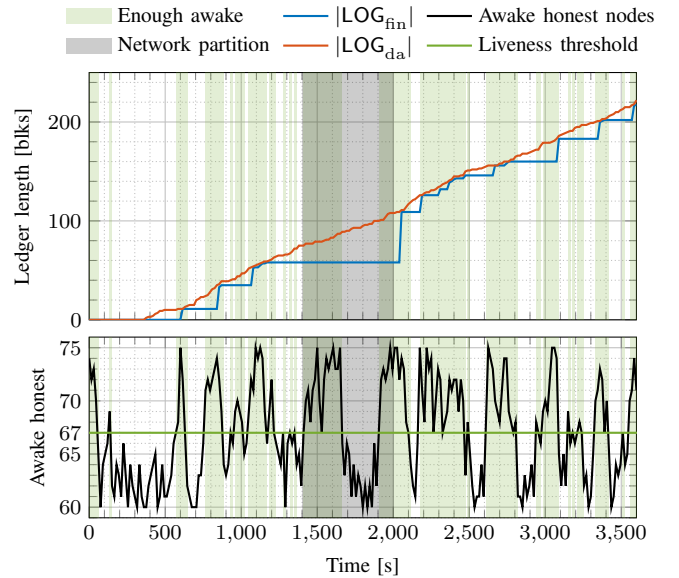
\begin{figure}
    \centering%
    \begin{tikzpicture}%
        \footnotesize
        \begin{axis}[
            mysimpleplot,
            name=plot1,
            ylabel={Ledger length [blks]},
            legend columns=2,
            xmin=0, xmax=3600,
            ymin=0, ymax=250,
            height=0.55\linewidth,
            xmajorticks=false,
            transpose legend,
        ]
        
            \def\DATAPREFIX{./figures/simulation-04-intro}

            \foreach \tStart/\tMidway/\tStop in {0/30.0/60, 135/142.5/150, 570/607.5/645, 765/825.0/885, 930/937.5/945, 960/990.0/1020, 1050/1110.0/1170, 1185/1207.5/1230, 1275/1282.5/1290, 1320/1327.5/1335, 1350/1357.5/1365, 1410/1537.5/1665, 1905/2010.0/2115, 2175/2325.0/2475, 2490/2497.5/2505, 2610/2715.0/2820, 2940/2955.0/2970, 2985/3037.5/3090, 3150/3157.5/3165, 3180/3187.5/3195, 3210/3232.5/3255, 3330/3375.0/3420, 3510/3517.5/3525, 3555/3577.5/3600} {
                \edef\temp{\noexpand\draw [fill=myParula05Green,fill opacity=0.2,draw=none] (axis cs:\tStart,-100) rectangle (axis cs:\tStop,1000);}
                \temp
            }
            
            \addlegendimage{area legend,draw opacity=0,fill=myParula05Green,fill opacity=0.2,draw=none};
            \addlegendentry{Enough awake};

            \draw [fill=black,fill opacity=0.2,draw=none] (axis cs:1400,-100) rectangle (axis cs:2000,1000);

            \addlegendimage{area legend,draw opacity=0,fill=black,fill opacity=0.2,draw=none};
            \addlegendentry{Network partition};

            \addplot [myparula11,mark=none] table [x=t,y=l_Lp]
            {\DATAPREFIX/sim-04-phases.dat};
            \label{leg:simulations-intro-Lp}
            \addlegendentry{$|\LOGfin{}{}|$};
            
            \addplot [myparula21,mark=none] table [x=t,y=l_Lda]
            {\DATAPREFIX/sim-04-phases.dat};
            \label{leg:simulations-intro-Lda}
            \addlegendentry{$|\LOGda{}{}|$};

            \addlegendimage{black,mark=none}
            \addlegendentry{Awake honest nodes};
            
            \addlegendimage{myParula05Green,no markers}
            \addlegendentry{Liveness threshold};

        \end{axis}
        \begin{axis}[
            mysimpleplot,
            name=plot2,
            at=(plot1.below south), anchor=above north,
            xlabel={Time [s]},
            ylabel={Awake honest},
            xmin=0, xmax=3600,
            ymin=59, ymax=76,
            height=0.45\linewidth,
            extra y ticks={67},
        ]
        
            \def\DATAPREFIX{./figures/simulation-04-intro}

            \foreach \tStart/\tMidway/\tStop in {0/30.0/60, 135/142.5/150, 570/607.5/645, 765/825.0/885, 930/937.5/945, 960/990.0/1020, 1050/1110.0/1170, 1185/1207.5/1230, 1275/1282.5/1290, 1320/1327.5/1335, 1350/1357.5/1365, 1410/1537.5/1665, 1905/2010.0/2115, 2175/2325.0/2475, 2490/2497.5/2505, 2610/2715.0/2820, 2940/2955.0/2970, 2985/3037.5/3090, 3150/3157.5/3165, 3180/3187.5/3195, 3210/3232.5/3255, 3330/3375.0/3420, 3510/3517.5/3525, 3555/3577.5/3600} {
                \edef\temp{\noexpand\draw [fill=myParula05Green,fill opacity=0.2,draw=none] (axis cs:\tStart,-100) rectangle (axis cs:\tStop,1000);}
                \temp
            }

            \draw [fill=black,fill opacity=0.2,draw=none] (axis cs:1400,-100) rectangle (axis cs:2000,1000);

            \addplot [name path=dacurve,black,mark=none] table [x=t,y=l_awake]
            {\DATAPREFIX/sim-04-phases.dat};
            \label{leg:simulations-intro-awake}
            
            \addplot[name path=threshold,myParula05Green,no markers,domain=-1000:4600] {67};
            \label{leg:simulations-intro-threshold}

        \end{axis}
    \end{tikzpicture}%
    \vspace{-0.75em}%
    \caption[]{A simulated run of an example \sac protocol (combining longest chain  and Streamlet \cite{streamlet}) under dynamic participation and network partition. The lengths of the two ledgers are plotted over time.  During network partition or when few nodes are awake, the finalized ledger falls behind the available ledger, but catches up after the network heals or when a sufficient number of nodes wake up. See Section~\ref{sec:simulations} for details on the simulation setup.}
    \label{fig:intro_sim}
\end{figure}

In this work, we make two contributions. First we define what an \eaf protocol is and its desired security property. While the goals of an \eaf protocol have been informally discussed to motivate finality-gadget-based designs such as Gasper and a few others (\eg, \cite{stewart2020grandpa}), to the best of our knowledge these informal goals have not been translated into a mathematically defined security property.
Second, we provide a construction of a class of protocols, which we call {\em \sac} protocols, that provably satisfies the \eaf security property with optimal resilience.
In contrast to Gasper's handcrafted design, the \sac construction uses an off-the-shelf dynamically available protocol\footnote{Longest chain protocols are representative members of this class of protocols, hence the notation $\PIlc$, but this class includes many other protocols as well.} $\PIlc$ and an off-the-shelf partially synchronous BFT protocol $\PIbft$ (Figure \ref{fig:protocol-overview-abstract}). Nodes execute the protocol by executing the two sub-protocols in parallel. The $\PIlc$ sub-protocol takes as inputs transactions $\Txs{}{}$ from the environment and outputs an ever-increasing ledger $\LOGlc{}{}$. Over time, each node takes {\em snapshots} of this ledger based on its own current view, and input these snapshots into the second sub-protocol $\PIbft{}{}$ to finalize some of the transactions.
The output ledger $\LOGbft{}{}$ of $\PIbft$ is an ordered list of such snapshots. To create the finalized ledger $\LOGfin{}{}$ of transactions, $\LOGbft{}{}$ is flattened (\ie, all snapshots
included in $\LOGbft{}{}$ are concatenated) and sanitized so that only the first appearance of a transaction remains. Finally, $\LOGfin{}{}$ is prepended to $\LOGlc{}{}$ and sanitized to form the available ledger $\LOGda{}{}$.  A simulated run of an example \sac protocol is shown in Figure \ref{fig:intro_sim}.

Even though honest nodes following a \sac protocol input snapshots of the (confirmed) ledger $\LOGlc{}{}$ into $\PIbft$, an adversary could, in an attempt to break safety, input an ostensible ledger snapshot which really contains unconfirmed transactions. This motivates the last ingredient of our construction: in the $\PIbft$ sub-protocol, each honest node
boycotts the finalization of snapshots that are not confirmed in $\PIlc$ in its view.
An off-the-shelf BFT protocol needs to be modified to implement this constraint. We show that fortunately the required modification is minor in several example protocols, including PBFT \cite{pbft}, Hotstuff \cite{yin2018hotstuff} and Streamlet \cite{streamlet}. 
When any of these slightly modified BFT protocols is used in conjunction with a permissioned longest chain protocol \cite{sleepy,snowwhite,david2018ouroboros}, we prove a formal security property for the resulting \sac protocol,
which is our definition of the desired goal of an \eaf protocol.

\begin{theorem*}[Informal]
Consider
a network environment where:
\begin{enumerate}
    \item Communication is asynchronous until a  global stabilization time $\GST$ after which communication becomes synchronous, and 
    \item honest nodes sleep and wake up until a global awake time $\GOT$ after which all nodes are awake. 
    Adversary nodes are always awake.
\end{enumerate}
Then
\begin{enumerate}
    \item ({\bf P1 - Finality}): The finalized ledger $\LOGfin{}{}$ is guaranteed to be safe at all times, and live after $\max\{\GST,\GOT\}$, provided that fewer than $33\%$ of all the nodes are adversarial.
    \item ({\bf P2 - Dynamic Availability}): If $\GST = 0$, the available ledger $\LOGda{}{}$ is guaranteed to be safe and live at all times, provided that at all times fewer than $50\%$ of the awake nodes are adversarial.
\end{enumerate}
\end{theorem*}
Note that the assumptions on the adversary are different for the security of the two ledgers, in line with the spirit of a flexible protocol \cite{flexibleBFT}. Together, \Pone and \Ptwo say that the finalized ledger $\LOGfin{}{}$ is safe under network partition, \ie, before $\max\{\GST,\GOT\}$, and afterwards catches up with the available ledger $\LOGda{}{}$, which is always live and safe provided that the majority of awake nodes is honest.

If $\GOT = 0$, then the environment is the classical partially synchronous network, and the ledger $\LOGfin{}{}$ has the optimal resilience achievable in that environment. On the other hand, if $\GST = 0$ and $\GOT = \infty$, then the environment is a synchronous network with dynamic participation, and the ledger $\LOGda{}{}$
has the optimal resilience achievable in that environment. Thus, our construction achieves consistency between the two ledgers without sacrificing the  best possible security guarantees of the individual ledgers. In that sense, our construction achieves the \eaf property in an optimal manner.

\subsection{Flexible BFT Revisited}

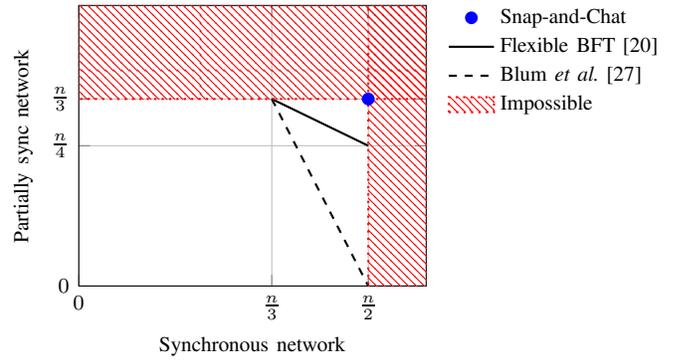
\begin{figure}
    \centering
    \begin{tikzpicture}
        \begin{axis}[
            mysimpleplot,
            width=0.7\linewidth,
            height=0.6\linewidth,
            legend columns=1,
            legend style={
                at={(1,1)},
                xshift=0.5em,
                anchor=north west,
                draw=none,
                /tikz/every even column/.append style={
                    column sep=0.3em
                },
                cells={
                    align=left
                },
            },
            xlabel={Synchronous network},
            ylabel={Partially sync network},
        	xmin=0, xmax=0.6,
        	ymin=0, ymax=0.5,
            xtick={0,0.33333,0.50000},
            xticklabels={$0$,$\frac{n}{3}$,$\frac{n}{2}$},
            ytick={0,0.25000,0.33333},
            yticklabels={$0$,$\frac{n}{4}$,$\frac{n}{3}$},
        ]
            \addplot [blue,only marks,mark=*] coordinates { (0.50000,0.33333) };
            \addlegendentry{\SAC};
            
            \addplot [] coordinates { (0.33333,0.33333) (0.50000,0.25000) };
            \addlegendentry{Flexible BFT \cite{flexibleBFT}};
            
            \addplot [dashed] coordinates { (0.33333,0.33333) (0.50000,0.0) };
            \addlegendentry{Blum \etal \cite{katzBA}};

            \draw[red,dotted,thick,pattern=north west lines, pattern color=red] (axis cs:-0.1,0.3333) rectangle (axis cs:1,1);
            \draw[red,dotted,thick,pattern=north west lines, pattern color=red] (axis cs:0.5,-0.1) rectangle (axis cs:1,1);
            
            \addlegendimage{area legend,red,dotted,thick,pattern=north west lines, pattern color=red}
            \addlegendentry{Impossible}; %
            
        \end{axis}
    \end{tikzpicture}%
    \vspace{-0.75em}%
    \caption{The flexible BFT protocol can simultaneously support clients who can tolerate $f$ adversaries in a synchronous environment and clients who can tolerate $(n-f)/2$ adversaries in a  partially synchronous environment, for any $f$ between $n/3$ and $n/2$. Thus, there is a tradeoff between the two guarantees. The \sac protocol achieves $(n/2, n/3)$, simultaneously optimal. No tradeoff is necessary.} 
    \label{fig:flexible_BFT}
\end{figure}

\Pone and \Ptwo together with prefix consistency
provide flexible consensus.
Our mathematical formulation of the \eaf property can be viewed as going beyond that of Flexible BFT \cite{flexibleBFT} in two ways. 
First, \cite{flexibleBFT} focuses on synchronicity assumptions and we bring dynamic participation as a new client belief into the story.
Second, the formulation in \cite{flexibleBFT} requires consistency between ledgers of two clients only when their assumptions are both correct, but we require prefix consistency between the ledgers in {\em all} circumstances.
In that sense, the flexibility our formulation offers is closer in nature to the flexibility offered by Nakamoto's longest chain protocol.
Prefix consistency under all circumstances is crucial,
\eg, for cryptocurrencies,
where eventually all clients, no matter their beliefs,
should converge on a unique ledger,
a single version of history to settle disputes regarding
`who owns what'.

But even for the formulation considered in \cite{flexibleBFT}, our construction provides a different solution and offers stronger security guarantees than the white-box construction in \cite{flexibleBFT}. More specifically, the flexible BFT protocol in \cite{flexibleBFT} can simultaneously support clients who can tolerate $n/2$ adversaries in a synchronous environment and clients who can tolerate a fraction of $n/4$ adversaries in a  partially synchronous environment. Since a synchronous environment is a special case of the dynamic participation environment (by setting $\GOT = 0$), our construction improves the security guarantees to simultaneously support clients who can tolerate $n/2$ adversaries in a synchronous environment and clients who can tolerate $n/3$ adversaries in a partially synchronous environment. Consistent with the optimality of our construction, these guarantees cannot be improved further (see Figure \ref{fig:flexible_BFT}).

It is also insightful to compare our results with those of \cite{katzBA}, which designed a randomized Byzantine agreement protocol  secure under both a synchronous and an asynchronous environment. The dashed line in Figure \ref{fig:flexible_BFT} shows the tradeoff between the resiliences the protocol can support in the two environments, and this tradeoff is proved to be optimal. Note that this protocol is {\em not} a flexible protocol, since a single value has to be agreed upon regardless of which of the two environments one is in. Thus, the gap between the resilience achieved by the \sac protocol and the protocol in \cite{katzBA} can be interpreted as the {\em value of flexibility}. Interestingly, the protocol in \cite{katzBA} is also constructed by the composition of two sub-protocols, but in contrast to the construction of \sac protocols, the two sub-protocols are not off-the-shelf, but are constructed tailored to the problem at hand.

\subsection{Outline}

The remainder of this manuscript is structured as follows.
First, we present a balancing attack on Gasper in Section~\ref{sec:attacks-gasper},
demonstrating that Gasper is not secure.
Section~\ref{sec:optimal} formulates the \eaf security property, 
describes the construction of \sac protocols in detail
and proves that they satisfy the \eaf security property with optimal resilience.
We show the results of simulation experiments
providing an insight into the behavior of \sac protocols
in
Section~\ref{sec:simulations}.
In Section~\ref{sec:finality-gadgets},
we compare the design of \sac protocols and
finality gadgets.
We conclude the paper
with
how to transfer our results to the PoW setting in Section~\ref{sec:discussion-pow}
and
an overview of features beyond security
provided out-of-the-box by \sac protocols
in Section~\ref{sec:sac-for-eth2}.
\section{A Balancing Attack on Gasper}
\label{sec:attack}

\newcommand{\CasperFFG}[0]{Casper FFG\xspace}

\label{sec:attacks-gasper}

Gasper \cite{buterin2020combining} is the current proposal for Ethereum 2.0's beacon chain.
In the following, we exhibit a liveness attack against
Gasper in the synchronous network model.\footnote{%
Source code of a simulation of the attack
(discussed in Appendix~\ref{sec:appendix-attack-gasper-simulation})
can be found at: \url{https://github.com/tse-group/gasper-attack}.%
}
What is more, the attack leads to loss of safety for the
underlying dynamically available ledger.
Thus, Gasper is not secure in the synchronous network model
and does not provide a resolution to the availability-finality dilemma.

Our attack uses that under synchrony,
network delay is \emph{adversarial}
(rather than merely \emph{stochastic},
as was analyzed in \cite{buterin2020combining}).
Considering, \eg,
state-sponsored adversaries
or malicious network providers,
at least some degree of adversarial network delay
cannot be ruled out.
Furthermore, the synchrony model with adversarial delay is a well-established baseline model
for which many secure protocols are known.

Gasper is a vote-based PoS protocol which combines \CasperFFG \cite{buterin2017casper} with
a committee-based blockchain block proposal mechanism where the fork
(\ie, the tip of the chain to propose new blocks on or vote for)
is chosen
using the `greedy heaviest observed sub-tree' (GHOST) rule
under the
`latest message driven' (LMD) paradigm,
\ie,
taking
into consideration only the most recent vote per validator.
A Gasper vote consists of two parts, a GHOST vote and a \CasperFFG vote.
While details of Gasper preclude the vanilla bouncing attack \cite{ethresearch-bouncing-attack,ethresearch-bouncing-attack-analysis,ethresearch-bouncing-attack-prevention}
on the \CasperFFG layer,
Gasper is vulnerable to a similar balancing attack on the GHOST layer.
\input{attack_gasper_figure_overview_tiny}

Recall that
Gasper proceeds in epochs which are further
subdivided into $C$ slots each.
For simplicity, let
$C$ divide $n$ so that
every slot has a \emph{committee} of
size $n/C$.
For each epoch, a random permutation of all $n$ validators
assigns
validators to slots' committees
and designates a 
\emph{proposer} per slot.
Per slot,
the proposer
produces a new block extending the tip determined by
the fork choice rule $\mathsf{HLMD}(G)$ executed in
local
view $G$
(see \cite[Algorithm 4.2]{buterin2020combining}).
Then, each validator of the slot's committee
decides what block to vote for
using $\mathsf{HLMD}(G)$
in local
view $G$.
For the \CasperFFG layer,
a block can only become finalized if two-thirds of
validators
vote for it.
The attacker aims
to keep honest validators split between two options
(`left' and `right' chain, see Figure~\ref{fig:attack-gasper-overview})
indefinitely,
so that neither option ever gets two-thirds votes
and thus no block ever gets finalized.
Key technique to maintain this split
is that some adversarial validators
(`swayers' in Figure~\ref{fig:attack-gasper-overview})
withhold their votes
and release them only at specific times
and to specific subsets of honest nodes
in order to influence the fork choice of
honest nodes and thus steer which honest nodes
vote `left'/`right'.

The basic idea of the attack is as follows
(for a detailed description, see Appendix~\ref{sec:appendix-attack-gasper}
and \cite{ethresearch-gasper-liveness-attack}%
).
The adversary waits for an opportune epoch to kick-start the attack.
An epoch is opportune if the proposer in the first slot is adversarial,
and
in every slot of the epoch
there are
enough
(six suffice; explained in detail in Appendix~\ref{sec:appendix-attack-gasper-attack})
adversarial validators
to fulfill certain tasks
in the attack (see \textcircled{\footnotesize a}--\textcircled{\footnotesize d} in Figure~\ref{fig:attack-gasper-overview}).
In particular in the regime of many validators ($n \to \infty$),
the probability that a particular epoch is opportune is roughly $f/n$
(see Appendix~\ref{sec:appendix-attack-gasper-analysis-and-simulation}). %
Note that for $n$ large,
any positive fraction $f/n$
of adversarial nodes
suffices to mount the attack,
with the first opportune epoch
occurring after $n/f$ epochs on average.
For ease of exposition, let epoch $0$ be opportune.

The adversarial proposer of slot $0$ equivocates and produces two conflicting blocks
(`left' and `right')
which it reveals to two suitably chosen equal-sized subsets of the committee. %
One subset votes `left', the other subset votes `right' -- a tie.
The adversary selectively releases withheld votes from slot $0$ to split validators
of slot $1$ into two equal-sized groups, one which sees `left' as leading and votes for it, and
one which sees `right' as leading and votes for it -- still a tie.
The adversary continues this strategy
to maintain the tie throughout epoch $0$.

During epoch $1$, the adversary
selectively releases additional withheld votes from epoch $0$ to keep splitting validators
into two groups, one of which sees `left' as leading and votes `left',
the other sees `right' as leading and votes `right'.
Note that these groups now do not have to be equal in size.
It suffices for the adversary
to release withheld votes selectively so as to reaffirm
honest validators in their illusion that
whatever chain they previously voted for
happens to still be leading, so that they renew their vote.
Due to the LMD paradigm of Gasper's fork choice rule,
only the most recent vote per validator counts
and thus the effective vote tally remains unchanged.
At the end of epoch $1$ there are still two chains
with equally many votes and thus neither gets finalized.

For epoch $2$ and beyond the adversary
repeats its actions of epoch $1$.
Note that the validators whose withheld epoch $0$
votes the adversary
used to sway honest validators in epoch $1$
have themselves not voted in epoch $1$ yet.
Thus, during epoch $2$ the adversary
selectively releases
votes from epoch $1$ to
maintain the tie between the two chains.
This continues indefinitely.

Thus, Gasper is not live in
the synchronous model.
Furthermore,
the
block proposal mechanism is rendered unsafe
by the modified fork choice rule as
the chosen fork flip-flops between `left' and `right'.
Since Gasper does not satisfy
the desired ebb-and-flow
security property,
we next introduce a provably secure family of
\eaf protocols.
\section{Optimal \EAF Protocols}
\label{sec:optimal}

In this section, we formulate precisely the \eaf security property, present the construction of \sac protocols, and show that \sac protocols achieve the
\eaf
property with optimal resilience.
For the construction, we build state machine replication protocols $\PIanyexample$
(\sac protocols)
by composing a dynamically available longest-chain protocol \cite{sleepy,snowwhite, kiayias2017ouroboros,david2018ouroboros} as $\PIlc$ with a partially synchronous BFT protocol \cite{pbft,yin2018hotstuff,streamlet} as $\PIbft$.
The focus of this paper is on the permissioned setting. The resulting permissioned protocol can be viewed as a core around which a full PoS protocol can be built, much like Sleepy \cite{sleepy} is the permissioned core of the PoS protocol SnowWhite \cite{snowwhite}. To build a full PoS protocol, issues such as stake grinding \cite{snowwhite,badertscher2018ouroboros} have to be considered.
\Sac protocols can also be used in a hybrid PoS-PoW setting, where validators run the BFT sub-protocol and miners power the dynamically available sub-protocol (see Section~\ref{sec:discussion-pow}).
These are topics for future work.

\subsection{Model and Formulation}
\label{sec:model}

The execution model of $\PIanyexample$ inherits the cryptographic assumptions and primitives used in \cite{sleepy,streamlet,yin2018hotstuff}.
The cornerstones of the model are:
\begin{itemize}
    \item There are in total $n$ nodes numbered from $1$ thru $n$.
    \item Time proceeds in slots. %
        Nodes have synchronized clocks.\footnote{%
        Bounded clock offsets can be captured as part of the network delay.}
    \item There
    is
    a public-key infrastructure and each node is equipped with a unique cryptographic identity. %
    \item There
    is
    a
    random oracle,
    which serves as the source of randomness in \theconstruction.
    \item The adversary is a probabilistic poly-time algorithm.
\end{itemize}

\emph{Corruption:}
Before the protocol execution starts, the adversary gets to corrupt (up to) $f$ nodes, then called \emph{adversarial}. Adversarial nodes surrender their internal state to the adversary and can deviate from the protocol arbitrarily (Byzantine faults) under the adversary's control.
The remaining $(n-f)$ nodes are \emph{honest} and
follow the protocol as specified.

\emph{Networking:}
Nodes can send each other messages which arrive with a certain delay controlled by the adversary, subject to constraints elaborated below.

\emph{Sleeping:}
The adversary chooses,
for every time slot and honest node,
whether the node is \emph{awake} or \emph{asleep} in that slot, subject to constraints elaborated below.
An honest node that is awake in a slot executes the protocol faithfully in that slot.
An honest node that is asleep in a slot does not execute the protocol in that slot, and messages that would have arrived in that slot are queued and delivered in the first slot in which the node is awake again.
Adversarial nodes are always awake.
Using the features above,
dynamic participation
in the permissioned setting can be modelled,
where
all nodes' cryptographic identities are common knowledge
but
honest nodes do not know which nodes are awake or asleep at any given time.
Thus, the permissioned nature and dynamic participation represent two orthogonal aspects of the environment.
As building blocks for the environment
adopted for \eaf protocols,
recall that in a traditional \emph{synchronous network},
messages sent by honest nodes arrive within a known finite delay bound.
In a \emph{partially synchronous network} \cite{model-psync},
initially,
messages can be delayed arbitrarily.
After some time, the network turns synchronous.
Thus, partial synchrony models a network
with a period of partition followed by synchrony.
Although in reality,
multiple such periods of (a-)synchrony could alternate,
we follow the long-standing practice in the BFT literature
and study only a single such transition.
Now, recall the informal Theorem of Section~\ref{sec:abstract-eaf-construction}.
The theorem provides two sets of security guarantees, labelled as \Pone and \Ptwo, for the finalized and available ledgers.
These guarantees are stated under two sets of assumptions on the environment $\Env$ and the adversary $\Adv$.
The assumptions model a partially synchronous network and a synchronous network with dynamic participation, respectively. 

$\AdvEnvParameterized{1}{\beta}$ formalizes the model of \textbf{P1},
a partially synchronous network under dynamic participation,
with respect to
the fraction $\beta$ of adversary nodes:
\begin{itemize}
    \item 
        $\Adv_1$ corrupts $f = \beta n$ nodes.
        
    \item
        Before a global stabilization time $\GST$,
        $\Adv_1$ can delay network messages arbitrarily.
        After $\GST$,
        $\Adv_1$ is required to deliver all messages sent between honest nodes in at most $\Delta$ slots.
        $\GST$ is chosen by $\Adv_1$, unknown to the honest nodes, and can be a causal function of the randomness in the protocol.
        
    \item
        Before a global awake time
        $\GOT$,
        $\Adv_1$ determines which honest nodes are awake/asleep and when.
        After $\GOT$,
        all honest nodes are awake.\footnote{%
        Without slightly restricting
        dynamic participation via a $\GOT$
        after which all nodes are awake,
        this adversary would fall under the CAP theorem
        so that no secure protocol against it can exist.
        In many applications
        it is realistic that every now and then there is a period in which all nodes are awake.%
        }
        $\GOT$ is chosen by $\Adv_1$, unknown to the honest nodes and can be a causal function of the randomness in the protocol.
\end{itemize}

$\AdvEnvParameterized{2}{\beta}$ formalizes the model of
\textbf{P2},
a synchronous network under dynamic participation,
with respect to
a bound $\beta$ on the fraction of awake nodes that are adversarial:
\begin{itemize}
    \item
        At all times,
        $\Adv_2$ is required to deliver all messages sent between honest nodes in at most $\Delta$ slots.

    \item
        At all times,
        $\Adv_2$ determines which honest nodes are awake/asleep and when,
        subject to the constraint that
        at all times
        at most fraction $\beta$ of awake nodes are adversarial
        and at least one honest node is awake.
\end{itemize}
We next formalize the notion of safety, liveness and security \emph{after a certain time}.
For this purpose, we adopt and modify the security definition given in \cite{sleepy}.
This definition has a security parameter $\sigma$
which in the context of longest-chain protocols represents the confirmation delay for transactions.
In our analysis,
we consider a finite time horizon of size polynomial in $\sigma$.
Note that in the definition below, $\LOG{i}{t}$ denotes the ledger $\LOG{}{}$ in view of node $i$ at time $t$.

\begin{definition}
\label{def:security-after}
Let $\Tconfirm$ be a polynomial function of the security parameter $\sigma$. 
We say that a state machine replication protocol $\PI$ outputting a ledger $\LOG{}{}$ is \emph{secure after time} $T$ and has transaction confirmation time $\Tconfirm$
if $\LOG{}{}$ satisfies:
\begin{itemize}
    \item \textbf{Safety:} For any two times $t \geq t' \geq T$, and any two honest nodes $i$ and $j$ awake at times $t$ and $t'$ respectively, either $\LOG{i}{t} \preceq \LOG{j}{t'}$ or $\LOG{j}{t'} \preceq \LOG{i}{t}$.
    \item \textbf{Liveness:} If a transaction is received by an awake honest node at some time $t \geq T$, then, for any time $t' \geq t+\Tconfirm$ and honest node $j$ that is awake at time $t'$, the transaction will be included in $\LOG{j}{t'}$.
\end{itemize}
\end{definition}

Definition \ref{def:security-after} formalizes the meaning of `safety, liveness and security \emph{after a certain time $T$}'.
In general, there it might be two different times after which a protocol is safe (live).
A protocol that is safe (live) at all times (\ie, after $T=0$) is simply called \emph{safe} (\emph{live}) without further qualification.
With a slight abuse of notation, we also call a ledger $\LOG{}{}$ secure/safe/live to mean that the protocol $\PI$ outputting the ledger $\LOG{}{}$ is secure/safe/live, respectively.
Now we are ready to define an \eaf protocol and its
notion of security. First we define formally a flexible protocol.

\begin{definition}
A {\em flexible} protocol is a pair of state machine replication protocols $(\PI_1, \PI_2)$, where $\PI_1$ and $\PI_2$ have the same input transactions $\Txs{}{}$ and output ledgers $\LOG{1}{}$ and $\LOG{2}{}$, respectively.
\end{definition}

\begin{definition}
\label{def:ebb-and-flow}
An \emph{$(\beta_1,\beta_2)$-secure \eaf protocol} $\PI$ is a flexible protocol $(\PIda,\PIfin)$ which outputs an available ledger $\LOGda{}{}$ and a finalized ledger $\LOGfin{}{}$, such that
for security parameter $\sigma \triangleq \Tconfirm$:
\begin{enumerate}
    \item \textbf{P1 - Finality:} Under $\AdvEnvParameterized{1}{\beta_1}$, $\LOGfin{}{}$ is safe at all times,
    and there exists a constant $\const$ such that $\LOGfin{}{}$ is live after time $\const (\max\{\GST, \GOT\}+\sigma)$ except with probability $\negl(\sigma)$.
    \item \textbf{P2 - Dynamic Availability:} Under $\AdvEnvParameterized{2}{\beta_2}$, $\LOGda{}{}$ is secure  except with probability $\negl(\sigma)$.
    \item \textbf{Prefix:} For any honest node $i$ and time $t$, $\LOGfin{i}{t}$ is a prefix of $\LOGda{i}{t}$.
\end{enumerate}
\end{definition}
In the above definition, the negligible function %
$\negl(\cdot)$
decays faster than all polynomials, \ie,
$\forall c>0: \exists \sigma_0: \forall \sigma > \sigma_0: \negl(\sigma) < \sigma^{-c}$.
Designing a state machine replication protocol $\PIfin$ that satisfies property \Pone is the well-studied problem of designing partially synchronous BFT protocols; the optimal resilience that can be achieved is $\beta_1 = \frac{1}{3}$. Designing a state machine replication protocol $\PIda$ that satisfies property \Ptwo is the problem of designing dynamically available protocols; the optimal resilience that can be achieved is $\beta_2 = \frac{1}{2}$. An \eaf protocol $(\PIda, \PIfin)$ has a further requirement that $\LOGfin{}{}$ should be a prefix of $\LOGda{}{}$; this requires a careful joint design of $(\PIda,\PIfin)$. We now present a construction for which we show that $\beta_1= \frac{1}{3}$ and $\beta_2 = \frac{1}{2}$ can be simultaneously achieved while respecting the prefix constraint.

\subsection{Protocol}
\label{sec:sleepy-streamlet}

\begin{figure*}
    \centering
    \begin{tikzpicture}[
            x=0.75cm,
            y=0.5cm,
        ]
        \small
        
        \tikzset{blockchain/.style={
                x=0.3cm,
                y=0.6cm,
                caption/.style = {
                    minimum width=0,
                    minimum height=0,
                },
                block/.style = {
                    minimum width=0.25cm,
                    minimum height=0.25cm,
                    draw,
                    shade,
                    top color=white,
                    bottom color=black!10,
                },
                link/.style = {
                },
                block-mining/.style = {
                    block,
                    draw=myParula01Blue,
                    text=myParula01Blue,
                    bottom color=myParula01Blue!10,
                    inner sep=0,
                },
                link-mining/.style = {
                    link,
                    draw=myParula01Blue,
                    densely dotted,
                },
                block-notarized/.style = {
                    block,
                    draw=myParula05Green!50!black,
                    bottom color=myParula05Green!50,
                },
                block-propose/.style = {
                    block,
                    draw=myParula01Blue,
                    text=myParula01Blue,
                    bottom color=myParula01Blue!10,
                    inner sep=0,
                },
                link-propose/.style = {
                    link,
                    draw=myParula01Blue,
                    densely dotted,
                },
                accepted/.style = {
                    draw=red,
                    ultra thick,
                },
                vote/.style = {
                    -latex,
                    draw=myParula01Blue,
                },
            }
        }
        
        \node (box_lc) at (-5,0) [draw,minimum width=6cm,minimum height=4cm,align=center] {%
            $\PIlc$: Permissioned-Longest-Chain \\[1em]
            \begin{tikzpicture}[blockchain,baseline=(baseline)]
                \coordinate (baseline) at (0,0.5);
                
                \node (caption) at (0,0.5) [caption] {Propose \& Vote:\\Lottery};
            
                \coordinate (genesis) at (0,-0.5);
                \node (b11) at (0,-1) [block] {};
                \node (b21) at (-1,-2) [block] {};
                \node (b22) at (+1,-2) [block] {};
                \node (b31) at (-1,-3) [block-mining] {\scriptsize\textbf{?}};
                \node (b32) at (+1,-3) [block-mining] {\scriptsize\textbf{?}};
                
                \draw [link] (b11) -- (genesis);
                \draw [link] (b21) -- (b11);
                \draw [link] (b22) -- (b11);
                \draw [link-mining] (b31) -- (b21);
                \draw [link-mining] (b32) -- (b22);
                
            \end{tikzpicture}
            \begin{tikzpicture}[x=0.25cm,y=0.25cm,baseline=(baseline)]
                \coordinate (baseline) at (0,6);
                
                \draw [draw=none,fill=black!10] (-0.5,1) -- (0,1) -- (0,2) -- (1,0) -- (0,-2) -- (0,-1) -- (-0.5,-1) -- cycle;
            \end{tikzpicture}
            \begin{tikzpicture}[blockchain,baseline=(baseline)]
                \coordinate (baseline) at (0,0.5);
                
                \node (caption) at (0,0.5) [caption] {Confirm:\\$T$-deep LC};
            
                \coordinate (genesis) at (0,-0.5);
                \node (b11) at (0,-1) [block,accepted] {};
                \node (b21) at (-1,-2) [block,accepted] {};
                \node (b22) at (+1,-2) [block] {};
                \node (b31) at (-1,-3) [block] {};
                
                \draw [link,accepted] (b11) -- (genesis);
                \draw [link,accepted] (b21) -- (b11);
                \draw [link] (b22) -- (b11);
                \draw [link] (b31) -- (b21);
                
                \draw [red,thick,shorten <=-0.2em,shorten >=-0.2em] (b31.north west) -- (b31.south east);
                \draw [red,thick,shorten <=-0.2em,shorten >=-0.2em] (b31.south west) -- (b31.north east);
                
            \end{tikzpicture}
        };
        
        \begin{scope}[spy using outlines={circle,black,densely dotted,magnification=3,size=1.3cm,connect spies}]
            \node (box_bft) at (+5,0) [draw,minimum width=6cm,minimum height=4cm,align=center] {%
                $\PIbft$: Psync-BFT (Streamlet)\\[0.75em]
                \begin{tikzpicture}[blockchain,baseline=(baseline)]
                    \coordinate (baseline) at (0,-0.5);
                    
                    \node (caption) at (0,0) [caption] {Propose:};
                
                    \coordinate (genesis) at (0,-0.5);
                    \node (b11) at (0,-1) [block-notarized] {};
                    \node (b21) at (-1,-2) [block-notarized] {};
                    \node (b22) at (+1,-2) [block] {};
                    \node (b31) at (-1,-3) [block-notarized] {};
                    \node (b41) at (-1,-4) [block-propose] {};
                    
                    \draw [link] (b11) -- (genesis);
                    \draw [link] (b21) -- (b11);
                    \draw [link] (b22) -- (b11);
                    \draw [link] (b31) -- (b21);
                    \draw [link-propose] (b41) -- (b31);
                    
                    \coordinate (refintrospectbftblock) at (-1,4);
                    
                \end{tikzpicture}
                \begin{tikzpicture}[x=0.25cm,y=0.25cm,baseline=(baseline)]
                    \coordinate (baseline) at (0,5);
                    
                    \draw [draw=none,fill=black!10] (-0.5,1) -- (0,1) -- (0,2) -- (1,0) -- (0,-2) -- (0,-1) -- (-0.5,-1) -- cycle;
                \end{tikzpicture}
                \begin{tikzpicture}[blockchain,baseline=(baseline)]
                    \coordinate (baseline) at (0,-0.5);
                    
                    \node (caption) at (0,0) [caption] {Vote:};
                
                    \coordinate (genesis) at (0,-0.5);
                    \node (b11) at (0,-1) [block-notarized] {};
                    \node (b21) at (-1,-2) [block-notarized] {};
                    \node (b22) at (+1,-2) [block] {};
                    \node (b31) at (-1,-3) [block-notarized] {};
                    \node (b41) at (-1,-4) [block] {};
                    
                    \draw [link] (b11) -- (genesis);
                    \draw [link] (b21) -- (b11);
                    \draw [link] (b22) -- (b11);
                    \draw [link] (b31) -- (b21);
                    \draw [link] (b41) -- (b31);
                    
                    \draw [vote] ($(b41.north west)-(1,0.5)$) -- (b41);
                    \draw [vote] ($(b41.west)-(1,0)$) -- (b41);
                    \draw [vote] ($(b41.south west)-(1,-0.5)$) -- (b41);
                    \draw [vote] ($(b41.north east)+(1,-0.5)$) -- (b41);
                    \draw [vote] ($(b41.east)+(1,0)$) -- (b41);
                    \draw [vote] ($(b41.south east)+(1,+0.5)$) -- (b41);
                    
                \end{tikzpicture}
                \hspace{0.1em}
                \begin{tikzpicture}[x=0.25cm,y=0.25cm,baseline=(baseline)]
                    \coordinate (baseline) at (0,5);
                    
                    \draw [draw=none,fill=black!10] (-0.5,1) -- (0,1) -- (0,2) -- (1,0) -- (0,-2) -- (0,-1) -- (-0.5,-1) -- cycle;
                \end{tikzpicture}
                \begin{tikzpicture}[blockchain,baseline=(baseline)]
                    \coordinate (baseline) at (0,-0.5);
                    
                    \node (caption) at (0,0) [caption] {Finalize:};
                
                    \coordinate (genesis) at (0,-0.5);
                    \node (b11) at (0,-1) [block-notarized,accepted] {};
                    \node (b21) at (-1,-2) [block-notarized,accepted] {};
                    \node (b22) at (+1,-2) [block] {};
                    \node (b31) at (-1,-3) [block-notarized,accepted] {};
                    \node (b41) at (-1,-4) [block-notarized] {};
                    
                    \draw [link,accepted] (b11) -- (genesis);
                    \draw [link,accepted] (b21) -- (b11);
                    \draw [link] (b22) -- (b11);
                    \draw [link,accepted] (b31) -- (b21);
                    \draw [link] (b41) -- (b31);
                    
                    \draw [red,thick,shorten <=-0.2em,shorten >=-0.2em] (b41.north west) -- (b41.south east);
                    \draw [red,thick,shorten <=-0.2em,shorten >=-0.2em] (b41.south west) -- (b41.north east);
                    
                \end{tikzpicture}
            };
            
            \spy on (2.05,-3.0) in node [fill=white] at (0.75,-5);
        \end{scope}
        
        \node (chain_lc) at (-5,-6) [align=center] {$\Chlc{i}{t}$};
        \node (chain_bft) at (5,-6) [align=center] {$\Chbft{i}{t}$};
        
        \node (txs) at (-5,6) [align=center] {Mempool $\Txs{i}{t}$};
        
        \node (ledger_da) at (-5,-8.25) [align=center] {$\LOGda{i}{t} \triangleq \sanitize(\flatten(\Chbft{i}{t}) \| \Chlc{i}{t})$};
        \node (ledger_fin) at (5,-8.25) [align=center] {$\LOGfin{i}{t} \triangleq \sanitize(\flatten(\Chbft{i}{t}))$};
        
        \draw[-Latex] (txs) -- (box_lc) node [midway,left] {Transactions};
        \draw[-Latex] (box_lc) -- (chain_lc) node [midway,left] {Confirmed chain};
        \draw[-Latex] (chain_lc) -- (-0.5,-6) -- (0.5,5.5) -- (5,5.5) -- (box_bft) node [midway,left] {Chain};
        \draw[-Latex] (box_bft) -- (chain_bft) node [midway,right] {Final chain of chains};
        
        \draw[-Latex,dashed] (0,0) -- (box_bft);
        
        \draw[-Latex] (chain_lc) -- (ledger_da);
        \draw[-Latex] (chain_bft) -- (ledger_da);
        \draw[-Latex] (chain_bft) -- (ledger_fin);

        \node (network) at (0,7.25) {Network};
        \draw [Latex-Latex,double] (network) -| (-2,4);
        \draw [Latex-Latex,double] (network) -| (8,4);
        \draw [Latex-Latex,double] (network) -| (txs);
        
        \node at (0.75,-5) {%
            \scalebox{0.45}{%
                \begin{tikzpicture}[blockchain]
                    \coordinate (genesis) at (0,-0.5);
                    \node (b11) at (0,-1) [block,accepted] {};
                    \node (b21) at (-1,-2) [block,accepted] {};
                    \draw [link,accepted] (b11) -- (genesis);
                    \draw [link,accepted] (b21) -- (b11);
                \end{tikzpicture}
            }
        };
        \draw [-Latex,red,dashed] (-2.8,-2.5) -- (0.5,-5);

    \end{tikzpicture}%
    \vspace{-0.75em}%
    \caption[]{%
        Example \sac protocol (\cf Figure~\ref{fig:protocol-overview-eaf})
        where $\PIlc$ is instantiated with permissioned longest chain
        and $\PIbft$ instantiated with Streamlet,
        as viewed by node $i$ at time $t$.
        Transactions are held in
        mempool $\Txs{i}{t}$. Batched into blocks,
        they are ordered by $\PIlc$
        which outputs a chain $\Chlc{i}{t}$
        (representing $\LOGlc{}{}$)
        of confirmed transactions.
        Snapshots of $\Chlc{}{}$
        (which themselves are chains,
        \cf the magnifying glass)
        are input to and ordered by $\PIbft$
        which outputs a chain $\Chbft{i}{t}$
        (representing $\LOGbft{}{}$)
        of final snapshots.
        In addition, $\Chlc{i}{t}$ is used
        as side information in $\PIbft$
        to boycott the finalization of invalid snapshots
        (dashed arrow).
        Finally, $\Chbft{i}{t}$ is flattened
        and
        sanitized
        to obtain the finalized ledger $\LOGfin{i}{t}$,
        which is prepended to $\Chlc{i}{t}$ and sanitized
        to form the available ledger $\LOGda{}{}$
        (\cf Figure~\ref{fig:ledger-extraction-details}).
    }
    \label{fig:protocol-overview}
\end{figure*}

In this section, we give an example of \theconstruction,
$\PItheexample{Sleepy}{Streamlet}$,
where we instantiate $\PIlc$ with
a permissioned longest-chain protocol %
and $\PIbft$ with a variant of (partially synchronous) Streamlet \cite{streamlet}.
Note that all of the longest chain protocols such as \cite{sleepy,snowwhite, kiayias2017ouroboros,david2018ouroboros, badertscher2018ouroboros,posat}
are suited to instantiate $\PIlc$.
For concreteness, we will follow Sleepy \cite{sleepy} when we get to details.
Streamlet \cite{streamlet} is the latest representative
of a line of works \cite{pbft,tendermint,sbft,yin2018hotstuff}
striving to
simplify and speed up
BFT consensus.
Due to its remarkable simplicity, Streamlet
is well-suited to illustrate our approach.
For application requirements,
other BFT protocols might be better suited.
We demonstrate
in Section~\ref{sec:other-bft-protocols}
that our technique readily extends to other BFT protocols
such as HotStuff \cite{yin2018hotstuff} and PBFT \cite{pbft}.

Before we delve into the details of our construction,
we review the basic mechanics of the constituent protocols
$\PIlc$ and $\PIbft$
(illustrated in
the two boxes of Figure~\ref{fig:protocol-overview}).
In permissioned longest chain protocols
a cryptographic lottery rate-limits
the production of new blocks
(\tikz[baseline=-3pt]{
    \node at (0,0) [
        minimum width=0.25cm,
        minimum height=0.25cm,
        draw,
        shade,
        top color=white,
        bottom color=black!10,
        draw=myParula01Blue,
        text=myParula01Blue,
        bottom color=myParula01Blue!10,
        inner sep=0,
    ] {\scriptsize\textbf{?}};
}).
Honest block proposers
extend the longest chain (and thus vote for it),
and blocks of a certain depth on the longest chain are confirmed
(\tikz[baseline=-3pt]{
    \node at (0,0) [
        minimum width=0.25cm,
        minimum height=0.25cm,
        draw,
        shade,
        top color=white,
        bottom color=black!10,
        draw=red,
        ultra thick,
    ] {};
}).
Streamlet proceeds in epochs of fixed duration,
each of which is associated with a pseudo-randomly chosen leader.
At the beginning of each epoch,
the leader proposes a new block
(\tikz[baseline=-3pt]{
    \node at (0,0) [
        minimum width=0.25cm,
        minimum height=0.25cm,
        draw,
        shade,
        top color=white,
        bottom color=black!10,
        draw=myParula01Blue,
        text=myParula01Blue,
        bottom color=myParula01Blue!10,
        inner sep=0,
    ] {};
})
extending the longest chain of notarized blocks
(\tikz[baseline=-3pt]{
    \node at (0,0) [
        minimum width=0.25cm,
        minimum height=0.25cm,
        draw,
        shade,
        top color=white,
        bottom color=black!10,
        draw=myParula05Green!50!black,
        bottom color=myParula05Green!50,
    ] {};
}).
Then, all nodes vote %
(\tikz[baseline=-3pt]{
    \node (b41) at (0,0) [
        minimum width=0.25cm,
        minimum height=0.25cm,
        draw,
        shade,
        top color=white,
        bottom color=black!10,
    ] {};
    \draw [x=0.25cm,y=0.5cm,-latex,draw=myParula01Blue] ($(b41.north west)-(1,0.5)$) -- (b41);
    \draw [x=0.25cm,y=0.5cm,-latex,draw=myParula01Blue] ($(b41.west)-(1,0)$) -- (b41);
    \draw [x=0.25cm,y=0.5cm,-latex,draw=myParula01Blue] ($(b41.south west)-(1,-0.5)$) -- (b41);
    \draw [x=0.25cm,y=0.5cm,-latex,draw=myParula01Blue] ($(b41.north east)+(1,-0.5)$) -- (b41);
    \draw [x=0.25cm,y=0.5cm,-latex,draw=myParula01Blue] ($(b41.east)+(1,0)$) -- (b41);
    \draw [x=0.25cm,y=0.5cm,-latex,draw=myParula01Blue] ($(b41.south east)+(1,+0.5)$) -- (b41);
}),
and the block becomes notarized if at least two-thirds of the nodes have voted for it.
Out of three adjacent notarized blocks from consecutive epochs,
the middle one gets finalized
(\tikz[baseline=-3pt]{
    \node at (0,0) [
        minimum width=0.25cm,
        minimum height=0.25cm,
        draw,
        shade,
        top color=white,
        bottom color=black!10,
        draw=myParula05Green!50!black,
        bottom color=myParula05Green!50,
        draw=red,
        ultra thick,
    ] {};
})
along with its prefix.

For the example construction, we follow the
blueprint of
Section~\ref{sec:abstract-eaf-construction}
but, in line with the protocols adopted for $\PIlc$ and $\PIbft$,
choose blockchains as a more suitable representation for
ledgers.
The above instantiation leads from
the high-level Figure~\ref{fig:protocol-overview-eaf}
to the concrete Figure~\ref{fig:protocol-overview}
which illustrates the overall protocol as viewed by node $i$ at time $t$.

Transactions are received from the environment and held in the
\emph{mempool} $\Txs{i}{t}$.
Batched into \emph{blocks}, they are ordered by $\PIlc$
which outputs a block\emph{chain} $\Chlc{i}{t}$
(comprised of \emph{\blockslc} and representing the ledger $\LOGlc{}{}$ in Figure~\ref{fig:protocol-overview-eaf})
of transactions considered \emph{confirmed}.
Snapshots of $\Chlc{}{}$
(which themselves are chains)
are input to and ordered by $\PIbft$
which outputs a blockchain $\Chbft{i}{t}$
(comprised of \emph{\blocksbft} and representing the ledger $\LOGbft{}{}$ in Figure~\ref{fig:protocol-overview-eaf})
of snapshots considered \emph{final}.
In addition, $\Chlc{i}{t}$ is used
as side information in $\PIbft$
to boycott the finalization of invalid snapshots proposed by the adversary.
Finally, $\Chbft{i}{t}$ is \emph{flattened} (\ie, all snapshots
are concatenated as ordered)
and
\emph{sanitized} (\ie, only the first
valid
occurrence of a
transaction
remains)
to obtain the finalized ledger $\LOGfin{i}{t}$,
which is prepended to $\Chlc{i}{t}$ and sanitized
to form the available ledger $\LOGda{}{}$
(see Section~\ref{sec:ledger-extraction-details}).%
\footnote{%
Formally, $\LOGda{}{}$ and $\LOGfin{}{}$ are now
represented as sequences of \blockslc.
Proper transactions ledgers are readily obtained by
concatenating the transactions contained in the blocks
and removing duplicate and invalid transactions (\emph{sanitization}).%
}

In the following, we provide more explanation
for the following three details,
\begin{enumerate*}
    \item how snapshots are represented efficiently,
    \item how Streamlet is modified
        to prevent that an adversary can input an ostensible snapshot
        which is really unconfirmed (this would break safety),
        and
    \item how the transaction ledgers are extracted from
        the blockchains $\Chlc{i}{t}$ and $\Chbft{i}{t}$.
\end{enumerate*}

\subsubsection{Efficient representation of snapshots}

We use (variants of) the symbols `$b$' and `$B$' to refer to
blocks in the blockchains $\Chlc{i}{t}$ and $\Chbft{i}{t}$
output by
$\PIlc$
and
$\PIbft$,
respectively.
An \blocklc $b$ contains as payload transactions denoted as `$b.\PayloadTxs$'.
Note that due to the blockchain structure, a single block
uniquely identifies a whole chain of blocks, namely that of its
ancestors all the way back to the genesis block.
A snapshot of a blockchain can thus be represented efficiently by
pointing to the block at the tip of the chain.
Thus, instead of copying a whole chain of \blockslc
into each \blockbft,
a \blockbft $B$ contains as payload only a \emph{reference},
denoted by `$B.\PayloadCh$',
to an \blocklc representing the snapshot.

For
ledgers and blockchains,
`$\preceq$' is canonically defined as the `is a prefix of' relation.
As blocks identify chains, the definition of `$\preceq$'
naturally carries over: for two blocks $\flat$ and $\flat'$,
$\flat \preceq \flat'$ iff the chain identified by $\flat$
is a prefix of the chain identified by $\flat'$.
The depth of a block
is
the length of the chain it identifies,
excluding the genesis block.

\subsubsection{Modification of Streamlet}

\begin{algorithm}
    \caption{\footnotesize Pseudocode of example \eaf construction with Sleepy as $\PIlc$ and Streamlet as $\PIbft$}
    \label{algo:pseudocode-sleepy-streamlet}
    \begin{algorithmic}[1]
        \footnotesize
        \Procedure{LcSlot}{$t$}
            \If{$\operatorname{SleepyIsWinningLotteryTicket}(i, t)$}
                \State $b^* \gets \Call{SleepyTipLC}{\null}$
                \State $b \gets \operatorname{SleepyNewBlock}(b^*, t, i, \Txs{}{t})$
                \State \Call{Broadcast}{$b$}
            \EndIf
        \EndProcedure
        \Procedure{BftSlot}{$t$}
            \State $e \gets \operatorname{StreamletEpoch}(t)$
            \If{$\operatorname{StreamletIsStartOfProposePhase}(t)$}
                \If{$\operatorname{StreamletEpochLeader}(e) = i$}
                    \State $B^* \gets \Call{StreamletTipNotarizedLC}{\null}$
                    \State $B \gets \operatorname{StreamletNewBlock}(B^*, e, \Chlc{}{t})$
                    \State $P \gets \operatorname{StreamletNewProposal}(B, i)$
                    \State \Call{Broadcast}{$B, P$}
                \EndIf
            \ElsIf{$\operatorname{StreamletIsStartOfVotePhase}(t)$}
                \State $P \gets \Call{StreamletFirstValidProposal}{e}$
                \If{\textcolor{red}{$P.B.\PayloadCh \preceq \Chlc{}{t}$}}
                        \label{algo:pseudocode-sleepy-streamlet-condition}
                    \State $V \gets \operatorname{StreamletNewVote}(P.B, i)$
                    \State \Call{Broadcast}{$V$}
                \EndIf
            \EndIf
        \EndProcedure
        \Procedure{Main}{\null}
            \For{time slot $t \gets 1, 2, 3, ...$}
                \State \Call{ProcessIncomingNetworkMessages}{\null}
                    \label{algo:pseudocode-sleepy-streamlet-process-msgs}
                \State \Call{EchoIncomingNetworkMessages}{\null}
                    \label{algo:pseudocode-sleepy-streamlet-echo-msgs}
                \State \Call{LcSlot}{$t$}
                \State $\Chlc{}{t} \gets \Call{LcConfirmedChain}{\null}$
                \State \Call{BftSlot}{$t$}
                \State $\Chbft{}{t} \gets \Call{BftFinalChain}{\null}$
            \EndFor
        \EndProcedure
    \end{algorithmic}
\end{algorithm}

With the payload
of Streamlet being snapshots,
honest epoch leaders are instructed to, when they propose a block,
take a snapshot of $\Chlc{i}{t}$ and include a reference to its tip
as payload in the new \blockbft.
Furthermore, Streamlet needs to be modified to ensure that
an adversary cannot input an ostensible snapshot
which is not really entirely confirmed.
To this end, the voting rule of Streamlet is extended by the following condition:
An honest node only votes for a proposed \blockbft $B$ if it views
$B.\PayloadCh$ as confirmed.
In effect,
side information about $\PIlc$ is used in $\PIbft$
to prevent the finalization of invalid snapshots
proposed by the adversary.
Pseudocode of the overall protocol as executed on node $i$
is found in Algorithm~\ref{algo:pseudocode-sleepy-streamlet}.
Proper functions of only their inputs
and procedures that access global state
are denoted as `$\operatorname{Function}(...)$'
and `$\Call{Procedure}{...}$', respectively.
Incoming network messages (new blocks, proposals and votes) are processed, and
the global state is adjusted accordingly, in line~\ref{algo:pseudocode-sleepy-streamlet-process-msgs}.
Honest nodes echo messages they receive,
see line~\ref{algo:pseudocode-sleepy-streamlet-echo-msgs}.
As a result,
if an honest node observes a message at time $t$ then
all honest nodes will have observed the message
by time $\max(\GST, t + \Delta)$.
The additional constraint in the voting rule
with respect to `vanilla'
Streamlet is highlighted \textcolor{red}{red}
(line~\ref{algo:pseudocode-sleepy-streamlet-condition}).
Note that Sleepy is applied unaltered and the modification required
for Streamlet is minor.
The same is true when instantiating
the sub-protocol $\PIbft$
with other
partially synchronous BFT protocols
such as HotStuff \cite{yin2018hotstuff} or PBFT \cite{pbft}, detailed
in Section~\ref{sec:other-bft-protocols}.

\subsubsection{Ledger extraction}
\label{sec:ledger-extraction-details}

\begin{figure}
    \centering
    \begin{tikzpicture}[
        x=1.8cm,
        y=1.2cm,
    ]
        \footnotesize
        
        \tikzset{
            _block/.style = {
                draw,
                shade,
                top color=white,
                bottom color=black!10,
            },
            largeblock/.style = {
                _block,
                minimum width=1cm,
                minimum height=1cm,
                draw,
                draw=myParula05Green!50!black,
                bottom color=myParula05Green!50,
            },
            largelink/.style = {
                draw,
            },
            smallblock/.style = {
                _block,
                minimum width=0.25cm,
                minimum height=0.28cm,
                draw,
                inner sep=0,
                text=black,
            },
            smalllink/.style = {
                draw,
            },
        }

        \coordinate (blgen) at (0,0.60);
        \node (bl1) at (0,0) [largeblock] {};
        \node (bl2) at (0,-1) [largeblock] {};
        \node (bl3) at (0,-2) [largeblock] {};
        
        \draw [largelink] (bl3) -- (bl2);
        \draw [largelink] (bl2) -- (bl1);
        \draw [largelink] (bl1) -- (blgen);

        \coordinate (bsgen) at (0,-2.75);
        \node (bs1) at (0,-3) [smallblock] {1};
        \node (bs2) at (0,-3.33) [smallblock] {2};
        \node (bs3) at (0,-3.66) [smallblock] {3};
        
        \draw [largelink] (bs3) -- (bs2);
        \draw [largelink] (bs2) -- (bs1);
        \draw [largelink] (bs1) -- (bsgen);

        \node at (-0.6,0) {\small $\Chbft{i}{t}$:};
        \node at (-0.6,-3) {\small $\Chlc{i}{t}$:};

        \begin{scope}[yshift=0.0cm]
            \coordinate (bsgen) at (0,0.2);
            \node (bs1) at (0,0) [smallblock] {1};
            
            \draw [largelink] (bs1) -- (bsgen);
        \end{scope}
        
        \begin{scope}[yshift=-1.05cm]
            \coordinate (bsgen) at (0,0.2);
            \node (bs1) at (0,0) [smallblock] {1};
            \node (bs2) at (0,-0.33) [smallblock] {2};
            
            \draw [largelink] (bs2) -- (bs1);
            \draw [largelink] (bs1) -- (bsgen);
        \end{scope}
        
        \begin{scope}[yshift=-2.25cm]
            \coordinate (bsgen) at (0,0.2);
            \node (bs1) at (0,0) [smallblock] {1};
            \node (bs2) at (0,-0.33) [smallblock] {2};
            
            \draw [largelink] (bs2) -- (bs1);
            \draw [largelink] (bs1) -- (bsgen);
        \end{scope}
        
        \draw [-Latex,shorten <=0.5em,shorten >=0.85em] (bl1) -- ++(1,0);
        \draw [-Latex,shorten <=0.5em,shorten >=0.85em] (bl2) -- ++(1,0) node [midway,xshift=-4pt,fill=white,rotate=90] {\scriptsize\emph{Flatten}};
        \draw [-Latex,shorten <=0.5em,shorten >=0.85em] (bl3) -- ++(1,0);
        
        \begin{scope}[yshift=0.0cm,xshift=1.8cm]
            \coordinate (bsgen) at (0,0.2);
            \node (bs1) at (0,0) [smallblock] {1};
        \end{scope}
        
        \begin{scope}[yshift=-1.05cm,xshift=1.8cm]
            \node (bs1) at (0,0) [smallblock] {1};
            \node (bs2) at (0,-0.33) [smallblock] {2};
            
            \draw [red,thick,shorten <=-0.2em,shorten >=-0.2em] (bs1.north west) -- (bs1.south east);
            \draw [red,thick,shorten <=-0.2em,shorten >=-0.2em] (bs1.south west) -- (bs1.north east);
        \end{scope}
        
        \begin{scope}[yshift=-2.25cm,xshift=1.8cm]
            \node (bs1) at (0,0) [smallblock] {1};
            \node (bs2) at (0,-0.33) [smallblock] {2};
            
            \draw [red,thick,shorten <=-0.2em,shorten >=-0.2em] (bs1.north west) -- (bs1.south east);
            \draw [red,thick,shorten <=-0.2em,shorten >=-0.2em] (bs1.south west) -- (bs1.north east);
            
            \draw [red,thick,shorten <=-0.2em,shorten >=-0.2em] (bs2.north west) -- (bs2.south east);
            \draw [red,thick,shorten <=-0.2em,shorten >=-0.2em] (bs2.south west) -- (bs2.north east);
        \end{scope}

        \draw [decorate,decoration={brace,amplitude=5pt,aspect=0.85}]
        ($(bl1.north east) + (0.9,0)$) -- ($(bl3.south east) + (0.9,0)$);
        
        \draw [-Latex,shorten >=0.85em] (1.35,-2) -- (1.75,-2) node [midway,xshift=-4pt,yshift=2.5em,rotate=90] {\scriptsize\emph{Sanitize}}; %

        \begin{scope}[xshift=-0.4cm]
        
            \begin{scope}[yshift=-2.25cm,xshift=3.6cm]
                \node (bs1) at (0,0) [smallblock] {1};
                \node (bs2) at (0,-0.33) [smallblock] {2};
            \end{scope}
            
            \draw [-Latex,shorten >=0.85em] (2.15,-2) -- (2.5,-2);
            
            \begin{scope}[yshift=-2.25cm,xshift=4.5cm]
                \node (bs1) at (0,0) [smallblock,alias=br11] {1};
                \node (bs2) at (0,-0.33) [smallblock] {2};
            \end{scope}
            
            \begin{scope}[xshift=4.5cm]
                \node (bs1) at (0,-3) [smallblock] {1};
                \node (bs2) at (0,-3.33) [smallblock] {2};
                \node (bs3) at (0,-3.66) [smallblock,alias=br12] {3};
                
                \draw [red,thick,shorten <=-0.2em,shorten >=-0.2em] (bs1.north west) -- (bs1.south east);
                \draw [red,thick,shorten <=-0.2em,shorten >=-0.2em] (bs1.south west) -- (bs1.north east);
                
                \draw [red,thick,shorten <=-0.2em,shorten >=-0.2em] (bs2.north west) -- (bs2.south east);
                \draw [red,thick,shorten <=-0.2em,shorten >=-0.2em] (bs2.south west) -- (bs2.north east);
            \end{scope}
            
            \draw [decorate,decoration={brace,amplitude=5pt,aspect=0.78}]
            ($(br11.north east) + (0.1,0)$) -- ($(br12.south east) + (0.1,0)$);
            
            \begin{scope}[xshift=-0.4cm]
            
            \begin{scope}[xshift=6.3cm]
                \node (bs1) at (0,-3) [smallblock] {1};
                \node (bs2) at (0,-3.33) [smallblock] {2};
                \node (bs3) at (0,-3.66) [smallblock,alias=br12] {3};
            \end{scope}
            
            \node (LOGda) at (3.5,-4.5) {$\LOGda{i}{t}$};
            \draw [-Latex,dashed,shorten <=3.9em] (3.5,-2.9) -- (LOGda);
            
            \end{scope}
            
            \draw [-Latex,shorten >=0.85em,shorten <=1em] (0.2,-3.33) -- (2.5,-3.33);
            
            \draw [-Latex,shorten >=0.85em,shorten <=1em] (2.7,-3.33) -- (3.25,-3.33) node [midway,xshift=0pt,yshift=2.5em,rotate=90] {\scriptsize\emph{Sanitize}}; %
            
            \node (LOGfin) at (2,-4.5) {$\LOGfin{i}{t}$};
            \draw [-Latex,dashed,shorten <=2em] (2,-2) -- (LOGfin);
        
        \end{scope}
    
    \end{tikzpicture}%
    \vspace{-0.75em}%
    \caption[]{%
        $\Chbft{i}{t}$ is flattened and sanitized
        to obtain the finalized ledger $\LOGfin{i}{t}$,
        which is prepended to $\Chlc{i}{t}$ and sanitized
        to form the available ledger $\LOGda{}{}$.
    }
    \label{fig:ledger-extraction-details}
 \end{figure}
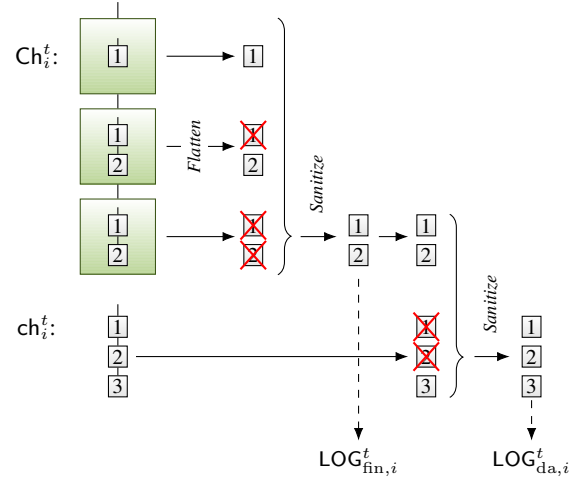

Finally,
how honest nodes compute $\LOGfin{i}{t}$ and $\LOGda{i}{t}$ from $\Chbft{i}{t}$ and $\Chlc{i}{t}$
is illustrated in Figure~\ref{fig:ledger-extraction-details}.
Recall that $\Chbft{i}{t}$ is an ordering of snapshots, \ie,
a chain of chains of \blockslc.
First, $\Chbft{i}{t}$ is flattened,
\ie, the chains of blocks are concatenated as ordered
to arrive at a single sequence of \blockslc.
Then, all but the first occurrence of each block are removed (sanitized)
to arrive at the finalized ledger $\LOGfin{i}{t}$ of \blockslc.
To form the available ledger $\LOGda{i}{t}$,
$\Chlc{i}{t}$, which is a sequence of \blockslc,
is appended to $\LOGfin{i}{t}$ and the result again sanitized.

\subsection{Analysis}
\label{sec:analysis}

\begin{figure}
    \centering
    \begin{tikzpicture}[
        x=1.5cm,
        y=1.8cm,
        _box/.style = {
            minimum width=2.5cm,
            minimum height=1.5cm,
            align=center,
            inner sep=0,
        },
        statement/.style = {
            _box,
            text depth=0.4cm,
            draw,
        },
        label/.style = {
            minimum height=0.4cm,
            minimum width=0.5cm,
            align=center,
            anchor=south west,
            fill=black!10,
            draw,
            inner sep=0,
        },
        reference/.style = {
            minimum height=0.4cm,
            align=center,
            anchor=south east,
            fill=black!10,
            draw,
            inner xsep=3pt,
        },
        leadsto/.style = {
            -Latex,
        },
    ]
        \footnotesize
        
        \node (ass1) at (-2,0) [statement] {Safety of\\ $\PIbft$};
        \node at (ass1.south west) [label] {1};
        \node at (ass1.south east) [reference] {Lemma~\ref{thm:streamlet-safety}};
        
        \node (ass2) at (0,0) [statement] {Liveness of\\ $\PIbft$ after\\ $\max(\GST,\GOT)$};
        \node at (ass2.south west) [label] {2};
        \node at (ass2.south east) [reference] {Lemmas~\ref{thm:streamlet-liveness-1},\ref{thm:streamlet-liveness}};
        
        \node (ass3) at (2,0) [statement] {Security of\\ $\PIlc$ after\\ $\max(\GST,\GOT)$};
        \node at (ass3.south west) [label] {3};
        \node at (ass3.south east) [reference] {Theorem~\ref{thm:pos-security-under-advp-main}};
        
        \node (con1) at (1,-1) [statement] {Liveness of\\ $\LOGfin{}{}$ after\\ $\max(\GST,\GOT)$};
        \node at (con1.south west) [label] {4};
        \node at (con1.south east) [reference] {Lemma~\ref{thm:pibft-liveness}};
        
        \node (con2) at (0,-2) [statement] {Security of\\ $\LOGfin{}{}$};
        \node at (con2.south west) [label] {5};
        \node at (con2.south east) [reference] {Theorem~\ref{thm:main-security}};
        
        \draw [leadsto] (ass2) -> (con1);
        \draw [leadsto] (ass3) -> (con1);
        \draw [leadsto] (ass1) -> (con2);
        \draw [leadsto] (con1) -> (con2);
    
    \end{tikzpicture}%
    \vspace{-0.5em}%
    \caption[]{Dependency of the security of $\LOGfin{}{}$ under $\AdvEnvOneOpt$ on the properties of $\PIlc$ and $\PIbft$. 
    Boxes represent the properties and the arrows indicate the implications of these properties.
    Theorems and lemmas used to validate the properties are displayed at the bottom right corner of each box.}
    \label{fig:proof-layout-1}
\end{figure}

\begin{figure}
    \centering
    \begin{tikzpicture}[
        x=1.5cm,
        y=1.8cm,
        _box/.style = {
            minimum width=2.5cm,
            minimum height=1.5cm,
            align=center,
            inner sep=0,
        },
        statement/.style = {
            _box,
            text depth=0.4cm,
            draw,
        },
        label/.style = {
            minimum height=0.4cm,
            minimum width=0.5cm,
            align=center,
            anchor=south west,
            fill=black!10,
            draw,
            inner sep=0,
        },
        reference/.style = {
            minimum height=0.4cm,
            align=center,
            anchor=south east,
            fill=black!10,
            draw,
            inner xsep=3pt,
        },
        leadsto/.style = {
            -Latex,
        },
    ]
        \footnotesize
        
        \node (ass1) at (-1,0) [statement] {Security of\\ $\PIlc$};
        \node at (ass1.south west) [label] {6};
        \node at (ass1.south east) [reference] {Given};
        
        \node (ass2) at (1,0) [statement] {Consistency of\\ 
        $\LOGfin{}{}$ with\\ 
        the output of $\PIlc$};
        \node at (ass2.south west) [label] {7};
        \node at (ass2.south east) [reference] {Lemma~\ref{thm:pibft-safe-under-2}};
        
        \node (con1) at (0,-1) [statement] {Security of\\ $\LOGda{}{}$};
        \node at (con1.south west) [label] {8};
        \node at (con1.south east) [reference] {Theorem~\ref{thm:main-security}};
        
        \draw [leadsto] (ass1) -> (con1);
        \draw [leadsto] (ass2) -> (con1);
    
    \end{tikzpicture}%
    \vspace{-0.5em}%
    \caption[]{Dependency of the security of $\LOGda{}{}$ under $\AdvEnvTwoOpt$ on the properties of $\PIlc$ and $\PIbft$. 
    Boxes represent the properties and the arrows indicate the implications of these properties.
    Theorems and lemmas used to validate the properties are displayed at the bottom right corner of each box.}
    \label{fig:proof-layout-2}
\end{figure}

In this section, we analyze the security of $\PItheexample{Sleepy}{Streamlet}$ as an \eaf protocol and show that it is optimally resilient:

\begin{theorem}
\label{thm:main-security}
$\PItheexample{Sleepy}{Streamlet}$ is a $(\frac{1}{3},\frac{1}{2})$-secure \eaf protocol.
\end{theorem}

Observe that no \eaf protocol can tolerate a Byzantine adversary $\Adv_1(\beta_1)$ with $\beta_1 \geq \frac{1}{3}$ in a partially synchronous network.
Similarly, no \eaf protocol can tolerate a Byzantine adversary $\Adv_2(\beta_2)$ with $\beta_2 \geq \frac{1}{2}$ in a synchronous network.
Hence the security of $\PItheexample{Sleepy}{Streamlet}$ implies that it is optimally resilient.
We denote the worst-case adversary-environments as
$\AdvEnvOneOpt \triangleq \AdvEnvParameterized{1}{\frac{1}{3}}$ and $\AdvEnvTwoOpt \triangleq \AdvEnvParameterized{2}{\frac{1}{2}}$.

We now focus on the proof of Theorem \ref{thm:main-security},
which
proceeds as illustrated in Figures~\ref{fig:proof-layout-1} and \ref{fig:proof-layout-2}
and
along with proofs for the Lemmas
can be found in Appendix~\ref{sec:analysis-streamlet}.
Proof of Theorem~\ref{thm:pos-security-under-advp-main} is given in Appendix~\ref{sec:sleepy-security}.

We first show the safety and liveness (after time $\max\{\GST,\GOT\}$) of the ledger $\LOGfin{}{}$ under $\AdvEnvOneOpt$. 
Figure \ref{fig:proof-layout-1} visualizes the dependency of the security of $\LOGfin{}{}$ on the properties of the sub-protocols $\PIlc$ and $\PIbft$.
We see from Figure \ref{fig:proof-layout-1} that the safety of $\PIbft$ (box 1) implies the safety of $\LOGfin{}{}$ (box 5).
However, in Figure \ref{fig:protocol-overview}, $\Txs{}{}$ do not immediately arrive at $\PIbft$. 
They are first received by $\PIlc$ and become part of its output ledger, snapshots of which are then inputted to $\PIbft$.
Consequently, liveness of $\LOGfin{}{}$ after time $\max\{\GST,\GOT\}$ (box 4) does not only require the liveness of $\PIbft$ (box 2), but also the security $\PIlc$ after time $\max\{\GST,\GOT\}$ (box 3).

We observe via Lemmas \ref{thm:streamlet-safety}, \ref{thm:streamlet-liveness-1} and \ref{thm:streamlet-liveness} that the changes in Streamlet described by lines 13 and 19 of Algorithm \ref{algo:pseudocode-sleepy-streamlet} does not affect the validity of the safety and liveness proofs in \cite{streamlet}.
Hence, security of $\PIbft$ (boxes 1 and 2) directly follows from the security proof of Streamlet.
However, showing the security of $\PIlc$ (box 3) claimed by Theorem \ref{thm:pos-security-under-advp-main}, requires some work. 
For this purpose, we extend the concept of \emph{pivots} as defined in \cite{sleepy}, to a partially synchronous network. 
Pivots are time slots such that every honest node has the same view of the prefix of the longest chain up to the pivot. %
The original definition of pivots in \cite{sleepy} ensures the convergence of longest chains by requiring any time interval containing the pivot to have more convergence opportunities (honest slots which are sufficiently apart)  than adversarial slots.
However, this requirement fails to ensure convergence under partial synchrony as the isolated honest nodes can fail to build a blockchain before $\max\{\GST,\GOT\}$.
Hence, we define the concept of a $\GST$-strong pivot that considers only the honest slots after $\max\{\GST,\GOT\}$ within any interval around the $\GST$-strong pivot.
Although this definition makes the arrival of $\GST$-strong pivots less likely, Appendix~\ref{sec:sleepy-security} proves that $\GST$-strong pivots appear in $O(\max\{\GST,\GOT\})$ time following $\max\{\GST,\GOT\}$, thus, concluding the security of $\PIlc$ after $\max\{\GST,\GOT\}$.
Finally, Lemma \ref{thm:pibft-liveness} combines the security of $\PIlc$ and liveness of $\PIbft$ after time $\max\{\GST,\GOT\}$ to show the liveness of $\LOGfin{}{}$.

We next show the safety and liveness of the ledger $\LOGda{}{}$ under $\AdvEnvTwoOpt$. 
Figure \ref{fig:proof-layout-2} visualizes the dependency of security of $\LOGda{}{}$ on the properties of sub-protocols $\PIlc$ and $\PIbft$.

In Figure \ref{fig:protocol-overview}, $\LOGda{}{}$ is a concatenation of $\LOGfin{}{}$ with the output ledger of $\PIlc$.
Hence, although the security of $\PIlc$ (box 6) is a necessary condition for the security of $\LOGda{}{}$ (box 8), we also need the prefix $\LOGfin{}{}$ to be consistent with the output of $\PIlc$ in the view of every honest node at all times (box 7), to guarantee the safety of the whole ledger $\LOGda{}{}$.

Security of $\PIlc$ follows from the security proofs of the respective protocol used for $\PIlc$.
However, proving the consistency of $\LOGfin{}{}$ with the output of $\PIlc$ as claimed by Lemma \ref{thm:pibft-safe-under-2}, requires a careful look at the finalization rule of $\PIbft$.
As indicated by Algorithm~\ref{algo:pseudocode-sleepy-streamlet}, a snapshot of the output of $\PIlc$ becomes \finalp as part of a \blockbft only if that snapshot is seen as \finalda by at least one honest node.
However, since $\PIlc$ is safe, the fact that one honest node sees that snapshot as \finalda implies that every honest node sees the same snapshot as \finalda.
Consequently, the ledger $\LOGfin{}{}$ will be generated from the same snapshots in the view of every honest node.
Moreover, as these snapshots are \finalda prefixes of the output of $\PIlc$ and $\PIlc$ is safe, $\LOGfin{}{}$ is a prefix of the output of $\PIlc$ in the view of any honest node at all times.
Finally, since $\LOGfin{}{}$ is a prefix of $\LOGda{}{}$ by construction, the prefix property holds trivially.

\tikzset{blockchain/.style={
        x=0.3cm,
        y=0.6cm,
        node distance=0.5cm,
        block/.style = {
            minimum width=0.25cm,
            minimum height=0.25cm,
            draw,
            shade,
            top color=white,
            bottom color=black!10,
        },
        block-notarized/.style = {
            block,
            draw=myParula05Green!50!black,
            bottom color=myParula05Green!50,
        },
        link/.style = {
        },
        refsnapshot/.style = {
            {Circle[length=3pt]}-latex,
            myParula01Blue,
            bend right=30,
            shorten <=-1.5pt,
            thick,
        },
    }
}

\begin{figure}%
    \centering%
    \vspace{-1.25em}%
    \subfloat[$\PIlc$ and $\PIbft$ safe\label{fig:examples-one-protocol-unsafe-allsafe}]{%
        \begin{minipage}[t]{0.3\linewidth}\centering%
        \begin{tikzpicture}%
            \small
            
            \begin{scope}[blockchain,xshift=-0.65cm]
            
                \node (label) at (0,0.25) {LC};
            
                \coordinate (b0) at (0,-0.5);
                \node (b11) at (0,-1) [block] {};
                \node (b21) at (0,-2) [block] {};
                \node (b31) at (0,-3) [block] {};
                \node (b41) at (0,-4) [block] {};
                \node (b51) at (0,-5) [block] {};
                \node (b61) at (0,-6) [block] {};

                \draw [link] (b11) -- (b0);
                \draw [link] (b21) -- (b11);
                \draw [link] (b31) -- (b21);
                \draw [link] (b41) -- (b31);
                \draw [link] (b51) -- (b41);
                \draw [link] (b61) -- (b51);
                
            \end{scope}
            
            \begin{scope}[blockchain,xshift=+0.65cm]
            
                \node (label) at (0,0.25) {BFT};
                
                \coordinate (B0) at (0,-0.5);
                \node (B11) at (0,-1) [block-notarized] {};
                \node (B21) at (0,-2) [block-notarized] {};
                \node (B31) at (0,-3) [block-notarized] {};
                \node (B41) at (0,-4) [block-notarized] {};
                \node (B51) at (0,-5) [block-notarized] {};

                \draw [link] (B11) -- (B0);
                \draw [link] (B21) -- (B11);
                \draw [link] (B31) -- (B21);
                \draw [link] (B41) -- (B31);
                \draw [link] (B51) -- (B41);
                
            \end{scope}
            
            \begin{scope}[blockchain]
            
                \draw [refsnapshot] (B11.center) to (b11);
                \draw [refsnapshot] (B21.center) to (b21);
                \draw [refsnapshot] (B31.center) to (b31);
                \draw [refsnapshot] (B41.center) to (b41);
                \draw [refsnapshot] (B51.center) to (b51);
                
            \end{scope}
            
        \end{tikzpicture}%
        \end{minipage}%
    }%
    \hfill%
    \subfloat[$\PIlc$ unsafe\label{fig:examples-one-protocol-unsafe-PIlc}]{%
        \begin{minipage}[t]{0.33\linewidth}\centering%
        \begin{tikzpicture}
            \small
            
            \begin{scope}[blockchain,xshift=-0.75cm]
            
                \node (label) at (0,0.25) {LC};
            
                \coordinate (b0) at (0,-0.5);
                \node (b11) at (0,-1) [block] {};
                \node (b21) at (-1,-2) [block] {};
                \node (b22) at (+1,-2) [block] {};
                \node (b31) at (-1,-3) [block] {};
                \node (b32) at (+1,-3) [block] {};
                \node (b41) at (-1,-4) [block] {};
                \node (b42) at (+1,-4) [block] {};
                \node (b52) at (+1,-5) [block] {};
                \node (b62) at (+1,-6) [block] {};

                \draw [link] (b11) -- (b0);
                \draw [link] (b21) -- (b11);
                \draw [link] (b22) -- (b11);
                \draw [link] (b31) -- (b21);
                \draw [link] (b32) -- (b22);
                \draw [link] (b41) -- (b31);
                \draw [link] (b42) -- (b32);
                \draw [link] (b52) -- (b42);
                \draw [link] (b62) -- (b52);
                
            \end{scope}
            
            \begin{scope}[blockchain,xshift=+0.75cm]
            
                \node (label) at (0,0.25) {BFT};
                
                \coordinate (B0) at (0,-0.5);
                \node (B11) at (0,-1) [block-notarized] {};
                \node (B21) at (0,-2) [block-notarized] {};
                \node (B31) at (0,-3) [block-notarized] {};
                \node (B41) at (0,-4) [block-notarized] {};
                \node (B51) at (0,-5) [block-notarized] {};

                \draw [link] (B11) -- (B0);
                \draw [link] (B21) -- (B11);
                \draw [link] (B31) -- (B21);
                \draw [link] (B41) -- (B31);
                \draw [link] (B51) -- (B41);
                
            \end{scope}
            
            \begin{scope}[blockchain]
            
                \draw [refsnapshot] (B11.center) to (b11);
                \draw [refsnapshot] (B21.center) to (b22);
                \draw [refsnapshot] (B31.center) to (b31);
                \draw [refsnapshot] (B41.center) to (b42);
                \draw [refsnapshot] (B51.center) to (b52);
                
            \end{scope}
            
        \end{tikzpicture}%
        \end{minipage}%
    }%
    \hfill%
    \subfloat[$\PIbft$ unsafe\label{fig:examples-one-protocol-unsafe-PIbft}]{%
        \begin{minipage}[t]{0.33\linewidth}\centering%
        \begin{tikzpicture}
            \small
            
            \begin{scope}[blockchain,xshift=-0.75cm]
            
                \node (label) at (0,0.25) {LC};
            
                \coordinate (b0) at (0,-0.5);
                \node (b11) at (0,-1) [block] {};
                \node (b21) at (0,-2) [block] {};
                \node (b31) at (0,-3) [block] {};
                \node (b41) at (0,-4) [block] {};
                \node (b51) at (0,-5) [block] {};
                \node (b61) at (0,-6) [block] {};

                \draw [link] (b11) -- (b0);
                \draw [link] (b21) -- (b11);
                \draw [link] (b31) -- (b21);
                \draw [link] (b41) -- (b31);
                \draw [link] (b51) -- (b41);
                \draw [link] (b61) -- (b51);
                
            \end{scope}
            
            \begin{scope}[blockchain,xshift=+0.75cm]
            
                \node (label) at (0,0.25) {BFT};
                
                \coordinate (B0) at (0,-0.5);
                \node (B11) at (0,-1) [block-notarized] {};
                \node (B21) at (-1,-2) [block-notarized] {};
                \node (B22) at (+1,-2) [block-notarized] {};
                \node (B31) at (-1,-3) [block-notarized] {};
                \node (B32) at (+1,-3) [block-notarized] {};
                \node (B41) at (-1,-4) [block-notarized] {};
                \node (B51) at (-1,-5) [block-notarized] {};

                \draw [link] (B11) -- (B0);
                \draw [link] (B21) -- (B11);
                \draw [link] (B22) -- (B11);
                \draw [link] (B31) -- (B21);
                \draw [link] (B32) -- (B22);
                \draw [link] (B41) -- (B31);
                \draw [link] (B51) -- (B41);
                
            \end{scope}
            
            \begin{scope}[blockchain]
            
                \draw [refsnapshot] (B11.center) to (b11);
                \draw [refsnapshot] (B21.center) to (b21);
                \draw [refsnapshot] (B31.center) to (b31);
                \draw [refsnapshot] (B41.center) to (b41);
                \draw [refsnapshot] (B51.center) to (b51);
                
                \draw [refsnapshot,bend right=20] (B22.center) to (b11);
                \draw [refsnapshot,bend right=35] (B32.center) to (b31);
                
            \end{scope}
            
        \end{tikzpicture}%
        \end{minipage}%
    }%
    \caption[]{%
        Snapshots are depicted as arrows
        (\tikz{ \draw [blockchain,refsnapshot,bend right=0,shorten <=0pt] (0,0) to (-2,0); }).
        \subref{fig:examples-one-protocol-unsafe-allsafe}
        Safe $\PIlc$ and $\PIbft$ means
        $\Chlc{}{}$ and $\Chbft{}{}$ do not fork.
        \subref{fig:examples-one-protocol-unsafe-PIlc}
        Forking in $\Chlc{}{}$ is absorbed by safe $\PIbft$.
        \subref{fig:examples-one-protocol-unsafe-PIbft}
        Safe $\PIlc$ renders forking in $\Chbft{}{}$ inconsequential.
    }
    \label{fig:examples-one-protocol-unsafe} 
\end{figure}

To understand how
$\LOGfin{}{}$ can be safe even if $\PIlc$ is unsafe
(\ie, under network partition)
or how
$\LOGda{}{}$ can be safe even if $\PIbft$ is unsafe
(\ie, when $n/3 < f < n/2$),
consider the following two examples (Figure~\ref{fig:examples-one-protocol-unsafe}).
During a network partition,
$\LOGlc{}{}$, the ledger output by $\PIlc$,
can be unsafe (Figure~\ref{fig:examples-one-protocol-unsafe-PIlc}).
Thus, snapshots taken by different nodes
or at different times can conflict.
However, $\PIbft$ is still safe
and thus orders these snapshots linearly.
Any transactions invalidated by conflicts
are sanitized during ledger extraction.
As a result, $\LOGfin{}{}$ remains safe.
In a synchronous network with $n/3 < f < n/2$,
$\PIlc$ and thus $\LOGlc{}{}$ is safe.
Even if $\PIbft$ is unsafe (Figure~\ref{fig:examples-one-protocol-unsafe-PIbft}),
finalization of a snapshot requires at least
one honest vote, and thus only valid snapshots
become finalized.
Since finalized snapshots are consistent,
$\LOGfin{}{}$ is consistent with $\LOGlc{}{}$.
Thus, prefixing $\LOGlc{}{}$ with $\LOGfin{}{}$
to form $\LOGda{}{}$ does not introduce inconsistencies,
and $\LOGda{}{}$ remains safe.

\subsection{Other BFT Sub-Protocols}
\label{sec:other-bft-protocols}

In the example of Section~\ref{sec:sleepy-streamlet},
Streamlet is readily replaced with other BFT sub-protocols for $\PIbft$,
such as HotStuff \cite{yin2018hotstuff} or PBFT \cite{pbft}.
Furthermore, the analysis of Section~\ref{sec:analysis}
carries over with minor alterations
and the security Theorem~\ref{thm:main-security} holds
for these variants as well.
The necessary modifications are described in the following.

\subsubsection{HotStuff}

\begin{algorithm}
    \caption{\footnotesize Pseudocode of example \sac construction with HotStuff as $\PIbft$ and a longest-chain protocol as $\PIlc$}
    \label{algo:pseudocode-sleepy-hotstuff}
    \begin{algorithmic}[1]
        \footnotesize
        \Procedure{Main}{\null}
            \State $\mathcal Q \gets \emptyset$
            \For{time slot $t \gets 1, 2, 3, ...$}
                \State \Call{EchoIncomingNetworkMessages}{\null}
                    \label{algo:pseudocode-sleepy-hotstuff-echo-msgs}
                \State $\mathcal M \gets \Call{GetIncomingNetworkMessages}{\null}$
                \State $\mathcal M_{\mathrm{lc}} \gets \operatorname{FilterForLcMessages}(\mathcal M)$
                \State $\mathcal M_{\mathrm{bft}} \gets \operatorname{FilterForBftMessages}(\mathcal M)$
                \State \Call{LcProcessNetworkMessages}{$\mathcal M_{\mathrm{lc}}$}
                    \label{algo:pseudocode-sleepy-hotstuff-lcmsgs-process}
                \State \Call{LcSlot}{$t$}
                \State $\Chlc{}{t} \gets \Call{LcConfirmedChain}{\null}$
                \State $\mathcal Q \gets \mathcal Q \cup \mathcal M_{\mathrm{bft}}$
                    \label{algo:pseudocode-sleepy-hotstuff-bftmsgs-queue}
                \State $\mathcal M_{\mathrm{bft},0} \gets \left\{ m \in \mathcal Q \,\middle\vert\,
                    \begin{varwidth}[c]{4cm}
                    \centering
                    $\operatorname{IsInLocalView}(m.\mathrm{node})$ \\
                    $\land\: m.\mathrm{node}.\PayloadCh \preceq \Chlc{}{t}$
                    \end{varwidth}
                \right\}$
                    \label{algo:pseudocode-sleepy-hotstuff-bftmsgs-process}
                \State $\mathcal Q \gets \mathcal Q \setminus \mathcal M_{\mathrm{bft},0}$
                \State \Call{BftProcessNetworkMessages}{$\mathcal M_{\mathrm{bft},0}$}
                \State \Call{ChainedHotstuffBftSlot}{$t$}
                \State $\Chbft{}{t} \gets \Call{BftFinalChain}{\null}$
            \EndFor
        \EndProcedure
    \end{algorithmic}
\end{algorithm}

Two minor modifications suffice to use HotStuff as $\PIbft$ in the
example of Section~\ref{sec:sleepy-streamlet}.

\paragraph{Snapshots as Payload}
To use HotStuff for $\PIbft$,
a HotStuff block $B$ contains a snapshot $B.\PayloadCh$ as payload.
Whenever the output of $\PIlc$ updates,
an honest leader $i$ takes a snapshot of its $\Chlc{i}{t}$
and proposes it in a HotStuff block.

\paragraph{Side information about $\PIlc$}
To ensure that honest nodes only vote for \blocksbft
of which the payload snapshot is viewed as confirmed in $\PIlc$,
we piggy-back on the following provision
(terminology adapted to that of this paper):
`During the protocol, a [node] [processes] a message only after the [chain] [identified] by the [block] is already in its local tree. [...] For brevity, these details are also omitted from the pseudocode.' \cite[Section 4.2]{yin2018hotstuff}
We add the condition that a node processes a message only
after the snapshot contained in the block referred to by the message
is viewed as confirmed.
We explicate the resulting queueing mechanism
as pseudocode in Algorithm~\ref{algo:pseudocode-sleepy-hotstuff}.

Messages for $\PIlc$ are unaffected by the changes
(line~\ref{algo:pseudocode-sleepy-hotstuff-lcmsgs-process}).
Messages for $\PIbft$ are queued in $\mathcal Q$
(line~\ref{algo:pseudocode-sleepy-hotstuff-bftmsgs-queue})
and only processed by $\PIbft$ once the blocks that are referred to by the message
are in view \emph{and the payload snapshot is viewed as confirmed}
(line~\ref{algo:pseudocode-sleepy-hotstuff-bftmsgs-process}).
Intuitively, for honest proposals soon after $\GST$
this leads to a delay of at most $\Delta$
until the \blockslc, which confirm the honest proposer's snapshot,
are received by all honest nodes, and thus the proposal
is considered for voting by all honest nodes.
Hence, liveness is unaffected.
On the other hand, adversarial proposals containing an unconfirmed snapshot
will look like tardy or missing proposals to HotStuff,
an adversarial behavior in the face of which HotStuff
remains safe.
Hence, safety is unaffected.
Proof of security follows the same structure outlined in Section \ref{sec:analysis}.
A detailed analysis with security proofs can be found in
\begin{onlyonarxiv}
Appendix~\ref{sec:appendix-hotstuff-analysis}.
\end{onlyonarxiv}
\begin{onlyinproceedings}
\cite[Appendix~D]{neu2020ebbandflow}.
\end{onlyinproceedings}

\subsubsection{PBFT and Other Propose-and-Vote Protocols}

Conceptually, the same adaptation as for HotStuff can be used to employ
one of the variety of propose-and-vote BFT protocols for $\PIbft$,
even ones from the pre-blockchain era.
Consider, \eg, PBFT \cite{pbft}.
PBFT is not blockchain-based, instead, it outputs a ledger
of client requests which are denoted by $m$.
To use PBFT as $\PIbft$ in the
example of Section~\ref{sec:sleepy-streamlet},
client requests are replaced by snapshots, $m \triangleq \PayloadCh$.
Whenever the output of $\PIlc$ updates,
an honest leader $i$ takes a snapshot of its $\Chlc{i}{t}$
and starts the
three-phase protocol that constitutes the core of PBFT
to atomically multicast the snapshot to the other nodes.
Honest clients queue
the messages \textsc{pre-prepare}, \textsc{prepare} and \textsc{commit},
which contain a snapshot as payload, and only processes them
once the snapshot is locally viewed as confirmed -- again,
conceptually similar to the adaptation for HotStuff.
The processing of the remaining messages is unaltered.
For PBFT, the output $\Chbft{i}{t}$ is not a blockchain
but still a sequence of snapshots
of the output of $\PIlc$.
Thus, the ledger extraction
(Section~\ref{sec:ledger-extraction-details})
carries over
readily.

Again, intuitively, as for HotStuff,
for honest proposals soon after $\GST$
the queueing of protocol messages leads to a delay of at most $\Delta$
until the \blockslc, which confirm the honest proposer's snapshot,
are received by all honest nodes, and thus the proposal
is considered for voting by all honest nodes.
Hence, liveness is unaffected.
On the other hand, adversarial proposals containing an unconfirmed snapshot
will look like tardy or missing proposals to PBFT,
an adversarial behavior in the face of which PBFT
remains safe.
Hence, safety is unaffected.

\section{Simulation Experiments}
\label{sec:simulations}

To give the reader some insight into the dynamics of the \eaf construction,
we simulate it in the presence of intermittent network partitions
and under dynamic participation of nodes.\footnote{%
The code of our simulations can be found here:
\url{https://github.com/tse-group/ebb-and-flow}}
The adversary attempts to prevent liveness for as long as possible,
\eg, by launching a private chain attack on $\PIlc$ after a partition using blocks
pre-mined during the partition, or by refusing to participate in $\PIbft$.

\paragraph{Setup}

We simulate a system of $n=100$ nodes,
$f=25$ of which adversarial.
Network messages are delayed by $\Delta=\SI{1}{\second}$.
For Sleepy, $\lambda = \SI{1e-1}{\per\second}$,
so that each node produces blocks at rate
$\lambda_0 = \lambda/n = \SI{1e-3}{\per\second}$.
One lottery slot takes $\SI{1}{\second}$. %
\Blockslc are \finalda if $k=20$ deep.
Streamlet uses $\Delta_{\mathrm{bft}}=\SI{5}{\second}$.
The system undergoes intermittent network partitions (as detailed below)
and dynamic participation of honest nodes (as detailed below).
At every time, a majority of at least $f+1 = 26$ honest nodes are awake.
Adversarial nodes are always awake.
We observe the length of the shortest ledgers $|\LOGda{i}{t}|$ and $|\LOGfin{i}{t}|$
observed by any honest node $i$, \ie,
\begin{IEEEeqnarray}{C}
    \min_i |\LOGda{i}{t}|
    \qquad
    \text{and}
    \qquad
    \min_i |\LOGfin{i}{t}|.
\end{IEEEeqnarray}

\paragraph{Dynamic Participation}

\begin{figure}
    \centering%
    \begin{tikzpicture}%
        \begin{axis}[
            mysimpleplot,
            name=plot1,
            ylabel={Ledger length [blks]},
            legend columns=2,
            xmin=0, xmax=3600,
            ymin=0, ymax=250,
            height=0.5\linewidth,
            xmajorticks=false,
            transpose legend,
        ]
        
            \def\DATAPREFIX{./figures/simulation-03}

            \foreach \tStart/\tStop in {255/285,1215/1410,1725/1785,1860/2100,2130/2205,2715/2745,3105/3210,3300/3405,3435/3480,3525/3600} {
                \edef\temp{\noexpand\draw [fill=myParula05Green,fill opacity=0.2,draw=none] (axis cs:\tStart,-100) rectangle (axis cs:\tStop,1000);}
                \temp
            }

            \addplot [myparula11,mark=none] table [x=t,y=l_Lp]
            {\DATAPREFIX/sim-03.dat};
            \label{leg:simulations-dynamic-availability-Lp}
            \addlegendentry{$\min_i |\LOGfin{i}{t}|$};
            
            \addplot [myparula21,mark=none] table [x=t,y=l_Lda]
            {\DATAPREFIX/sim-03.dat};
            \label{leg:simulations-dynamic-availability-Lda}
            \addlegendentry{$\min_i |\LOGda{i}{t}|$};

            \addlegendimage{black,mark=none}
            \addlegendentry{Awake honest nodes};
            
            \addlegendimage{myParula05Green,no markers}
            \addlegendentry{Liveness threshold};

            \addlegendimage{area legend,draw opacity=0,fill=myParula05Green,fill opacity=0.2,draw=none};
            \addlegendentry{Quorum live};

        \end{axis}%
        \begin{axis}[
            mysimpleplot,
            name=plot2,
            at=(plot1.below south), anchor=above north,
            xlabel={Time [s]},
            ylabel={Awake honest},
            xmin=0, xmax=3600,
            ymin=45, ymax=80,
            height=0.4\linewidth,
            ytick={50,67,75},
        ]
        
            \def\DATAPREFIX{./figures/simulation-03}

            \foreach \tStart/\tCenter/\tStop in {255/270/285,1215/1312.5/1410,1725/1755/1785,1860/1980/2100,2130/2167.5/2205,2715/2730/2745,3105/3152.5/3210,3300/3352.5/3405,3435/3462.5/3480,3525/3562.5/3600} {
                \edef\temp{\noexpand\draw [fill=myParula05Green,fill opacity=0.2,draw=none] (axis cs:\tStart,-100) rectangle (axis cs:\tStop,1000);}%
                \temp
            }

            \addplot [name path=dacurve,black,mark=none] table [x=t,y=l_awake]
            {\DATAPREFIX/sim-03.dat};
            \label{leg:simulations-dynamic-availability-awake}

            \addplot[name path=threshold,myParula05Green,no markers,domain=-1000:4600] {67};

        \end{axis}
    \end{tikzpicture}%
    \vspace{-0.75em}%
    \caption[]{In a synchronous network where
        the number of awake honest nodes is modelled by a reflected Brownian motion, %
        $\LOGda{}{}$ grows steadily over time.
        During intervals in which enough honest nodes are awake
        there is a $2/3$-quorum to advance $\LOGfin{}{}$
        so that it catches up with $\LOGda{}{}$.}
    \label{fig:simulations-dynamic-availability}
\end{figure}
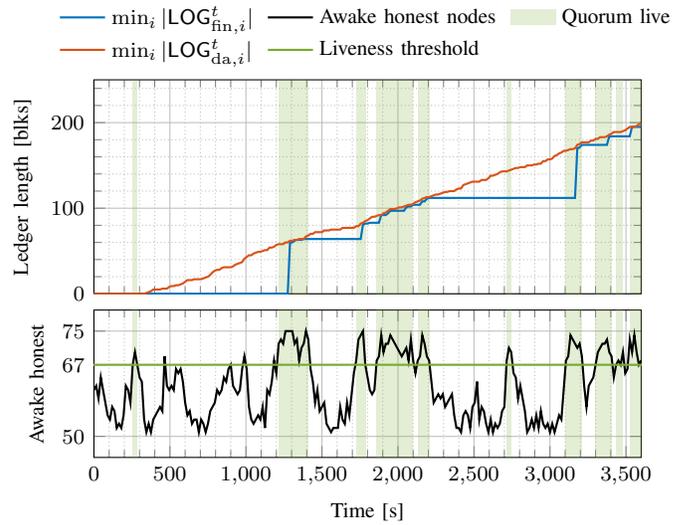

We examine the effect of dynamic participation of honest nodes on our construction.
For this purpose, we assume a synchronous network, \ie, $\GST = 0$.
The number of awake honest nodes follows a reflected Brownian motion between
$51$ and $75$.

Figure~\ref{fig:simulations-dynamic-availability} shows a sample path of the simulation.
$\LOGda{}{}$ grows steadily over time (because the conditions of \Ptwo are satisfied,
$\LOGda{}{}$ is secure, cf. \ref{leg:simulations-dynamic-availability-Lda})
at a rate proportional to the number of awake nodes (cf. \ref{leg:simulations-dynamic-availability-awake}).
Only during intervals when $67$ or more honest nodes are awake (shaded in Figure~\ref{fig:simulations-dynamic-availability},
recall that the adversary refuses to
participate in the protocol) there is a $2/3$-quorum to advance $\LOGfin{}{}$
(cf. \ref{leg:simulations-dynamic-availability-Lp}), whenever
conditions in Streamlet permit (\ie, whenever there is a sufficiently long sequence
of honest leaders).
During a sufficiently long such interval, $\LOGfin{}{}$ catches up with $\LOGda{}{}$.

\paragraph{Intermittent Network Partitions}

\begin{figure}
    \centering%
    \begin{tikzpicture}%
        \begin{axis}[
            mysimpleplot,
            xlabel={Time [s]},
            ylabel={Ledger length [blks]},
            legend columns=2,
            transpose legend,
            xmin=0, xmax=3600,
            ymin=0, ymax=225,
            height=0.5\linewidth,
        ]
        
            \def\DATAPREFIX{./figures/simulation-01}

            \draw[fill=black,fill opacity=0.1,draw=none] (axis cs:600,0) rectangle (axis cs:1200,500);
            \draw[fill=black,fill opacity=0.1,draw=none] (axis cs:1800,0) rectangle (axis cs:2700,500);

            \addplot [myparula11,mark=none] table [x=t,y=l_Lp]
            {\DATAPREFIX/sim-01-phase1.dat};
            \label{leg:simulations-partitions-Lp}
            \addlegendentry{$\min_i |\LOGfin{i}{t}|$};
            
            \addplot [myparula21,mark=none] table [x=t,y=l_Lda]
            {\DATAPREFIX/sim-01-phase1.dat};
            \label{leg:simulations-partitions-Lda}
            \addlegendentry{$\min_i |\LOGda{i}{t}|$};

            \addplot [myparula11,mark=none,forget plot] table [x=t,y=l_Lp]
            {\DATAPREFIX/sim-01-phase2.dat};

            \addplot [myparula22,mark=none] table [x=t,y=l_Lda_A]
            {\DATAPREFIX/sim-01-phase2.dat};
            \label{leg:simulations-partitions-Lda-A}
            \addlegendentry{$\min_{i\in P_1} |\LOGda{i}{t}|$};
            
            \addplot [myparula23,mark=none] table [x=t,y=l_Lda_B]
            {\DATAPREFIX/sim-01-phase2.dat};
            \label{leg:simulations-partitions-Lda-B}
            \addlegendentry{$\min_{i\in P_2} |\LOGda{i}{t}|$};

            \addlegendimage{area legend,draw opacity=0,fill=black,fill opacity=0.1,draw=none};
            \addlegendentry{Partitions};

            \addplot [myparula11,mark=none] table [x=t,y=l_Lp]
            {\DATAPREFIX/sim-01-phase3.dat};
            \addplot [myparula21,mark=none] table [x=t,y=l_Lda]
            {\DATAPREFIX/sim-01-phase3.dat};

            \addplot [myparula11,mark=none] table [x=t,y=l_Lp]
            {\DATAPREFIX/sim-01-phase4.dat};
            \addplot [myparula22,mark=none] table [x=t,y=l_Lda_A]
            {\DATAPREFIX/sim-01-phase4.dat};
            \addplot [myparula23,mark=none] table [x=t,y=l_Lda_B]
            {\DATAPREFIX/sim-01-phase4.dat};

            \addplot [myparula11,mark=none] table [x=t,y=l_Lp]
            {\DATAPREFIX/sim-01-phase5.dat};
            \addplot [myparula21,mark=none] table [x=t,y=l_Lda]
            {\DATAPREFIX/sim-01-phase5.dat};

        \end{axis}
    \end{tikzpicture}%
    \vspace{-0.75em}%
    \caption[]{Under intermittent network partitions, during which
        honest nodes are split into two parts  of $2(n-f)/3$ and $(n-f)/3$ nodes, respectively,
        \finalizationp of \blocksbft stalls because no $2/3$-quorum is live.
        The ledgers $\LOGda{}{}$ as seen by the different parts drift apart.
        Once the network reunites, the honest nodes converge on the longer $\LOGda{}{}$
        and $\LOGfin{}{}$ catches up.}
    \label{fig:simulations-partitions}
\end{figure}
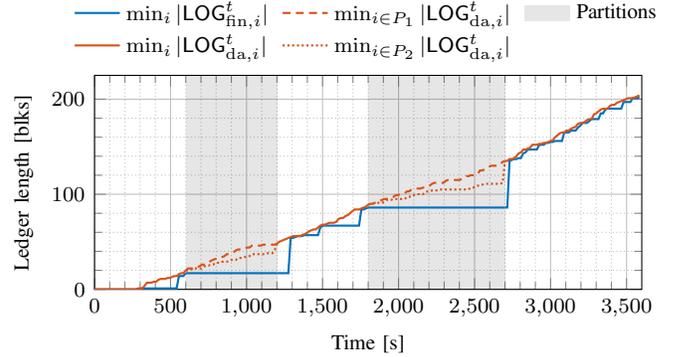

We simulate the system under intermittent network partitions, during which
honest nodes are split into two parts $P_1$ and $P_2$ of $2(n-f)/3$ and $(n-f)/3$ nodes, respectively.
Inter-part communication is prevented, intra-part communication incurs
delay $\Delta$.
All honest nodes are awake throughout the experiment.
During partitions %
we consider
the ledgers as seen by honest nodes in the respective parts.

Figure~\ref{fig:simulations-partitions} shows a sample path of the simulation.
Periods of network partition are shaded in Figure~\ref{fig:simulations-partitions}.
As expected, \finalizationp of \blocksbft stalls during periods of partition (cf. \ref{leg:simulations-partitions-Lp}),
because no $2/3$-quorum consensus is achieved, as communication between parts is blocked.
The ledgers $\LOGda{}{}$ as seen by nodes in the different parts $P_1$ and $P_2$ drift apart (cf. \ref{leg:simulations-partitions-Lda-A}, \ref{leg:simulations-partitions-Lda-B}).
Once the network reunites, the honest nodes converge on the longer $\LOGda{}{}$
(which is that produced by part $P_1$, cf. \ref{leg:simulations-partitions-Lda})
and $\LOGfin{}{}$ quickly catches up with $\LOGda{}{}$.
Note that the shorter $\LOGda{}{}$ (produced by part $P_2$) is abandoned
and disappears from $\LOGda{}{}$ after the partition.

Note that because part $P_1$ outnumbers the adversary,
the adversary does not have a chance to build a long enough private chain that it can use to
delay honest nodes' convergence on $\LOGda{}{}$.
Instead, convergence on $\LOGda{}{}$ is reached once honest nodes have synchronized
their blocktrees and picked the longest chain.
This is different if honest nodes are partitioned into smaller parts, as examined next.

\paragraph{Convergence of $\LOGdaBLANKFIX$ After Network Partition and/or Low Participation}

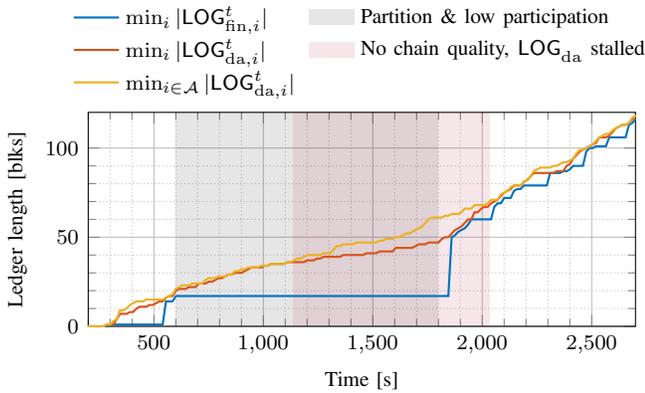
\begin{figure}
    \centering%
    \begin{tikzpicture}%
        \begin{axis}[
            mysimpleplot,
            xlabel={Time [s]},
            ylabel={Ledger length [blks]},
            legend columns=3,
            transpose legend,
            xmin=200, xmax=2700,
            ymin=0, ymax=120,
            height=0.5\linewidth,
        ]
        
            \def\DATAPREFIX{./figures/simulation-02}

            \draw[fill=black,fill opacity=0.1,draw=none] (axis cs:600,0) rectangle (axis cs:1800,500);

            \draw[fill=myParula07Red,fill opacity=0.1,draw=none] (axis cs:1135,0) rectangle (axis cs:2035,500);

            \addplot [myparula11,mark=none] table [x=t,y=l_Lp]
            {\DATAPREFIX/sim-02-phases.dat};
            \label{leg:simulations-convergence-Lp}
            \addlegendentry{$\min_i |\LOGfin{i}{t}|$};
            
            \addplot [myparula21,mark=none] table [x=t,y=l_Lda]
            {\DATAPREFIX/sim-02-phases.dat};
            \label{leg:simulations-convergence-Lda}
            \addlegendentry{$\min_i |\LOGda{i}{t}|$};

            \addplot [myparula31,mark=none] table [x=t,y=l_Lda_adv]
            {\DATAPREFIX/sim-02-phases.dat};
            \label{leg:simulations-convergence-Lda-adv}
            \addlegendentry{$\min_{i\in\Adv} |\LOGda{i}{t}|$};

            \addlegendimage{area legend,draw opacity=0,fill=black,fill opacity=0.1,draw=none};
            \addlegendentry{Partition \& low participation};
            
            \addlegendimage{area legend,draw opacity=0,fill=myParula07Red,fill opacity=0.1,draw=none};
            \addlegendentry{No chain quality, $\LOGda{}{}$ stalled};

        \end{axis}
    \end{tikzpicture}%
    \vspace{-0.75em}%
    \caption[]{During a period of network partition and low participation
        honest block production slows down
        and the adversary can successfully pre-mine a private adversarial structure.
        The adversary releases private blocks to displace
        honest blocks from the longest chain.
        Thus, the longest chain suffers from low chain quality
        and the dynamic ledger $\LOGda{}{}$ stalls.
        Once the network reunites and all honest nodes awake,
        $\LOGda{}{}$ grows at a fast rate and
        the adversary eventually runs out of pre-mined blocks.
        Honest blocks enter the longest chain and liveness of $\LOGda{}{}$ and with it
        liveness of $\LOGfin{}{}$ ensues.}
    \label{fig:simulations-convergence}
\end{figure}

We focus on the convergence of $\LOGda{}{}$ after a network partition and/or period of
low participation, where the largest awake part is smaller than $f$.
For this purpose, suppose that during a partition, $50$ of the $75$ honest nodes are asleep.
The remaining $25$ nodes are awake but partitioned into two parts of $15$ and $10$
nodes, respectively.
Thus, the largest awake part with $15$ honest nodes is smaller than $f = 25$
and the adversary can successfully pre-mine a private chain during the period
of partition and low participation.
As before, only inter-part communication is prohibited.

Figure~\ref{fig:simulations-convergence} shows a sample path of the simulation.
Before the partition, honest block production is fast and the adversary cannot build a substantial
private chain.
During the period
of partition and low participation
the honest block production slows down.
The adversary gains a considerable lead
in that its private chain
(cf. \ref{leg:simulations-convergence-Lda-adv})
grows much faster than the longest chain in honest view
(cf. \ref{leg:simulations-convergence-Lda}).
As honest nodes produce blocks, the adversary releases its withheld adversarial blocks to displace
the honest blocks from the longest chain.
As a result, the longest chain suffers a sustained period of low chain quality (all blocks
are adversarial and thus might not include any transactions)
and the dynamic ledger $\LOGda{}{}$ effectively stalls.
Once the network reunites and asleep honest nodes awake,
all honest nodes
join forces on $\LOGda{}{}$, which now grows at a fast rate again.
Eventually, the adversary runs out of pre-mined blocks and cannot displace honest blocks any longer.
An honest block enters the longest chain and liveness of $\LOGda{}{}$ and with it
liveness of $\LOGfin{}{}$ ensues (cf. \ref{leg:simulations-convergence-Lp}).

Note that during the period of partition and low participation,
$\LOGfin{}{}$ does not grow because (as in the previous experiment) no $2/3$-quorum consensus is achieved.
Once the network reunites, $\LOGfin{}{}$ catches up with $\LOGda{}{}$, but since the most recent blocks in $\LOGda{}{}$
are adversarial (and thus potentially empty), neither $\LOGda{}{}$ nor $\LOGfin{}{}$ are live for some time.
Once honest blocks return to $\LOGda{}{}$ and get referenced by $\LOGfin{}{}$, both return to be live.

\section{Discussion and Conclusion}
\label{sec:discussion}

\subsection{\SAC Protocols and Finality Gadgets}
\label{sec:finality-gadgets}

Finality gadgets, initiated by \cite{buterin2017casper}, are a body of work
\cite{fg_nightshade,dinsdale2019afgjort,ethresearch-hierarchical-finality-gadget,stewart2020grandpa}
that aims to add finality to a Nakamoto-style protocol.
As far as we can gather, there is no mathematical definition of  a finality gadget; indeed different works have different goals on what their finality gadgets are supposed to achieve, and these goals are often not explicitly spelled out. For example, \cite{dinsdale2019afgjort} seems to be using their finality gadget to achieve opportunistic responsiveness. On the other hand,  the goals of \cite{stewart2020grandpa} seem to be aligned with the \eaf property we studied here, but there is no mathematical formulation on what should be achieved.  In contrast, we focus on the \eaf property, precisely define what it means, and construct \sac protocols to achieve the property. So it is difficult to have a scientific comparison between \sac protocols and finality gadgets. However, there is one important {\em structural} difference between the construction of \sac protocols and the construction of all existing finality gadgets which we want to point out.
The difference is that the \sac protocol construction can use any off-the-shelf dynamically available protocol unmodified (and the BFT sub-protocol with minor modifications), while all existing 
finality gadgets involve a joint design of the finality voting and the fork choice rule of 
the underlying Nakamoto-style chain. In particular, the native fork choice rule of the Nakamoto-style chain has to be altered to accommodate the finalization process.
In Casper FFG \cite{buterin2017casper}, for example, the `correct by construction' rule specifies that blocks should be proposed on the chain with the highest justified block, as opposed to the longest chain.
Another example is the hierarchical finality gadget  \cite{ethresearch-hierarchical-finality-gadget}, which specifies that proposal should be done on the chain with the deepest finalized block.
In contrast, the dynamically available sub-protocol in our construction is off-the-shelf and so the fork choice rule as well as the confirmation rule are unaltered. Finalization by the BFT sub-protocol occurs {\em after} transactions are confirmed in the $\LOGlc{}{}$ ledger. The confirmation and the finalization properties are completely {\em decoupled}.
The decoupled nature of our construction has several advantages. First, construction adds finality to {\em any} existing dynamically available chain without change. Second, our construction allows the use of state-of-the-art dynamically available protocols and state-of-the-art partially synchronous BFT protocols without the need to reinvent the wheel. In contrast, existing finality gadget designs entail handcrafting brand new protocols
(\eg, \cite{dinsdale2019afgjort,stewart2020grandpa}),
and the tight coupling between the two layers makes reasoning about security difficult in the design process. The attack on Gasper in Section~\ref{sec:attacks-gasper} is a good example of the perils of this approach. Another example is the {\em bouncing attack} on Casper FFG \cite{ethresearch-bouncing-attack,ethresearch-bouncing-attack-analysis}
\begin{onlyonarxiv}
(recapitulated in Appendix~\ref{sec:attacks-casper}).
\end{onlyonarxiv}
\begin{onlyinproceedings}
(recapitulated in \cite[Appendix~E]{neu2020ebbandflow}).
\end{onlyinproceedings}
Third, our construction is `future-proof' because it can take advantage of future advances in the design of dynamically available protocols and in the design of  partially synchronous BFT protocols; both problems have received and are continuing to receive significant attention from the community.
\subsection{\EAF and \SAC for Proof-of-Work}
\label{sec:discussion-pow}

Another common goal for finality gadgets
is to add a permissioned finality layer
to a permissionless PoW Na\-ka\-mo\-to-style protocol.
In this setting,
nodes come in two flavors:
\emph{miners} are quantified by hash rate and power the PoW longest chain,
and
\emph{validators} with unique cryptographic identities provide finality.
The two different resources (hash rate and cryptographic identities)
require a
modification of the environment
in
the Theorem of Section~\ref{sec:abstract-eaf-construction}.
\begin{onlyinproceedings}
The \sac construction
using Nakamoto's PoW longest chain as $\PIlc$
and any of the BFT protocols from Sections~\ref{sec:sleepy-streamlet} and \ref{sec:other-bft-protocols} as $\PIbft$
satisfies an \eaf variant
where
the adversarial hash rate is less than
$50\%$
of the awake hash rate for both \textbf{P1} and \textbf{P2},
and
the number of adversarial validators is fewer than
$33\%$
and
$66\%$
of the validators for \textbf{P1} and \textbf{P2}, respectively.
Security is proved analogously
to the permissioned case (Section~\ref{sec:analysis}).
For further discussion, see %
\cite[Section~\ref{sec:discussion-pow}]{neu2020ebbandflow}.
\end{onlyinproceedings}
\begin{onlyonarxiv}
The \sac construction
using Nakamoto's PoW longest chain as $\PIlc$
and any of the BFT protocols from Sections~\ref{sec:sleepy-streamlet} and \ref{sec:other-bft-protocols} as $\PIbft$
satisfies the following \eaf variant:
\begin{theorem*}[Informal, Permissioned Finality for Permissionless PoW]
Consider
a network environment where:
\begin{enumerate}
    \item Communication is asynchronous until a global stabilization time $\GST$ after which communication becomes synchronous, 
    \item validators are always awake, and
    \item miners can sleep and wake up at any time.
\end{enumerate}
Then
\begin{enumerate}
    \item ({\bf P1 - Finality}): The finalized ledger $\LOGfin{}{}$ is guaranteed to be safe at all times, and live after $\GST$, if fewer than $33\%$ of validators are adversarial, \emph{fewer than $50\%$ of awake hash rate is adversarial, and awake hash rate is bounded away from zero}.
    \item ({\bf P2 - Dynamic Availability}): If $\GST = 0$, then the available ledger $\LOGda{}{}$ is guaranteed to be safe and live at all times, if \emph{fewer than $66\%$ of validators are adversarial}, fewer than $50\%$ of awake hash rate is adversarial, and awake hash rate is bounded away from zero.
\end{enumerate}
\end{theorem*}
Security is proved analogously to Section~\ref{sec:analysis}.
Observe that the $33\%$ bound on adversarial validators
under \textbf{P1}
and
the $50\%$ bound on adversarial hash rate
under \textbf{P2}
are analogous to the permissioned
case analyzed in this paper.
The \sac construction
for permissioned finality on top of permissionless PoW Nakamoto
has further requirements.
Under \textbf{P1}, fewer than $50\%$ of awake hash rate have to be adversarial, and awake hash rate has to be bounded away from zero, since otherwise $\PIlc$ might not be live and there would be no input for $\PIbft$ to finalize -- then $\LOGfin{}{}$ would not be live.
Similarly, under \textbf{P2}, fewer than $66\%$ of validators have to be adversarial, since otherwise adversarial validators could finalize $\PIlc$ blocks that are not on the longest chain, which would later enter the prefix of $\LOGda{}{}$ -- then $\LOGda{}{}$ would not be safe.
It remains an open question whether these additional requirements
are fundamental or a limitation
of the \sac construction.
\end{onlyonarxiv}

\subsection{\SAC Protocols for Ethereum 2.0}
\label{sec:sac-for-eth2}

Our construction yields provably secure \eaf protocols
from off-the-shelf sub-protocols and provides
a flexible resolution
of the availability-finality dilemma.
In addition,
the composition enables us
to benefit from advances
in the design of sub-protocols
and to pass along
(rather than having to build from scratch)
additional features of
the constituent protocols
which are desired from
a decentralized Internet-scale open-participation
consensus infrastructure
such as Ethereum.

\paragraph{Scalability to Many Nodes}

The partially synchronous BFT sub-protocol $\PIbft$ used in the \sac construction presents the main
scalability bottleneck.
HotStuff is the BFT protocol with the lowest known
message complexity
$O(n)$.
When used alongside a longest-chain-based
protocol, which are known to scale well to
many
participants, the overall \sac protocol
promises good scalability.

\paragraph{Accountability}

Gasper \cite{buterin2020combining} provides accountability in the form that
a safety violation implies that at least a third of
nodes have provably violated the protocol.
As a
punitive and deterrent response, those nodes'
stake is slashed.
This attaches a price tag to safety violations
and leads to notions of economic security.
\Sac protocols
inherit accountability properties from the BFT sub-protocol $\PIbft$
for the finalized ledger $\LOGfin{}{}$.
For instance,
for many partially synchronous BFT protocols following
the propose-and-vote paradigm, such as HotStuff, PBFT
or Streamlet,
a safety violation
requires equivocating votes from more than a third of the nodes.
(Recall that this fact is the cornerstone of these protocols' safety argument.)
Due to the use of digital signatures,
equivocating votes can be attributed to nodes irrefutably,
and equivocating nodes can be held accountable
for the safety violation
(\cf \cite{neu2020snapandchat,sheng2020bft})
, \eg, by slashing the nodes' stake.
To what extent accountability can be provided for
the available ledger $\LOGda{}{}$
is less clear at this point, both because accountability
has not been widely studied in the context of
dynamically available protocols, as well as due to the
non-trivial ledger extraction that leads to $\LOGda{}{}$.
\paragraph{High Throughput}
High transaction throughput can be achieved
by choosing a high throughput $\PIlc$,
such as a longest chain protocol
with separate transaction and backbone blocks
(\cf Prism \cite{prism}) or OHIE \cite{ohie} or ledger combiners \cite{fitzi2020ledger}.

\paragraph{Fast Confirmation Latency}

Using
ledger combiners \cite{fitzi2020ledger} or Prism \cite{prism}
for $\PIlc$, fast latency, in particular, latency independent of the
confirmation error probability,
can be achieved by \sac protocols.
For $\PIbft$, responsive BFT protocols can be used which
finalize snapshots with a latency in the order of the actual network
delay
rather than the delay bound $\Delta$.
Hence, $\PIbft$ does not present a bottleneck in terms of reducing the latency
of \sac protocols
and the finalized ledger $\LOGfin{}{}$ can catch up with the available ledger $\LOGda{}{}$ very quickly, when network conditions allow.

\section*{Acknowledgment}
We thank Yan X.\ Zhang, Danny Ryan and Vitalik Buterin for fruitful discussions.
JN is supported by the Reed-Hodgson Stanford Graduate Fellowship.
ENT is supported by the
Stanford Center for Blockchain Research.

\nocite{nakamoto_paper}
\bibliographystyle{IEEEtran}
\bibliography{references}

%
\begin{thebibliography}{10}
\providecommand{\url}[1]{#1}
\csname url@samestyle\endcsname
\providecommand{\newblock}{\relax}
\providecommand{\bibinfo}[2]{#2}
\providecommand{\BIBentrySTDinterwordspacing}{\spaceskip=0pt\relax}
\providecommand{\BIBentryALTinterwordstretchfactor}{4}
\providecommand{\BIBentryALTinterwordspacing}{\spaceskip=\fontdimen2\font plus
\BIBentryALTinterwordstretchfactor\fontdimen3\font minus
  \fontdimen4\font\relax}
\providecommand{\BIBforeignlanguage}[2]{{%
\expandafter\ifx\csname l@#1\endcsname\relax
\typeout{** WARNING: IEEEtran.bst: No hyphenation pattern has been}%
\typeout{** loaded for the language `#1'. Using the pattern for}%
\typeout{** the default language instead.}%
\else
\language=\csname l@#1\endcsname
\fi
#2}}
\providecommand{\BIBdecl}{\relax}
\BIBdecl

\bibitem{neu2020ebbandflow}
J.~Neu, E.~N. Tas, and D.~Tse, ``Ebb-and-flow protocols: A resolution of the
  availability-finality dilemma,'' \emph{Preprint, arXiv:2009.04987}, 2020.

\bibitem{nakamoto_paper}
S.~Nakamoto, ``Bitcoin: A peer-to-peer electronic cash system,''
  \url{https://bitcoin.org/bitcoin.pdf}, 2008.

\bibitem{sleepy}
R.~Pass and E.~Shi, ``The sleepy model of consensus,'' in \emph{{ASIACRYPT}
  {(2)}}, ser. LNCS.\hskip 1em plus 0.5em minus 0.4em\relax Springer, 2017, pp.
  380--409.

\bibitem{snowwhite}
P.~Daian, R.~Pass, and E.~Shi, ``{Snow White}: Robustly reconfigurable
  consensus and applications to provably secure proof of stake,'' in
  \emph{Financial Cryptography}, ser. LNCS.\hskip 1em plus 0.5em minus
  0.4em\relax Springer, 2019, pp. 23--41.

\bibitem{david2018ouroboros}
B.~David, P.~Gazi, A.~Kiayias, and A.~Russell, ``{Ouroboros Praos}: An
  adaptively-secure, semi-synchronous proof-of-stake blockchain,'' in
  \emph{{EUROCRYPT} {(2)}}, ser. LNCS.\hskip 1em plus 0.5em minus 0.4em\relax
  Springer, 2018, pp. 66--98.

\bibitem{badertscher2018ouroboros}
C.~Badertscher, P.~Gazi, A.~Kiayias, A.~Russell, and V.~Zikas, ``{Ouroboros
  Genesis}: Composable proof-of-stake blockchains with dynamic availability,''
  in \emph{{CCS}}.\hskip 1em plus 0.5em minus 0.4em\relax {ACM}, 2018, pp.
  913--930.

\bibitem{posat}
S.~Deb, S.~Kannan, and D.~Tse, ``{PoSAT}: Proof-of-work availability and
  unpredictability, without the work,'' in \emph{To appear in Financial
  Cryptography, arXiv:2010.08154}, 2021.

\bibitem{pbft}
M.~Castro and B.~Liskov, ``Practical {Byzantine} fault tolerance,'' in
  \emph{{OSDI}}.\hskip 1em plus 0.5em minus 0.4em\relax {USENIX} Association,
  1999, pp. 173--186.

\bibitem{tendermint_thesis}
E.~Buchman, ``Tendermint: Byzantine fault tolerance in the age of
  blockchains,'' Master's thesis, University of Guelph, 2016.

\bibitem{tendermint}
E.~Buchman, J.~Kwon, and Z.~Milosevic, ``The latest gossip on {BFT}
  consensus,'' \emph{Preprint, arXiv:1807.04938}, 2018.

\bibitem{yin2018hotstuff}
M.~Yin, D.~Malkhi, M.~K. Reiter, G.~Golan{-}Gueta, and I.~Abraham,
  ``{HotStuff}: {BFT} consensus with linearity and responsiveness,'' in
  \emph{{PODC}}.\hskip 1em plus 0.5em minus 0.4em\relax {ACM}, 2019, pp.
  347--356.

\bibitem{streamlet}
B.~Y. Chan and E.~Shi, ``Streamlet: Textbook streamlined blockchains,'' in
  \emph{{AFT}}.\hskip 1em plus 0.5em minus 0.4em\relax {ACM}, 2020, pp. 1--11.

\bibitem{libraBFT}
{Libra Association}, ``White paper,''
  \url{https://libra.org/en-US/white-paper/}, 2020.

\bibitem{baudet2018state}
M.~Baudet, A.~Ching, A.~Chursin, G.~Danezis, F.~Garillot, Z.~Li, D.~Malkhi,
  O.~Naor, D.~Perelman, and A.~Sonnino, ``State machine replication in the
  {Libra} blockchain,'' Report, Libra Association, 2018.

\bibitem{chen2016algorand}
J.~Chen and S.~Micali, ``Algorand,'' \emph{Preprint, arXiv:1607.01341}, 2016.

\bibitem{gilad2017algorand}
Y.~Gilad, R.~Hemo, S.~Micali, G.~Vlachos, and N.~Zeldovich, ``{Algorand}:
  Scaling {Byzantine} agreements for cryptocurrencies,'' in
  \emph{{SOSP}}.\hskip 1em plus 0.5em minus 0.4em\relax {ACM}, 2017, pp.
  51--68.

\bibitem{cap}
S.~Gilbert and N.~A. Lynch, ``Brewer's conjecture and the feasibility of
  consistent, available, partition-tolerant web services,'' \emph{{SIGACT}
  News}, vol.~33, no.~2, pp. 51--59, 2002.

\bibitem{lewispye2020resource}
A.~Lewis-Pye and T.~Roughgarden, ``Resource pools and the {CAP} theorem,''
  \emph{Preprint, arXiv:2006.10698}, 2020.

\bibitem{PS_partition}
Y.~Guo, R.~Pass, and E.~Shi, ``Synchronous, with a chance of partition
  tolerance,'' in \emph{{CRYPTO} {(1)}}, ser. LNCS.\hskip 1em plus 0.5em minus
  0.4em\relax Springer, 2019, pp. 499--529.

\bibitem{flexibleBFT}
D.~Malkhi, K.~Nayak, and L.~Ren, ``Flexible {Byzantine} fault tolerance,'' in
  \emph{{CCS}}.\hskip 1em plus 0.5em minus 0.4em\relax {ACM}, 2019, pp.
  1041--1053.

\bibitem{buterin2020combining}
V.~Buterin, D.~Hernandez, T.~Kamphefner, K.~Pham, Z.~Qiao, D.~Ryan, J.~Sin,
  Y.~Wang, and Y.~X. Zhang, ``Combining {GHOST} and {Casper},'' \emph{Preprint,
  arXiv:2003.03052}, 2020.

\bibitem{buterin2017casper}
V.~Buterin and V.~Griffith, ``{Casper} the friendly finality gadget,''
  \emph{Preprint, arXiv:1710.09437}, 2017.

\bibitem{ethresearch-liveness-requirement-eth2}
\BIBentryALTinterwordspacing
V.~Buterin. (2020) Explaining the liveness guarantee (comment 8). [Online].
  Available: \url{https://ethresear.ch/t/4228/8}
\BIBentrySTDinterwordspacing

\bibitem{ryan2020}
D.~Ryan, Ethereum Foundation, Personal communication, June 2020.

\bibitem{cachin2017blockchain}
C.~Cachin and M.~Vukolić, ``Blockchain consensus protocols in the wild,''
  \emph{Preprint, arXiv:1707.01873}, 2017.

\bibitem{stewart2020grandpa}
A.~Stewart and E.~Kokoris-Kogia, ``{GRANDPA}: a {Byzantine} finality gadget,''
  \emph{Preprint, arXiv:2007.01560}, 2020.

\bibitem{katzBA}
E.~Blum, J.~Katz, and J.~Loss, ``Synchronous consensus with optimal
  asynchronous fallback guarantees,'' in \emph{{TCC} {(1)}}, ser. LNCS.\hskip
  1em plus 0.5em minus 0.4em\relax Springer, 2019, pp. 131--150.

\bibitem{ethresearch-bouncing-attack}
\BIBentryALTinterwordspacing
V.~Buterin and A.~Stewart. (2018) Beacon chain {Casper} mini-spec (comments 17
  and 19). [Online]. Available: \url{https://ethresear.ch/t/2760/17}
\BIBentrySTDinterwordspacing

\bibitem{ethresearch-bouncing-attack-analysis}
\BIBentryALTinterwordspacing
R.~Nakamura. (2019) Analysis of bouncing attack on {FFG}. [Online]. Available:
  \url{https://ethresear.ch/t/6113}
\BIBentrySTDinterwordspacing

\bibitem{ethresearch-bouncing-attack-prevention}
\BIBentryALTinterwordspacing
------. (2019) Prevention of bouncing attack on {FFG}. [Online]. Available:
  \url{https://ethresear.ch/t/6114}
\BIBentrySTDinterwordspacing

\bibitem{ethresearch-gasper-liveness-attack}
\BIBentryALTinterwordspacing
J.~Neu, E.~N. Tas, and D.~Tse. (2020) A balancing attack on {Gasper}, the
  current candidate for {Eth2}’s beacon chain. [Online]. Available:
  \url{https://ethresear.ch/t/8079}
\BIBentrySTDinterwordspacing

\bibitem{kiayias2017ouroboros}
A.~Kiayias, A.~Russell, B.~David, and R.~Oliynykov, ``{Ouroboros}: {A} provably
  secure proof-of-stake blockchain protocol,'' in \emph{{CRYPTO} {(1)}}, ser.
  LNCS.\hskip 1em plus 0.5em minus 0.4em\relax Springer, 2017, pp. 357--388.

\bibitem{model-psync}
C.~Dwork, N.~A. Lynch, and L.~J. Stockmeyer, ``Consensus in the presence of
  partial synchrony,'' \emph{J. {ACM}}, vol.~35, no.~2, pp. 288--323, 1988.

\bibitem{sbft}
G.~Golan{-}Gueta, I.~Abraham, S.~Grossman, D.~Malkhi, B.~Pinkas, M.~K. Reiter,
  D.~Seredinschi, O.~Tamir, and A.~Tomescu, ``{SBFT:} {A} scalable and
  decentralized trust infrastructure,'' in \emph{{DSN}}.\hskip 1em plus 0.5em
  minus 0.4em\relax {IEEE}, 2019, pp. 568--580.

\bibitem{fg_nightshade}
A.~Skidanov, ``Fast finality and resilience to long range attacks with proof of
  space-time and {Casper}-like finality gadget,'' \url{http://near.ai/post},
  2019.

\bibitem{dinsdale2019afgjort}
T.~Dinsdale{-}Young, B.~Magri, C.~Matt, J.~B. Nielsen, and D.~Tschudi,
  ``Afgjort: {A} partially synchronous finality layer for blockchains,'' in
  \emph{{SCN}}, ser. LNCS.\hskip 1em plus 0.5em minus 0.4em\relax Springer,
  2020, pp. 24--44.

\bibitem{ethresearch-hierarchical-finality-gadget}
\BIBentryALTinterwordspacing
R.~Nakamura. (2020) Hierarchical finality gadget. [Online]. Available:
  \url{https://ethresear.ch/t/6829}
\BIBentrySTDinterwordspacing

\bibitem{neu2020snapandchat}
J.~Neu, E.~N. Tas, and D.~Tse, ``Snap-and-chat protocols: System aspects,''
  \emph{Preprint, arXiv:2010.10447}, 2020.

\bibitem{sheng2020bft}
P.~Sheng, G.~Wang, K.~Nayak, S.~Kannan, and P.~Viswanath, ``{BFT} protocol
  forensics,'' \emph{Preprint, arXiv:2010.06785}, 2020.

\bibitem{prism}
V.~K. Bagaria, S.~Kannan, D.~Tse, G.~C. Fanti, and P.~Viswanath, ``{Prism}:
  Deconstructing the blockchain to approach physical limits,'' in
  \emph{{CCS}}.\hskip 1em plus 0.5em minus 0.4em\relax {ACM}, 2019, pp.
  585--602.

\bibitem{ohie}
H.~Yu, I.~Nikolic, R.~Hou, and P.~Saxena, ``{OHIE:} blockchain scaling made
  simple,'' in \emph{{IEEE} Symp. Secur. Privacy}, 2020, pp. 90--105.

\bibitem{fitzi2020ledger}
M.~Fitzi, P.~Gazi, A.~Kiayias, and A.~Russell, ``Ledger combiners for fast
  settlement,'' Cryptology ePrint Archive, Report 2020/675, 2020.

\bibitem{model-sync}
N.~A. Lynch, \emph{Distributed Algorithms}.\hskip 1em plus 0.5em minus
  0.4em\relax San Francisco, CA, USA: Morgan Kaufmann Publishers Inc., 1996.

\bibitem{dem20}
A.~Dembo, S.~Kannan, E.~N. Tas, D.~Tse, P.~Viswanath, X.~Wang, and O.~Zeitouni,
  ``Everything is a race and {Nakamoto} always wins,'' in \emph{{CCS}}.\hskip
  1em plus 0.5em minus 0.4em\relax {ACM}, 2020, pp. 859--878.

\end{thebibliography}

\appendices

\section{Details of the Liveness Attack on Gasper}
\label{sec:appendix-attack-gasper}

\newcommand{\Ev}[0]{\ensuremath{\mathcal E}}
\newcommand{\cEv}[0]{\ensuremath{\overline{\Ev}}}

\subsection{Setting of the Attack}
\label{sec:appendix-attack-gasper-introduction}

This appendix describes an attack on the liveness of the Gasper protocol \cite{buterin2020combining}.
We first state the assumptions about the adversary's capabilities and control over the network
that suffice for the adversary to launch our attack.
Subsequently, we describe the attack in detail.
(The attack is summarized in Section~\ref{sec:attacks-gasper}.)
Then we demonstrate
using probabilistic analysis and Monte Carlo simulation
that the adversary is likely in a position to launch the attack
within a short period of time.

\subsubsection{Goal}
\label{sec:appendix-attack-gasper-introduction-goal}

We describe an attack on the liveness of the Gasper protocol \cite{buterin2020combining}.
That is, we describe a situation which is likely to occur
and a sequence of adversarial actions
such that the adversary can prevent any Casper finalizations indefinitely.

Our exposition assumes the reader is familiar with Gasper \cite{buterin2020combining},
Casper \cite{buterin2017casper},
and the synchronous network model \cite{model-sync}.

\subsubsection{Assumptions}
\label{sec:appendix-attack-gasper-introduction-assumptions}

We assume an adversary has the following capabilities:
\begin{enumerate*}[label=(\alph*)]
    \item
        The adversary knows at what point in time honest validators
        execute the Gasper fork choice rule
        $\mathsf{HLMD}(G)$ \cite[Algorithm 4.2]{buterin2020combining}.

    \item
        The adversary is able to target a message (such as a block or a vote) for
        delivery to an honest validator just before a certain point in time.
    
    \item
        Honest validators cannot update each other
        arbitrarily quickly about messages
        they have just received.
\end{enumerate*}

Note that (a) is given by design of Gasper which has predetermined
points in time at which honest validators are supposed to cast their votes.
Conditions
(b) and (c) are satisfied in standard consensus-theoretic adversary
and network models such as $\Delta$-synchrony \cite{model-sync}
or $\Delta$-partial-synchrony \cite{model-psync}
where the adversary controls network delay.
The probabilistic liveness proof of \cite{buterin2020combining}
does not apply
because it assumes a weaker adversary
(network delays are assumed to be stochastic
rather than adversarial in \cite{buterin2020combining})
which does not have
capability (b).

\subsubsection{Terminology}
\label{sec:appendix-attack-gasper-introduction-terminology}

Recall that Gasper proceeds in epochs which are subdivided into slots.
We assume that Gasper is run with $C$ slots per epoch,
$n$ validators in total,
of which $f$ are adversarial.
Let $\beta \triangleq f/n$.
We assume that $C$ divides $n$ such that
every slot has a \emph{committee} of integer size $n/C$.
For each epoch, a random permutation of all $n$ validators
fixes the assignment of validators to committees.
The first validator in every committee is the designated
\emph{proposer} for the respective slot and gets to propose
a new block at a location in the block tree determined by $\mathsf{HLMD}(G)$.
Then, each validator of the slot's committee executes $\mathsf{HLMD}(G)$
in its own view $G$ to determine what block to vote for.

A vote consists of a GHOST vote and a Casper (FFG) vote.
The Casper vote's source and target blocks are deterministic functions
of the block the GHOST vote is cast for
(see \cite[Definition 4.7]{buterin2020combining}).
A block can only become finalized if a supermajority of $\geq 2n/3$ validators
vote for it.
The goal of the attack is to keep honest validators split between two options
(a `left' and a `right' chain, see Figure~\ref{fig:attack-gasper-overview}, p.~\pageref{fig:attack-gasper-overview})
indefinitely, such that no supermajority of $\geq 2n/3$ validators
ever votes for one of the two options and thus no block ever gets finalized.

\subsection{Attack}
\label{sec:appendix-attack-gasper-attack}

In this section we describe our attack in detail,
\cf \cite{ethresearch-gasper-liveness-attack}.
For an illustration of the attack, see Figure~\ref{fig:attack-gasper-overview}
(p.~\pageref{fig:attack-gasper-overview}).

\subsubsection{Recap: Proposing and Voting in Gasper}
\label{sec:appendix-attack-gasper-attack-recap-gasper}

To understand how the adversary can keep the honest nodes split indefinitely
between two chains it is necessary to revisit the proposing and voting
algorithms of Gasper.
For each of the two roles, proposing and voting,
honest validators use the fork chain rule
$\mathsf{HLMD}(G)$
(see \cite[Algorithm 4.2]{buterin2020combining})
in their local view $G$ to determine
(a) when proposing, what block to extend, and
(b) when voting, what block to endorse with a vote.

Roughly speaking, $\mathsf{HLMD}(G)$ does this.
First, $\mathsf{HLMD}(G)$ finds
the justified pair with highest attestation epoch
among all possible chains,
but taking into account only votes that have
already been referenced on said chain
(see \cite[Algorithm 4.2]{buterin2020combining}, line $3$, `$J(\mathsf{ffgview}(B_l))$').
Votes that the validator might have received from the network but have not
yet been referenced in a block are not considered.
Second, $\mathsf{HLMD}(G)$ filters for only those chains that
contain said highest justified pair, \ie, are consistent with the prior
justification
(see \cite[Algorithm 4.2]{buterin2020combining}, lines $4$ and $5$).
Third, among the remaining chains, $\mathsf{HLMD}(G)$ picks greedily the
`heaviest' chain (GHOST paradigm), \ie, the chain which among
the most recent votes for each validator has received the most votes
(LMD paradigm,
see \cite[Algorithm 4.2]{buterin2020combining}, lines $7$ to $10$).

In addition, to vote, the source and target of the Casper vote are determined
as follows (see \cite[Definition 4.7]{buterin2020combining}).
The Casper vote's source $\mathsf{LJ}$
is the last justified pair, considering only
votes that have been included in blocks on the chain determined by $\mathsf{HLMD}(G)$.
This ensures that all validators voting for the tip of a certain chain
have a consistent view of and vote from the last justified pair.
The Casper vote's target $\mathsf{LE}$
is the last epoch boundary pair (\ie, of the current epoch)
on the chain determined by $\mathsf{HLMD}(G)$.
Again, all validators voting for the tip of a certain chain
have a consistent view of and vote for the same last epoch boundary pair.

\subsubsection{How to Sway Honest Validators}
\label{sec:appendix-attack-gasper-attack-swaying}

Suppose there are two competing chains as depicted in
Figure~\ref{fig:attack-gasper-overview}.
The only time a non-trivial fork choice occurs
in $\mathsf{HLMD}(G)$
(see \cite[Algorithm 4.2]{buterin2020combining}, line $9$)
is when a validator
chooses whether to go down the `left' or the `right' chain.
This decision is based on where the majority of the most recent votes
(one per validator) fall, \emph{in the instant when $\mathsf{HLMD}(G)$ is
executed}. Thus, if half of the most recent votes are `left'
and the other half is `right', then the adversary can release a single withheld
vote to an honest validator who is just about to execute $\mathsf{HLMD}(G)$
and thereby `tip the balance' and
sway that honest validator to vote on a chain of the adversary's choosing.
Note that the adversary can release that same withheld vote to multiple honest
validators, all of which will then vote for the chain of the adversary's choice.
Furthermore, note that the adversary can release two different withheld votes to different
sets of honest validators and thus steer one group towards `left' and the other
group towards `right'. Ultimately (due to the assumption of there being
a bound $\Delta$ on the maximum network delay the adversary can inflict on a
message, and the fact that honest validators gossip about recently received
messages in an attempt to keep consistent views of the protocol execution)
the withheld votes will become known to all honest validators, but
(a) the adversary can prevent this synchronization until after the honest
validators have cast their votes by releasing the withheld votes just
before the honest validators execute $\mathsf{HLMD}(G)$, and
(b) after two withheld votes, one for `left' and one for `right', are released,
and if the honest validators either vote `left' and `right' in equal number (during epoch $0$)
or simply reaffirm their prior votes (during epoch $1$ and beyond),
then after sharing all votes with all honest validators there is still
an equal number of votes for `left' and `right', respectively.
Thus, in the next slot the adversary can release another two withheld votes
to continue keeping up the equal split of honest validators. And so on.

Swaying honest validators by releasing withheld votes selectively is the
key
technique underlying our attack.
Since the Casper votes are consistent with the GHOST votes by construction,
as long as the GHOST votes are split equally between the two chains, the Casper
votes are split equally between the two chains.
Thus, neither of the two chains will ever receive a supermajority of $\geq 2n/3$
votes as would be necessary for a justification or finalization.
Thus, no epoch boundary pair will ever get finalized and thus liveness is lost
indefinitely and with certainty, once the attack has been launched.
In the remainder of this section we describe under what sufficient condition
and with what sequence of adversarial actions
the adversary is able to affect a permanent split among honest validators
and thus a permanent loss of liveness of Gasper.

\subsubsection{Epoch \texorpdfstring{$0$}{0}: Kick-Starting the Attack}
\label{sec:appendix-attack-gasper-attack-epoch0}

The adversary waits for an opportune epoch to kick-start the attack.
For ease of exposition, we assume that epoch $0$ is opportune.
An epoch is opportune if there are enough adversarial validators in every
slot of the epoch to fill the following roles:
\begin{itemize}[leftmargin=1em]
    \item
        The proposer of slot $0$ needs to be adversarial.
        The adversarial proposer equivocates and produces two conflicting blocks
        (`left' and `right', dashed blocks $0$ and $0'$
        in Figure~\ref{fig:attack-gasper-overview})
        which it reveals to two suitably chosen subsets of the validators in slot $0$.
        Thus, the honest validators' votes are split equally between the two chains.
        (Equivocating on block production is a slashable offense and thus the stake
        corresponding to the adversarial block producer will be slashed.
        Besides this equivocation, none of the adversarial actions are slashable.
        We note that there are variants of our attack that do not require any slashable
        adversarial actions, but these variants are more involved.)
    
    \item
        For every but the last slot of epoch $0$ the adversary
        recruits two `swayers'. The role of these swayers is to withhold
        their votes in slot $i$ and release the votes selectively to subsets of
        the honest validators in slot $i+1$ in order to split the honest validators' votes equally
        between the two chains.
    
    \item
        For every slot of epoch $0$ the adversary
        recruits two more `swayers'. The role of these additional swayers is to withhold
        their votes during slot $i$ of epoch $0$ and release the votes selectively to subsets of
        the honest validators in slot $C+i$ of epoch $1$ in order to split the honest validators' votes equally
        between the two chains in epoch $1$.
        Similarly, these swayers withhold
        their votes during epoch $1$ and release the votes selectively to subsets of
        the honest validators in epoch $2$ in order to split the honest validators' votes equally
        between the two chains in epoch $2$.
        This repeats beyond epoch $2$.
    
    \item
        Finally, to achieve an equal split of honest validators' votes for every slot
        in epoch $0$, we require that every slot has an even number of honest validators.
        If a slot does not have an even number of honest validators, then the adversary
        recruits a `filler'
        (`\tikz{ \node[gasperattackplot,voter-adversarial-filler,scale=0.75] {}; }'
        in Figure~\ref{fig:attack-gasper-overview}) which behaves like an honest validator
        for the rest of the attack.
\end{itemize}

Thus, sufficient for an epoch to be opportune to start the attack
is that the following conditions are all satisfied:
\begin{enumerate}[label=(\alph*),leftmargin=1.25cm]
    \item[$\Ev_{\mathrm{(a)}}^{(0)}$:] The proposer of slot $0$ is adversarial.
    
    \item[$\Ev_{\mathrm{(b)}}^{(0)}$:] Slot $0$ has $\geq 6$ adversarial validators (the adversarial proposer, two swayers
        for epoch $0$, two swayers for epoch $1$, potentially one filler).
    
    \item[$\Ev_{\mathrm{(c)},i}^{(0)}$:] Slots $i = 1,...,(C-2)$ have $\geq 5$ adversarial validators
        (two swayers
        for epoch $0$, two swayers for epoch $1$, potentially one filler).
    
    \item[$\Ev_{\mathrm{(d)}}^{(0)}$:] Slot $(C-1)$ has $\geq 3$ adversarial validators
        (two swayers for epoch $1$, potentially one filler).
\end{enumerate}
We show in Appendix~\ref{sec:appendix-attack-gasper-analysis} that,
in particular in the regime of many validators ($n \to \infty$),
the probability that a particular epoch is opportune is approximately
equal to $\beta$,
the fraction of adversarial validators.

For slots $i=1,...,(C-1)$ of epoch $0$
the adversary uses two `swayers'to withhold
their votes in slot $i$ and release the votes selectively to equally sized subsets of
the honest validators in slot $i+1$ in order to split the honest validators' votes equally
between the two chains.
Thus, in each slot, an equal number of validators votes `left' and `right', respectively,
so that at the end of epoch $0$ both chains have equal weight.
In particular, none of the chains achieves a supermajority.
Thus, no Casper finalization can take place.

\subsubsection{Epoch \texorpdfstring{$1$}{1}:
Transition to Steady-State
}
\label{sec:appendix-attack-gasper-attack-epoch1}

During epoch $1$, the adversary uses the other group of swayers recruited in epoch $0$
to
selectively release more withheld votes from epoch $0$ to keep splitting validators
into two groups, one of which sees `left' as leading and votes for it,
the other sees `right' as leading and votes for it.
All the adversary needs to do is release withheld votes
so as to
reaffirm
the honest validators in their illusion that whatever chain they previously voted
on in epoch $0$
happens to be still leading, so that they renew their vote.
At the end of epoch $1$ there are still two chains
with equal number of votes and thus neither gets finalized.

\subsubsection{Epoch \texorpdfstring{$2$}{2} and Beyond:
Steady-State
}
\label{sec:appendix-attack-gasper-attack-epoch2}

During epoch $2$ and beyond the attack reaches steady-state
in that the adversarial actions now repeat in each epoch.
Note that the validators whose epoch $0$ votes the adversary released during epoch $1$
to sway honest validators have themselves not voted in epoch $1$ yet.
Thus, during epoch $2$ the adversary
selectively releases withheld votes from epoch $1$ to keep honest validators
split between the two chains.
Again, all the adversary needs to do is to release withheld votes such that it reaffirms
the honest validators in their illusion that whatever chain they previously voted
on in epoch $1$
happens to be still leading, so that they renew their vote.
This continues indefinitely.
Neither chain ever reaches a supermajority, thus, no Casper finalizations
take place.
As a result of this attack, the ledger of Gasper does not incorporate new transactions
and thus is not live.

\subsection{Analysis \& Simulation}
\label{sec:appendix-attack-gasper-analysis-and-simulation}
\label{sec:appendix-attack-gasper-analysis}

We analyze the probability $\Prob{\Ev^{(\ell)}}$ that an adversary can launch the attack in epoch $\ell$.
Without loss of generality, we consider $\ell = 0$.
Recall that the
events $\Ev_{\mathrm{(a)}}^{(0)}$ to $\Ev_{\mathrm{(d)}}^{(0)}$
are sufficient for the adversary to be able
to launch the attack.
Obviously,
\begin{IEEEeqnarray}{rCl}
    \Prob{\cEv_{\mathrm{(a)}}^{(0)}} &=& 1 - \beta.
\end{IEEEeqnarray}

For fixed $C$ and large $n$ such that $\beta n/C \geq 6$,
due to tail bounds for the hypergeometric distribution,
\begin{IEEEeqnarray}{rCl}
    \Prob{\cEv_{\mathrm{(b)}}^{(0)}}, \Prob{\cEv_{\mathrm{(c)},i}^{(0)}}, \Prob{\cEv_{\mathrm{(d)}}^{(0)}} &\leq& \exp\left(- \Theta(n) \right)
\end{IEEEeqnarray}

Thus, with a straightforward application of the union bound,
\begin{IEEEeqnarray}{rCl}
    \Prob{\Ev^{(\ell)}} &\geq& \beta - C \exp\left(- \Theta(n) \right).
\end{IEEEeqnarray}

Note that, since the events $\Ev^{(\ell_1)}$ and $\Ev^{(\ell_2)}$
of the adversary being able to kick-start
the attack in two epochs $\ell_1 \neq \ell_2$ are independent,
the number of epochs until the first epoch in which the adversary
can kick-start the attack follows a geometric distribution
with mean $1 / \Prob{\Ev^{(0)}}$.
It is thus exponentially unlikely (in the number of epochs considered)
that the adversary is not able to kick-start the attack in any
of a number of epochs,
even for small $\beta$.
As soon as an opportune epoch occurs and
the adversary can kick-start the attack, liveness
is prevented with certainty, assuming that the networking assumptions
given in Appendix~\ref{sec:appendix-attack-gasper-introduction-assumptions} are satisfied.
\label{sec:appendix-attack-gasper-simulation}

\begin{figure}
    \centering%
    \begin{tikzpicture}[]%
        \begin{axis}[
            mysimpleplot,
            xmode=log,
            xtick={50,100,200,500,1000},
            log ticks with fixed point,
            ymode=normal,
            xlabel={Committee size $n/C$},
            ylabel={$\Prob{\Ev^{(\ell)}}$},
            ymin=0, ymax=0.4,
            legend columns=2,
            transpose legend,
            height=0.4\linewidth,
            width=0.9\linewidth,
            y tick label style={
                /pgf/number format/.cd,
                    fixed,
                    precision=2,
                /tikz/.cd
            },
            x tick label style={
                /pgf/number format/.cd,
                    fixed,
                    precision=2,
                /tikz/.cd
            },
        ]

            \def\DATAPREFIX{./figures/pgfplot-opportune-epochs-2}

            \addlegendimage{myparula11,mark=none};
            \addlegendentry{$\beta=0.1$};
            
            \addlegendimage{myparula21,mark=none};
            \addlegendentry{$\beta=0.3$};
            
            \addlegendimage{myparula11,gray};
            \addlegendentry{$C=8$};
            
            \addlegendimage{myparula12,gray};
            \addlegendentry{$C=16$};
            
            \addlegendimage{myparula13,gray};
            \addlegendentry{$C=32$};
            
            \addlegendimage{myparula14,gray};
            \addlegendentry{$C=64$};

            \addplot [myparula11] table [x=N_multiple,y=probability]
            {\DATAPREFIX/beta0.1-C8.dat};
            \label{leg:simulation-beta0.1-C8}

            \addplot [myparula12] table [x=N_multiple,y=probability]
            {\DATAPREFIX/beta0.1-C16.dat};
            \label{leg:simulation-beta0.1-C16}

            \addplot [myparula13] table [x=N_multiple,y=probability]
            {\DATAPREFIX/beta0.1-C32.dat};
            \label{leg:simulation-beta0.1-C32}

            \addplot [myparula14] table [x=N_multiple,y=probability]
            {\DATAPREFIX/beta0.1-C64.dat};
            \label{leg:simulation-beta0.1-C64}

            \addplot [myparula21] table [x=N_multiple,y=probability]
            {\DATAPREFIX/beta0.3-C8.dat};
            \label{leg:simulation-beta0.3-C8}

            \addplot [myparula22] table [x=N_multiple,y=probability]
            {\DATAPREFIX/beta0.3-C16.dat};
            \label{leg:simulation-beta0.3-C16}

            \addplot [myparula23] table [x=N_multiple,y=probability]
            {\DATAPREFIX/beta0.3-C32.dat};
            \label{leg:simulation-beta0.3-C32}

            \addplot [myparula24] table [x=N_multiple,y=probability]
            {\DATAPREFIX/beta0.3-C64.dat};
            \label{leg:simulation-beta0.3-C64}

        \end{axis}
    \end{tikzpicture}%
    \vspace{-0.75em}%
    \caption[]{%
        Monte Carlo estimate of the
        probability that an adversary who controls $\beta$ fraction of
        stake can launch the attack in epoch $\ell$,
        as a function of number of slots per epoch $C$
        and committee size $n/C$.
        Observe that $\Prob{\Ev^{(\ell)}} \approx \beta$
        is a good rule of thumb, even for moderate $n$.
    }
    \label{fig:attack-gasper-simulation}
\end{figure}
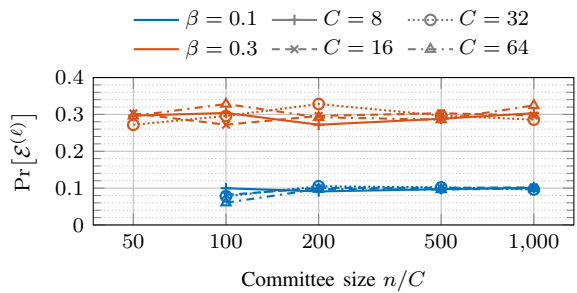

We use a Monte Carlo simulation to numerically evaluate
the probability $\Prob{\Ev^{(\ell)}}$.\footnote{%
The source code of the simulation
can be found at: \url{https://github.com/tse-group/gasper-attack}.
}
The result is shown
in Figure~\ref{fig:attack-gasper-simulation}.
We observe that
the approximation $\Prob{\Ev^{(\ell)}} \approx \beta$
is a pretty good rule of thumb, even for moderate
numbers of validators.
This matches the intuition that the probability
of successfully kick-starting the attack in a given epoch
is largely dominated by the probability
that the proposer in the first slot of the epoch is adversarial.
All further conditions are satisfied as soon as there are six
adversarial validators per each slot, which happens with high
probability as $n$ grows and $\beta$ is held fix.
\section{Analysis and Security Proof for the \SAC Construction Using Streamlet}
\label{sec:analysis-streamlet}

We prove Theorem \ref{thm:main-security} for the protocol $\PItheexample{Sleepy}{Streamlet}$ composing a permissioned longest chain protocol and Streamlet.

\begin{lemma}[Safety Lemma for $\PIbft$]
\label{thm:streamlet-safety}
(See \cite[Lemma 14, Theorem 3]{streamlet} and Algorithm~\ref{algo:pseudocode-sleepy-streamlet})
If some honest node sees a notarized chain with three adjacent \blocksbft $B_0$, $B_1$, $B_2$ with consecutive epoch numbers $e$, $e+1$, and $e+2$, then there cannot be a conflicting block $B \neq B_1$ that also gets notarized in any honest view at the same depth as $B_1$.
Hence, there cannot be conflicting \finalp \blocksbft in %
any honest view.
\end{lemma}

\begin{proof}
The proof of \cite[Lemma 14]{streamlet}, which is based on a quorum intersection argument, is unaffected by the fact that honest nodes do not vote for a proposed \blockbft if they do not view the referenced \blocklc as \finalda.
Even with the modification shown at line~\ref{algo:pseudocode-sleepy-streamlet-condition} of
Algorithm~\ref{algo:pseudocode-sleepy-streamlet}, honest nodes would not equivocate or vote for proposed \blocksbft that do not extend the longest notarized chain.
Then, via \cite[Theorem 3]{streamlet}, there cannot be conflicting \finalp \blocksbft in the views of honest nodes.
\end{proof}

By the ledger extraction
explained in
Figure~\ref{fig:ledger-extraction-details},
Lemma \ref{thm:streamlet-safety}
completes the proof of safety for $\LOGfin{}{}$.

\begin{lemma}
\label{thm:streamlet-liveness-1}
(See \cite[Lemma 5]{streamlet} and Algorithm~\ref{algo:pseudocode-sleepy-streamlet})
After $\max\{\GST,\GOT\}$, suppose there are three consecutive epochs $e$, $e+1$, and $e+2$, all with honest leaders denoted by $L_e$, $L_{e+1}$, and $L_{e+2}$, \textbf{and the leaders' proposals reference \blockslc that are viewed as \finalda by all honest nodes.} 
Then the following holds: 
(Below, let $B$ denote the block proposed by $L_{e+2}$ during epoch $e+2$.)

\begin{enumerate}[label=(\alph*)]
    
    \item By the beginning of epoch $e+3$, every honest node will observe a notarized chain ending at $B$, which was not notarized before the beginning of epoch $e$.
        
    \item No conflicting block $B' \neq B$ with the same length as $B$ will ever get notarized in honest view.

\end{enumerate}
\end{lemma}

\begin{proof}
Note that every honest node is awake and the network is $\Delta$ synchronous after $\max\{\GST,\GOT\}$.
Due to the \textbf{highlighted condition} added to the Lemma, %
all honest nodes view the \blockslc referenced by the proposals as \finalda, thus, the additional condition for an honest node to cast a vote (see line~\ref{algo:pseudocode-sleepy-streamlet-condition} of Algorithm~\ref{algo:pseudocode-sleepy-streamlet}) is satisfied. 
Then, all honest nodes behave as they would in Streamlet, and the liveness lemma \cite[Lemma 5]{streamlet} ensures the validity of (a) and (b).
\end{proof}

\begin{lemma}[Liveness Lemma for $\PIbft$]
\label{thm:streamlet-liveness}
After $\max\{\GST,\GOT\}$, suppose that there are five consecutive epochs $e,e+1,..,e+4$ with honest leaders and \textbf{the leaders' proposals reference \blockslc that are viewed as \finalda by all honest nodes.}
Then, by the beginning of epoch $e+5$, every honest node observes a new \finalp \blockbft, proposed by an honest leader, that was not \finalp at the beginning of epoch $e$.
\end{lemma}

Proof follows from Lemma~\ref{thm:streamlet-liveness-1} and \cite[Theorem 6]{streamlet}.

Notice that Lemma~\ref{thm:streamlet-liveness}, by itself, is not sufficient to show the liveness of $\LOGfin{}{}$ after $\max\{\GST,\GOT\}$ under $\AdvEnvOneOpt$, due to the highlighted condition in the lemma's statement.
In this context, the following theorem shows that after $\max\{\GST,\GOT\}$, the \blockslc referenced by honest proposals in $\PIbft$ are viewed as \finalda by all honest nodes, thus, ensuring that the \textbf{highlighted condition} in the statement of Lemma~\ref{thm:streamlet-liveness} is satisfied after $\max\{\GST,\GOT\}$.
Although, Theorem~\ref{thm:pos-security-under-advp-main} below is stated for the \emph{static} version of the longest chain protocol described in \cite{sleepy}, a similar statement can be made for \cite{david2018ouroboros}.
$\PIlc$ is initialized with a parameter $p$ which denotes the probability that any given node gets to produce a block in any given time slot.

\begin{theorem}
\label{thm:pos-security-under-advp}
\label{thm:pos-security-under-advp-main}
For all
    \begin{IEEEeqnarray}{C}
        p < \frac{n-2f}{2\Delta n(n-f)},
    \end{IEEEeqnarray}
there exists a constant\footnote{Value of $\const$ depends on $p$, $n$, $f$ and $\Delta$.} $\const>0$ such that for any $\GST$ and $\GOT$ specified by $\AdvEnvOneOpt$, $\PIlc(p)$ is secure after $\const (\max\{\GST,\GOT\}+\sigma)$, with transaction confirmation time $\Tconfirm = \sigma$, except with probability $e^{-\Omega(\sqrt{\sigma})}$.%
\footnote{Using the recursive bootstrapping argument developed in \cite[Section~4.2]{dem20}, it is possible to bring the error probability $e^{-\Omega(\sqrt{\sigma})}$ as close to an exponential decay as possible. 
In this context, for any $\epsilon>0$, it is possible to find constants $A_{\epsilon}$, $a_{\epsilon}$ such that $\PIlc(p)$ is secure after $\const \max\{\GST,\GOT\}$ with confirmation time $\Tconfirm = \sigma$ except with probability $A_{\epsilon} e^{-a_{\epsilon} \sigma^{1-\epsilon}}$.}
\end{theorem}

Full proof and the associated analysis can be found in Appendix~\ref{sec:sleepy-security}.
The proof extends the technique of {\em pivots} in \cite{sleepy} from the synchronous model to the partially synchronous model. 
The %
technique of Nakamoto blocks \cite{dem20} can be used to further strengthen the result to get an optimal bound for the block generation rate $p$ given $n$, $f$ and $\Delta$.

Finally, the following Lemma completes the proof of liveness for $\LOGfin{}{}$ after $\max\{\GST,\GOT\}$:

\begin{lemma}[Liveness Lemma for $\LOGfin{}{}$]
\label{thm:pibft-liveness}
There exists a constant $\const>0$ such that for any $\GST$ and $\GOT$ specified by $\AdvEnvOneOpt$, $\LOGfin{}{}$ is live after time $\const (\max\{\GOT,\GST\}+\sigma)$ except with probability $e^{-\Omega(\sqrt{\sigma})}$.
\end{lemma}

\begin{proof}
Via Theorem~\ref{thm:pos-security-under-advp-main}, there exists a constant $\const>0$ such that for any $\GST$ and $\GOT$ specified by $\AdvEnvOneOpt$, $\PIlc$ is safe and live, with confirmation time $\sigma$, after time $\const (\max\{\GOT,\GST\}+\sigma)$ except with probability $e^{-\Omega(\sqrt{\sigma})}$.
Hence, the following observation is true for any \blocklc $b$ except with probability $e^{-\Omega(\sqrt{\sigma})}$: 
If $b$ is first viewed as \finalda by an honest node at some time $t>\const (\max\{\GOT,\GST\}+\sigma)$, then, it will be regarded as \finalda in the views of all of the honest nodes by time $t+\Delta$.

Each \blockbft proposed by an honest leader at time $t$ references the deepest \finalda \blocklc in the view of the leader at time $t$.
Moreover, honest nodes vote $\Delta$ time into an epoch, \ie, $\Delta$ time after they see a proposal.
Hence, after time $\const(\max\{\GOT,\GST\}+\sigma)$, all of the proposals by honest leaders in $\PIbft$ reference \blockslc that are viewed as \finalda by all honest nodes when they vote, except with probability $e^{-\Omega(\sqrt{\sigma})}$.
Finally, via Lemma~\ref{thm:streamlet-liveness}, after time $\const(\max\{\GOT,\GST\}+\sigma)$, every honest node observes a new \finalp \blockbft proposed by an honest leader after all of the five consecutive honest epochs, except with probability $e^{-\Omega(\sqrt{\sigma})}$.

Next, consider a time interval $[s,s+\sigma]$ such that $s>\const(\max\{\GOT,\GST\}+\sigma)$.
Since the proposer of an epoch in $\PIbft$ is determined uniformly at random among all of the nodes, after time $\GOT$, any epoch has an honest proposer independent from other epochs, with probability at least $2/3$ under $\AdvEnvOneOpt$.
Hence, there exists a sequence of five consecutive honest epochs within the interval $[s+\sigma/2,s+\sigma]$ except with probability $e^{-\Omega(\sigma)}$.
Then, every honest node observes a new \finalp \blockbft proposed by an honest leader within the interval $[s+\sigma/2,s+\sigma]$ except with probability $e^{-\Omega(\sigma)}$.

Finally, via the liveness of $\PIlc$ after $\const(\max\{\GOT,\GST\}+\sigma)$, a transaction $\tx$ received by an awake honest node at time $s$ will be included in a \finalda \blocklc $b'$ by time $s+\sigma/2$ except with probability $e^{-\Omega(\sqrt{\sigma})}$.
Now, let $b$ denote the \finalda \blocklc referenced by the new \finalp \blockbft that was proposed by an honest node within the interval $[s+\sigma/2,s+\sigma]$. 
Via the safety $\PIlc$, we know that $b$ extends $b'$ containing the transaction $\tx$ except with probability $e^{-\Omega(\sqrt{\sigma})}$.
Consequently, any transaction received by an honest node at some time $s>\const(\max\{\GOT,\GST\}+\sigma)$ becomes part of $\LOGfin{}{}$ in the view of any honest node $i$, by time $s+\sigma$, except with probability $e^{-\Omega(\sigma)}+e^{-\Omega(\sqrt{\sigma})} = e^{-\Omega(\sqrt{\sigma})}$. 
This concludes the proof.
\end{proof}

The following Lemma shows the consistency of $\LOGfin{}{}$ with the output of $\PIlc$ under $\AdvEnvTwoOpt$, which is a necessary condition for the safety of $\LOGda{}{}$.
\begin{lemma}
\label{thm:pibft-safe-under-2}
$\LOGfin{}{}$ is a safe prefix of the output of $\PIlc$ in the view of every honest node at all times under $\AdvEnvTwoOpt$ except with probability $e^{-\Omega(\sqrt{\sigma})}$.
\end{lemma}

\begin{proof}
Via the security of $\PIlc$ under $\AdvEnvTwoOpt$, if any two honest nodes $i$ and $j$ view $b_i$ and $b_j$ as \finalda (at any time), either $b_i \preceq b_j$ or $b_j \preceq b_i$, except with probability $e^{-\Omega(\sqrt{\sigma})}$.
Moreover, for a \blockbft to become \finalp in the view of an honest node $i$ under $\AdvEnvTwoOpt$, at least one vote from an honest node is required, and honest nodes only vote for a \blockbft if they view the referenced \blocklc as \finalda.
Hence, given any two honest nodes $i$ and $j$, if \blockslc $b_i$ and $b_j$ are referenced by the \blocksbft $B_i$ and $B_j$ that are \finalp in the views of $i$ and $j$ respectively, then either $b_i \preceq b_j$ or $b_j \preceq b_i$.
This is true even if the \blocksbft $B_i$ and $B_j$ conflict with each other in the output of $\PIbft$ (see Figure~\ref{fig:examples-one-protocol-unsafe}).

Since the \blockslc referenced by \finalp \blocksbft in the view of an honest node $i$ does not conflict with the \blockslc referenced by \finalp \blocksbft in the view of any other honest node $j$ under $\AdvEnvTwoOpt$ (even when these \blocksbft might be conflicting), the ledgers $\LOGfin{i}{t}$ and $\LOGfin{j}{t'}$ also do not conflict for $i$ and $j$ at any times $t,t'$, except with probability $e^{-\Omega(\sqrt{\sigma})}$.
Finally, since the ledgers $\LOGfin{}{}$ are constructed from \finalda snapshots of the prefix of the output of $\PIlc$ which is safe, $\LOGfin{}{}$ is a safe prefix of the output of $\PIlc$ at any time and in the view of any honest node under $\AdvEnvTwoOpt$, except with probability $e^{-\Omega(\sqrt{\sigma})}$.
\end{proof}

Finally, we can start the main proof for Theorem~\ref{thm:main-security}.

\begin{proof}
We first observe via Lemma~\ref{thm:streamlet-safety} that $\PIbft$ is safe at all times under $\AdvEnvOneOpt$.
Then, since the ledger extraction for $\LOGfin{}{}$
(Section~\ref{sec:ledger-extraction-details})
preserves the safety of $\PIbft$, $\LOGfin{}{}$ is safe under $\AdvEnvOneOpt$ as well.
Second, via Lemma~\ref{thm:pibft-liveness}, there exists a constant $\const>0$ such that for any $\GST$ and $\GOT$ specified by $\AdvEnvOneOpt$, $\LOGfin{}{}$ is live after time $\const (\max\{\GOT,\GST\}+\sigma)$ except with probability $e^{-\Omega(\sqrt{\sigma})}$.
Consequently, under $\AdvEnvOneOpt$, $\LOGfin{}{}$ is safe with probability $1$ and live after time $\const (\max\{\GOT,\GST\}+\sigma)$ except with probability $e^{-\Omega(\sqrt{\sigma})}$.
This shows the property \textbf{P1}. %

Via \cite[Theorem~3, Lemma~1]{sleepy}, $\PIlc$ is secure with $\Tconfirm=\sigma$ under $\AdvEnvTwoOpt$ for any $p < (n-f)/(2\Delta n(n-f))$, except with probability $e^{-\Omega(\sqrt{\sigma})}$.
Moreover, via Lemma~\ref{thm:pibft-safe-under-2}, $\LOGfin{}{}$ is a safe prefix of the output of $\PIlc$ in the view of any honest node, under $\AdvEnvTwoOpt$ except with probability $e^{-\Omega(\sqrt{\sigma})}$.
Observe that the ledger extraction for $\LOGda{}{}$
(Section~\ref{sec:ledger-extraction-details})
preserves the liveness of $\PIlc$ and ensures the safety of $\LOGda{}{}$ as long as $\LOGfin{}{}$ is a safe prefix of the output of $\PIlc$.
Consequently, $\LOGda{}{}$ is secure under $\AdvEnvTwoOpt$, except with probability $e^{-\Omega(\sqrt{\sigma})}$.
This shows the property \textbf{P2}. %

Finally, $\LOGfin{}{}$ is always a prefix of $\LOGda{}{}$ by construction,
concluding the proof of Theorem \ref{thm:main-security}.
\end{proof}

\section{Security Proof for Longest Chain Protocol
After \texorpdfstring{$\max\{\GST,\GOT\}$}{max(GST,GAT)}}
\label{sec:sleepy-security}

In this section, we formalize and prove the fact that security of $\PIlc(p)$ is restored after $\max\{\GST,\GOT\}$ under $\AdvEnvOneOpt$ provided that $p$,
the probability that a given node gets to propose an \blocklc at a given time slot,
is sufficiently small (Theorem~\ref{thm:pos-security-under-advp-main}).
This is a prerequisite for the liveness of $\LOGfin{}{}$.

\begin{onlyonarxiv}
To understand why the security of $\PIlc$ matters for the liveness of $\LOGfin{}{}$ (see Figure \ref{fig:proof-layout-1}), consider the following two example attacks.
In the first example, before $\max\{\GST,\GOT\}$, the adversary isolates all of the honest nodes or puts them to sleep so that they cannot build a chain of \blockslc.
The adversary simultaneously builds a long and private chain with empty \blockslc.
After $\max\{\GST,\GOT\}$, honest nodes wake up and the communication between them is restored, thus, they start building a chain. 
However, whenever they release an honest \blocklc, the adversary replaces it with one of the pre-mined empty \blockslc and prompts the honest miners to mine on that empty \blocklc, thus, attacking the quality of $\PIlc$'s output chain.
In this scenario, although finalization of \blocksbft can occur in $\PIbft$, the \finalp \blocksbft only reference empty \blockslc for a long time after $\max\{\GST,\GOT\}$, implying the loss of liveness for $\LOGfin{}{}$.
In the second example, adversary builds two conflicting private chains of \blockslc before $\max\{\GST,\GOT\}$ while the honest nodes are asleepy or isolated.
After $\max\{\GST,\GOT\}$, the adversary releases these pre-mined private chains block-by-block, thus, making the honest nodes switch back and forth between the two chains.
If the adversary releases new blocks at opportune times, then the honest nodes are not able to agree on \finalda \blockslc, and thus, no finalization occurs in $\PIbft$ for a long time after $\max\{\GST,\GOT\}$.
However, since the honest nodes can collectively grow a chain of \blockslc faster than the adversary after $\max\{\GST,\GOT\}$, the adversary cannot sustain the aforementioned attacks except for a limited period of time, as it would eventually run out of private \blockslc to release.
Hence, in this case, $\PIlc$ eventually gains its safety and liveness after $\max\{\GST,\GOT\}$. 
\end{onlyonarxiv}

Before we state the main theorem for the security of $\PIlc(p)$ after $\max\{\GST,\GOT\}$ under $\AdvEnvOneOpt$, we recall the notation from Section~\ref{sec:model} and introduce %
notation from \cite[Section~4.3]{sleepy}.
Recall that $n$ denotes the total number of nodes and
$f$ denotes the number of adversarial nodes.
Let $\beta$ be the expected number of adversary nodes elected leader in any single time slot of Sleepy. 
Observe that $\beta = pf$.
Let $\alpha$ be the expected number of awake honest nodes elected leader in any single time slot of Sleepy. 
Since every node is awake after $\GOT$, $\alpha = p(n-f)$ after $\GOT$.

Since $f<n/3$ under $\AdvEnvOneOpt$, for any given $f$, $n$ and $\Delta$, $p$ can be selected such that there exist constants $0<c<1$ and $0<\Phi$ for which
    \begin{IEEEeqnarray}{C}
        2pn\Delta < 1-c, \qquad
        \frac{n-f}{f} \geq \frac{1+\Phi}{1-2pn\Delta}.
    \end{IEEEeqnarray}
(This holds for any $p$ smaller than $(n-2f)/(2\Delta n(n-f))$.)

Then, we observe that for such a  $p$, after $\GOT$, 
    \begin{IEEEeqnarray}{C}
        \beta < \alpha(1-2pn\Delta),
    \end{IEEEeqnarray}
and $\AdvEnvOneOpt$ becomes `$\PIlc(p)$-compliant' as defined in \cite[Section~4.3]{sleepy}. 
The property of $\PIlc(p)$-compliance will be useful in subsequent proofs when we directly use results from \cite{sleepy} to achieve our goals.

Informally, by adjusting $p$ above, we ensure that the honest nodes are elected leaders at time slot which are more than $\Delta$ apart from each other.
Hence, after $\max\{\GST,\GOT\}$, the honest blocks do not get mapped into the same depths in the blocktree.
As long as $f < n/2$, via such an adjustment, we can always guarantee that the chain extended by the honest nodes grow faster than any private chain grown by the adversary after $\max\{\GST,\GOT\}$.
Consequently, in the rest of this section and
in Theorem~\ref{thm:pos-security-under-advp}
we will assume that $p$ is sufficiently small so that $\beta < \alpha(1-2pn\Delta)$ and $\AdvEnvOneOpt$ is $\PIlc(p)$-compliant per \cite[Section~4.3]{sleepy}.

To prove Theorem \ref{thm:pos-security-under-advp}, we use the notion of \emph{strong pivot} defined in \cite{sleepy}.
In this context, we slightly change the definition of strong pivot given in \cite[Definition~5]{sleepy} to ensure that strong pivots force the convergence of the longest chains in views of different honest nodes when $\max\{\GST,\GOT\}>0$.
In Definition~\ref{def:pivot} below, we use the same definition for the convergence opportunity as given in \cite[Sections~2.2 and 5.2]{sleepy}.
Let $A[t_{\mathrm{a}},t_{\mathrm{b}}]$ and $C[t_{\mathrm{a}},t_{\mathrm{b}}]$ denote the number of adversarial slots and convergence opportunities respectively, between slots $t_{\mathrm{a}}$ and $t_{\mathrm{b}} \geq t_{\mathrm{a}}$.

\begin{definition}
\label{def:pivot}
A time slot $t \geq \max\{\GST,\GOT\}$ is said to be a \emph{$\GST$-strong pivot} if for any $t_{\mathrm{a}},t_{\mathrm{b}}$, $0 \leq t_{\mathrm{a}} \leq t \leq t_{\mathrm{b}}$, the number of convergence opportunities within $[\max\{t_{\mathrm{a}},\GST,\GOT\},t_{\mathrm{b}}]$ is greater than the number of adversarial slots in $[t_{\mathrm{a}},t_{\mathrm{b}}]$, \ie,
    \begin{IEEEeqnarray}{C}
        C[\max\{t_{\mathrm{a}},\GST,\GOT\},t_{\mathrm{b}}]>A[t_{\mathrm{a}},t_{\mathrm{b}}].
    \end{IEEEeqnarray}
\end{definition}

In the definition of $\GST$-strong pivots, we only count the number of convergence opportunities that happen after $\max\{\GST,\GOT\}$. 
This is because the useful properties of convergence opportunities do not hold in an asynchronous network, which is the case before $\GST$, and all honest nodes are potentially asleep before $\GOT$.

We can now focus on the proof of Theorem~\ref{thm:pos-security-under-advp}, which depends on the following propositions.
Recall that while proving the propositions below, we can assume that $\beta < \alpha(1-2pn\Delta)$ and $\AdvEnvOneOpt$ is `$\PIlc(p)$-compliant' as defined in \cite[Section~4.3]{sleepy}. 

\begin{proposition}
\label{thm:pos-convergence}
Consider two honest nodes $i$ and $j$, and, let $t$, $\max\{\GST,\GOT\} \leq t$, be a $\GST$-strong pivot.
Then, given any $r,r'$ such that $r' \geq r > t + (\sigma/\beta)$, the prefixes ending at time $t$ are the same for the longest chains seen by $i$ and $j$ at times $r$ and $r'$.
\end{proposition}
Note that every $\GST$-strong pivot is also a strong pivot as given in \cite[Definition~5]{sleepy} and the network is $\Delta$ synchronous after time $\max\{\GST,\GOT\}$.
Hence, the proof of Proposition~\ref{thm:pos-convergence} follows from the proof of Lemma~5 in \cite{sleepy}.

\begin{proposition}
\label{thm:pos-before-gst-pre}
For any $\epsilon>0$, there exist constants $C_{\epsilon}, c_{\epsilon}>0$ such that
    \begin{IEEEeqnarray}{C}
        \Prob{A[0,t] < (1+\epsilon) \beta t, \ \forall t \geq s} > 1-C_{\epsilon} e^{-c_{\epsilon} s}.
    \end{IEEEeqnarray}
\end{proposition}

\begin{proof}

We first consider the time sequence $\{t_n\}_{n\geq 0}$ given by the following formula:
    \begin{IEEEeqnarray}{C}
        t_0 = 0, \qquad t_n = \left(\frac{2+2\epsilon}{2+\epsilon}\right)^{n-1}
        \qquad
        \text{for $n \geq 1$.}
    \end{IEEEeqnarray}

Let's define $E_n$ as the event that $A[0,t_n] > (1+\epsilon)\beta t_{n-1}$, \ie, there are more than $(1+\epsilon)\beta t_{n-1}$ adversarial slots within the time interval $[0,t_n]$.
Similarly, let's define $F_s$ as the event that for any time $t \geq s$, $A[0,t] \leq (1+\epsilon)\beta t$, \ie, the number of adversarial slots within the time interval $[0,t]$ is smaller than $(1+\epsilon)\beta t$ for any $t \geq s$.

Given these definitions, we can express $\overline{F}_s$, $s>1$, in terms of the events $E_n$ as
    $\overline{F}_s \subseteq \bigcup_{n=n_s}^{\infty} E_n$,
where $n_s$ is an integer such that
    \begin{IEEEeqnarray}{C}
        \left(\frac{2+2\epsilon}{2+\epsilon}\right)^{n_s-2} \leq s < \left(\frac{2+2\epsilon}{2+\epsilon}\right)^{n_s-1}.
    \end{IEEEeqnarray}
    
We next calculate the probability of the event $E_n$.
Fact~2 in \cite{sleepy} states that for any constant $\epsilon>0$ and $t_{\mathrm{a}}, t_{\mathrm{b}}$ such that $t \triangleq t_{\mathrm{b}}-t_{\mathrm{a}} \geq 0$,
    \begin{IEEEeqnarray}{C}
        \Prob{A[t_{\mathrm{a}},t_{\mathrm{b}}]>(1+\epsilon)\beta t} \leq e^{-\frac{\epsilon^2\beta t}{3}}.
    \end{IEEEeqnarray}
Then, as
    \begin{IEEEeqnarray}{C}
        t_n = \frac{2+2\epsilon}{2+\epsilon}t_{n-1} = \frac{1+\epsilon}{1+\epsilon/2}t_{n-1},
    \end{IEEEeqnarray}
we infer that
    \begin{IEEEeqnarray}{rCl}
        \Prob{E_n} &=& \Prob{A[0,t_n] > (1+\epsilon)\beta t_{n-1}} \\
        &=& \Prob{A[0,t_n] > (1+\epsilon/2)\beta t_n} < e^{-\frac{\epsilon^2\beta t_n}{12}}.   \IEEEeqnarraynumspace
    \end{IEEEeqnarray}

Finally, using
    \begin{IEEEeqnarray}{C}
        t_{n_s} = \left(\frac{2+2\epsilon}{2+\epsilon}\right)^{n_s-1} \geq s \geq \lfloor s \rfloor,
    \end{IEEEeqnarray}
and the union bound, we observe that for any $s>1$,
    \begin{IEEEeqnarray}{rCl}
        \Prob{\overline{F}_s} &\leq& \sum_{n=n_s}^{\infty} \Prob{E_n} \leq \sum_{n=n_s}^{\infty} e^{-\frac{\epsilon^2\beta t_n}{12}} \\
        &\leq& \sum_{i=\lfloor s \rfloor}^{\infty} e^{-\frac{\epsilon^2\beta i}{12}} \leq \frac{1}{A_{\epsilon,\beta}(1-A_{\epsilon,\beta})} A_{\epsilon,\beta}^s
    \end{IEEEeqnarray}
where
    \begin{IEEEeqnarray}{C}
        A_{\epsilon,\beta} = e^{-\frac{\epsilon^2\beta}{12}} < 1
        \qquad
        \text{for any $\epsilon>0$.}
    \end{IEEEeqnarray}

We conclude the proof by setting
    \begin{IEEEeqnarray}{C}
        C_{\epsilon} = \frac{1}{A_{\epsilon,\beta}(1-A_{\epsilon,\beta})}, \qquad
        c_{\epsilon} = -\ln{(A_{\epsilon,\beta})} > 0.
        \IEEEeqnarraynumspace
    \end{IEEEeqnarray}
\end{proof}

\begin{corollary}
\label{thm:pos-before-gst}
Given any $\epsilon>0$, the following statement is true for any $s>1$ except with probability $C_{\epsilon}e^{-c_{\epsilon}s}$:
For any $\GST$ and $\GOT$ specified by $\AdvEnvOneOpt$, the number of adversarial slots by $\max\{\GST,\GOT\}$ is less than $(1+\epsilon) \beta \max\{s, \GST, \GOT\}$.
\end{corollary}

\begin{proposition}
\label{thm:pos-catch-up}
For any positive integer $\extra$, $\epsilon>0$ and times $t_0, t_1$,
there exist positive constants $\tilde{C}_{\epsilon}$ and $\tilde{c}_{\epsilon}$ such that
    \begin{IEEEeqnarray}{C}
        \Prob{A[t_0,t_1] + \extra \leq C[t_0,t_1]} \geq 1-e^{-\tilde{c}_{\epsilon} \extra}
    \end{IEEEeqnarray}
if $t \triangleq t_1 - t_0 \geq \tilde{C}_{\epsilon} \extra$.
\end{proposition}

\begin{onlyinproceedings}
Proof follows from \cite[Fact~2, Lemma~2]{sleepy} \cite[Appendix~C]{neu2020ebbandflow}.
\end{onlyinproceedings}
\begin{onlyonarxiv}
Proof follows from \cite[Fact~2, Lemma~2]{sleepy}. %
\begin{proof}
Define
    \begin{IEEEeqnarray}{C}
        \tilde{C}_{\epsilon} =  \frac{1+\epsilon}{\alpha(1-2pn\Delta) - \beta}, 
    \end{IEEEeqnarray}
and, let 
    \begin{IEEEeqnarray}{C}
        \epsilon_1 = \frac{\epsilon(\alpha(1-2pn\Delta)-\beta)}{(1+\epsilon)(\alpha(1-2pn\Delta)+\beta)}.
    \end{IEEEeqnarray}
Due to \cite[Fact~2]{sleepy}, for any $0<\epsilon_1<1$,
    \begin{IEEEeqnarray}{C}
        \Prob{A[t_0,t_1]>(1+\epsilon_1)\beta t} < e^{-\frac{\epsilon^2_1 \beta t}{3}}.
    \end{IEEEeqnarray}
Due to \cite[Lemma~2]{sleepy}, for any $\epsilon_1>0$, there exists a positive $\epsilon_2$ such that 
    \begin{IEEEeqnarray}{C}
        \Prob{C[t_0,t_1]<(1-\epsilon_1)\alpha(1-2pn\Delta) t} < e^{-\epsilon_2 \beta t}.
        \IEEEeqnarraynumspace
    \end{IEEEeqnarray}
Finally, for the values of $t$ and $\epsilon_1$ chosen above, we note that
    \begin{IEEEeqnarray}{C}
        (1-\epsilon_1)\alpha(1-2pn\Delta) t - (1+\epsilon_1)\beta t = \extra.
    \end{IEEEeqnarray}
Then, via union bound,
    \begin{IEEEeqnarray}{rCl}
        &&\Prob{A[t_0,t_1]+\extra\leq C[t_0,t_1]} \\
        &=& 1 - \Prob{A[t_0,t_1]+\extra > C[t_0,t_1]} \\
        &\geq& 1 - \Prob{A[t_0,t_1]>(1+\epsilon_1)\beta t} \nonumber\\
        && -\> \Prob{C[t_0,t_1] < (1-\epsilon_1)\alpha(1-2pn\Delta) t} \\
        &=& 1 - e^{-\epsilon_2 \beta t} - e^{-\frac{\epsilon^2_1 \beta t}{3}}
    \end{IEEEeqnarray}
where $t=O(\extra)$.
Consequently, there exists a constant $\tilde{c}_{\epsilon}$ such that
    \begin{IEEEeqnarray}{C}
        \Prob{A[t_0,t_1]+\extra \leq C[t_0,t_1]} \geq 1 - e^{-\tilde{c}_{\epsilon} \extra}.
    \end{IEEEeqnarray}
\end{proof}
\end{onlyonarxiv}

Define $T$ as the minimum time $t \geq \max\{\GST,\GOT\}$ such that the number of convergence opportunities in $[\max\{\GST,\GOT\},t]$ equals the number of adversarial slots within $[0,t]$:
    \begin{IEEEeqnarray}{C}
        T = \min_{t \geq \max\{\GST,\GOT\};\ C[\max\{\GST,\GOT\},t] = A[0,t]} t.
        \IEEEeqnarraynumspace
    \end{IEEEeqnarray}

\begin{proposition}
\label{thm:pos-normalization-time}
There exists a constant $\const$ such that for any given security parameter $\sigma$ and $\GST$, $\GOT$ specified by $\AdvEnvOneOpt$,
    \begin{IEEEeqnarray}{C}
        T \leq \const(\max\{\GST,\GOT\}+\sigma)
    \end{IEEEeqnarray}
except with probability $e^{-\Omega(\sigma)}$.
\end{proposition}

\begin{proof}
From Corollary~\ref{thm:pos-before-gst}, we know that given a constant $\epsilon>0$, the following statement is true for any $s>1$ except with probability $C_{\epsilon}e^{-c_{\epsilon}s}$:
For any $\GST$ and $\GOT$ specified by $\AdvEnvOneOpt$, the number of adversarial slots by $\max\{\GST,\GOT\}$, $A[0,\max\{\GST,\GOT\}]$, is less than $(1+\epsilon) \beta \max\{s, \GST, \GOT\}$.
Moreover, Proposition~\ref{thm:pos-catch-up} implies that for any positive integer $\extra$ and $\epsilon>0$,
there exist positive constants $\tilde{C}_{\epsilon}$ and $\tilde{c}_{\epsilon}$ such that
    \begin{IEEEeqnarray}{C}
        \Prob{A[0,t] + \extra \leq C[0,t]} \geq 1-e^{-\tilde{c}_{\epsilon} \extra}
    \end{IEEEeqnarray}
where $t = \tilde{C}_{\epsilon} \extra$.

Next, we fix some $\epsilon>0$ and set $s=\sigma$ where $\sigma$ is our security parameter.
Then, for any $\GST$ and $\GOT$ specified by $\AdvEnvOneOpt$, the number of adversarial slots by $\max\{\GST,\GOT\}$ is upper bounded by 
    \begin{IEEEeqnarray}{rCl}
        && (1+\epsilon) \beta \max\{\sigma, \GST,\GOT\} \nonumber\\
        &\leq& (1+\epsilon) \beta (\sigma + \max\{\GST,\GOT\})
    \end{IEEEeqnarray}
except with probability $e^{-\Omega(\sigma)}$. 
Furthermore, setting
    \begin{IEEEeqnarray}{rCl}
        \extra = (1+\epsilon)\beta(\sigma + \max\{\GST,\GOT\}),
    \end{IEEEeqnarray}
we can assert that
    \begin{IEEEeqnarray}{rCl}
        && \Prob{A[0,t] \leq C[\max\{\GST,\GOT\},t]} \nonumber\\
        &\geq& 1 - e^{-\tilde{c}_{\epsilon} \sigma} - C_{\epsilon}e^{-c_{\epsilon}s} = 1-e^{-\Omega(\sigma)}
    \end{IEEEeqnarray}
for
    \begin{IEEEeqnarray}{rCl}
        t &=& \max\{\GST,\GOT\}   \nonumber\\
        &&+\> \tilde{C}_{\epsilon} (1+\epsilon) \beta (\sigma + \max\{\GST,\GOT\}).
    \end{IEEEeqnarray}

Finally, we conclude that for any $\GST$ and $\GOT$ specified by $\AdvEnvOneOpt$, $C[\max\{\GST,\GOT\},t] \geq A[0,t]$ for
    \begin{IEEEeqnarray}{C}
        t = \GST + \tilde{C}_{\epsilon} (1+\epsilon)^2 \beta (\sigma + \max\{\GST,\GOT\})
    \end{IEEEeqnarray}
except with probability $e^{-\Omega(\sigma)}$.
Hence, there is a constant $\const>0$ such that for any given security parameter $\sigma$, $\GST$ and $\GOT$,
    \begin{IEEEeqnarray}{C}
        T \leq \const(\max\{\GST,\GOT\}+\sigma)
    \end{IEEEeqnarray}
except with probability $e^{-\Omega(\sigma)}$.
\end{proof}

Finally, we have all the components to start the proof of Theorem~\ref{thm:pos-security-under-advp}. 
The proof uses the same concepts as $(T_G,g_0,g_1)$-chain growth, $(T_Q,\mu)$-chain quality and $T_C$-safety introduced in Sections~3.2.1, 3.2.2 and 3.2.3 of \cite{sleepy}, respectively.

\begin{proof}
First, recall the definition of $T$ as the minimum time $t \geq \max\{\GST,\GOT\}$ such that $C[\max\{\GST,\GOT\},t] = A[0,t]$.
Due to Proposition~\ref{thm:pos-normalization-time}, there exists a constant $\const>0$ such that for any given security parameter $\sigma$,
    \begin{IEEEeqnarray}{C}
        T \leq \const(\max\{\GST,\GOT\}+\sigma)
    \end{IEEEeqnarray}
except with probability $e^{-\Omega(\sigma)}$.

From \cite[Theorem~5, Corollary~4]{sleepy}, we know that within any time period $[s,t]$ such that $t-s$ is a polynomial of $\sigma$, there exists a strong pivot as given in \cite[Definition~5]{sleepy} except with probability $e^{-\Omega(\sqrt{\sigma})}$.
Observe that if $s>T$, then any strong pivot in the interval $[s,t]$ is also a $\GST$-strong pivot.
Consequently, within any time period $[s,t]$ such that $s>\const(\max\{\GST,\GOT\}+\sigma)$, there exists a $\GST$-strong pivot except with probability $e^{-\Omega(\sqrt{\sigma})} + e^{-\Omega(\sigma)} = e^{-\Omega(\sqrt{\sigma})}$.

Via Proposition~\ref{thm:pos-convergence}, a $\GST$-strong pivot at time $t$ forces the convergence of the longest chains seen by all honest nodes up till some time $t-O(1)$.
Then, using \cite[Theorem~7]{sleepy}, Proposition~\ref{thm:pos-convergence} and the observations above, we infer that $\PIlc(p)$ is $\sigma$-consistent after time $\const(\max\{\GST,\GOT\}+\sigma)$ except with probability $e^{-\Omega(\sqrt{\sigma})}$.
Moreover, $\sigma$-consistency of $\PIlc(p)$ after time $\const(\max\{\GST,\GOT\}+\sigma)$ implies, through \cite[Lemmas~3, 4 and 8]{sleepy}, that for any $\epsilon>0$, $\PIlc(p)$ satisfies $(\sigma,g_0,g_1)$-chain growth and $(\sigma,\mu)$-chain quality after time $\const(\max\{\GST,\GOT\}+\sigma)$, except with probability $e^{-\Omega(\sqrt{\sigma})}$, where $g_0$, $g_1$ and $\mu$ are constants that depend on the parameters of $\PIlc(p)$ and \AdvEnvOneOpt.
Specifically, $g_0=(1-\epsilon)\alpha(1-2pn\Delta)$.

Finally, using \cite[Lemma~1]{sleepy} and its proof, we conclude that if $\PIlc(p)$ satisfies $(T_G,g_0,g_1)$-chain growth, $(T_Q,\mu)$-chain quality and $T_C$-safety after time $\const(\max\{\GST,\GOT\}+\sigma)$, then, it is secure with confirmation time
    \begin{IEEEeqnarray}{C}
        \Tconfirm \leq O\!\left( \frac{T_G+T_Q+T_C}{g_0}+\Delta \right),
    \end{IEEEeqnarray}
after time $\const(\max\{\GST,\GOT\}+\sigma)$.
Consequently, $\PIlc(p)$ is secure with confirmation time
    \begin{IEEEeqnarray}{C}
        \Tconfirm \leq O\!\left(\frac{3\sigma}{(1-\epsilon)\alpha(1-2pn\Delta)}+\Delta\right) = O(\sigma),   \IEEEeqnarraynumspace
    \end{IEEEeqnarray}
after time $\const(\max\{\GST,\GOT\}+\sigma)$ except with probability $e^{-\Omega(\sqrt{\sigma})}$. 
This concludes the proof.
\end{proof}

\begin{onlyonarxiv}
\section{Analysis and Security Proof for the \SAC Construction Using HotStuff}
\label{sec:appendix-hotstuff-analysis}

In this section, we prove Theorem \ref{thm:main-security} for the protocol $\PItheexample{Sleepy}{HotStuff}$ composing a permissioned longest chain protocol and HotStuff.
Note that the safety and liveness proofs for HotStuff as presented in \cite{yin2018hotstuff} remain unaffected by the composition with Sleepy. 
Hence, using \cite[Lemma 1, Theorem 2, Lemma 3, Theorem 4]{yin2018hotstuff}, we can replace the safety and liveness lemmas for $\PIbft$ given in Section \ref{sec:analysis} by the following lemmas derived from \cite{yin2018hotstuff} under the model $\AdvEnvOneOpt \triangleq \AdvEnvParameterized{1}{\frac{1}{3}}$.
\begin{lemma}[Safety Lemma for $\PIbft$]
\label{thm:hotstuff-safety}
If $B_1$ and $B_2$ are two conflicting \blocksbft, then they cannot be both \finalp in the view of any honest node.
\end{lemma}
Proof is by \cite[Lemma 1, Theorem 2]{yin2018hotstuff}, which remain unaffected by the composition.
Lemma \ref{thm:hotstuff-safety} shows the safety of $\PIbft$ at all times.

\begin{lemma}[Liveness Lemma for $\PIbft$]
\label{thm:hotstuff-liveness-1}
There exists a bounded time period $T_{\mathrm{f}}$ after $\max\{\GST,\GOT\}$ such that if all honest nodes remain in some view $v$ during $T_{\mathrm{f}}$ and $v$ has an honest leader, then a new \blockbft becomes \finalp over $v$.
\end{lemma}
Since the network delay is bounded and all of the honest nodes are awake after $\max\{\GST,\GOT\}$, the proof follows from \cite[Lemma 3, Theorem 4]{yin2018hotstuff}.

Observe that the proof of Theorem \ref{thm:pos-security-under-advp-main} stays the same since we use the same $\PIlc$ protocol as Section \ref{sec:sleepy-streamlet}.
Hence, combining Lemma \ref{thm:hotstuff-liveness-1} and Theorem \ref{thm:pos-security-under-advp-main}, we can assert the liveness of $\LOGfin{}{}$ after $\max\{\GST,\GOT\}$ as shown below.

\begin{lemma}[Liveness Lemma for $\LOGfin{}{}$]
\label{thm:hotstuff-liveness}
There exists a constant $\const>0$ such that for any $\GST$ and $\GOT$ specified by $\AdvEnvOneOpt$, $\LOGfin{}{}$ is live after time $\const (\max\{\GOT,\GST\}+\sigma)$ except with probability $e^{-\Omega(\sqrt{\sigma})}$.
\end{lemma}

\begin{proof}
Via Theorem \ref{thm:pos-security-under-advp-main}, there exists a constant $\const>0$ such that for any $\GST$ and $\GOT$ specified by $\AdvEnvOneOpt$, $\PIlc$ is safe and live, with confirmation time $\sigma$, after time $\const (\max\{\GOT,\GST\}+\sigma)$ except with probability $e^{-\Omega(\sqrt{\sigma})}$.
Hence, the following observation is true for any \blocklc $b$ except with probability $e^{-\Omega(\sqrt{\sigma})}$: 
If $b$ is first viewed as \finalda by an honest node at some time $t>\const (\max\{\GOT,\GST\}+\sigma)$, then, it will be regarded as \finalda in the views of all of the honest nodes by time $t+\Delta$.

Now, if an honest leader sends a message that points to a \blockbft $B$ at some time $t$ and in some view $v$, then the \blocklc referenced by $B$ must be \finalda in the view of this leader at time $t$.
Then, by the above observation, if $t>\const (\max\{\GOT,\GST\}+\sigma)$, all honest nodes would see the \blocklc referenced by $B$ as \finalda and add $B$ to their blocktrees, by time $t+\Delta$, except with probability $e^{-\Omega(\sqrt{\sigma})}$.
Hence, after time $\const (\max\{\GOT,\GST\}+\sigma)$, the requirements outlined in line 12 of Algorithm \ref{algo:pseudocode-sleepy-hotstuff} can be modeled by a $\Delta$ delay.
In other words, every \blockbft pointed by the message of an honest node enters the blocktree of every honest node at most $\Delta$ time after the first such message.

Via Lemma \ref{thm:hotstuff-liveness-1}, there exists a bounded time period $T_{\mathrm{f}}$ after $\max\{\GST,\GOT\}$ such that if all honest nodes remain in some view $v$ during $T_{\mathrm{f}}$ and $v$ has an honest leader, then a new \blockbft becomes \finalp over $v$.
Then, we can assert the following statement for $\PIbft$ except with probability $e^{-\Omega(\sqrt{\sigma})}$:
If all honest nodes remain in some view $v$ during a time period $[s,s+T_{\mathrm{f}}]$ such that $s>\const (\max\{\GOT,\GST\}+\sigma)$ and $v$ has an honest leader, then a new \blockbft becomes \finalp over $v$.
Since HotStuff implements a round robin leader section and an exponential back-off mechanism for view change, there will be a view $v$ with an honest leader within a constant time $T_{\mathrm{bounded}}$ after $\const (\max\{\GOT,\GST\}+\sigma)$ such that the honest nodes will remain in view $v$ for longer than time $T_{\mathrm{f}}$.

Finally, let $\sigma > 2(T_{\mathrm{bounded}}+T_{\mathrm{f}})$ and consider a time interval $[s,s+\sigma]$ such that $s>\const (\max\{\GOT,\GST\}+\sigma)$.
Observe that since $\sigma/2>T_{\mathrm{bounded}}+T_{\mathrm{f}}$ and $s>\const (\max\{\GOT,\GST\}+\sigma)$, a new \blockbft $b$ becomes \finalp in the interval $[s+\sigma/2,s+\sigma]$ except with probability $e^{-\Omega(\sqrt{\sigma})}$. 
Moreover, via the liveness of $\PIlc$ after $\const(\max\{\GOT,\GST\}+\sigma)$, a transaction $tx$ received by an awake honest node at time $s$ will be included in a \finalda \blocklc $b'$ in the view of all honest nodes by time $s+\sigma/2$ except with probability $e^{-\Omega(\sqrt{\sigma})}$.
Via the safety of $\PIlc$, we know that $b$ extends $b'$ containing the transaction $tx$ except with probability $e^{-\Omega(\sqrt{\sigma})}$.
Consequently, any transaction received by an honest node at some time $s>\const(\max\{\GOT,\GST\}+\sigma)$ becomes part of the ledger $\LOGfin{}{}$ in the view of any honest node $i$ by time $s+\sigma$, except with probability $e^{-\Omega(\sigma)}+e^{-\Omega(\sqrt{\sigma})} = e^{-\Omega(\sqrt{\sigma})}$. 
This concludes the proof.
\end{proof}

Finally, recall Figure \ref{fig:proof-layout-1}, and observe that Lemma \ref{thm:hotstuff-liveness-1} (box 2) and Theorem \ref{thm:pos-security-under-advp-main} (box 3) imply Lemma \ref{thm:hotstuff-liveness} (box 4) whereas the Lemmas \ref{thm:hotstuff-safety} (box 1) and \ref{thm:hotstuff-liveness} imply the security of $\LOGfin{}{}$ outputted $\PItheexample{Sleepy}{Streamlet}$ (box 5).
Moreover, the proof of the security of $\LOGda{}{}$ stays the same as we use the same $\PIlc$ protocol as Section \ref{sec:sleepy-streamlet}.
Hence, we conclude the proof of Theorem \ref{thm:main-security} for $\PItheexample{Sleepy}{HotStuff}$. 
\end{onlyonarxiv}

\begin{onlyonarxiv}
\section{Bouncing Attack on Casper FFG}
\label{sec:attacks-casper}

Applications of \CasperFFG are two-tiered. %
A blockchain serves as dynamically available block proposal mechanism,
and \CasperFFG is a voting-based BFT-style overlay protocol
to add finalization on top of
said blockchain.
Usually, only some
`checkpoint' blocks
are candidates
for finalization, \eg, blocks at depths that are multiples of $100$.
First, a checkpoint
becomes `justified' once two-thirds vote for it. %
Subsequently, roughly speaking, a justified checkpoint becomes finalized once
two-thirds vote for a direct child checkpoint of the justified checkpoint.
To ensure consistency among the two tiers, %
the fork choice rule of the blockchain is modified
to always respect %
`the justified checkpoint of the greatest [depth]' \cite{buterin2017casper}.
There is thus a bidirectional interaction between the block proposal
and the
finalization layer:
blocks proposed by the blockchain are input to finalization,
while justified %
checkpoints constrain future block proposals.
This bidirectional interaction
is intricate to reason about
and a gateway for liveness attacks.

\begin{figure}
    \centering
    \begin{tikzpicture}
        \footnotesize
        
        \tikzset{finalitygadget/.style={
                x=0.3cm,
                y=0.6cm,
                _block/.style = {
                    draw,
                    shade,
                    top color=white,
                    bottom color=black!10,
                },
                block/.style = {
                    _block,
                    minimum width=0.475cm,
                    minimum height=0.45cm,
                },
                ancestor/.style = {
                    _block,
                    minimum width=0.3cm,
                    minimum height=0.3cm,
                },
                link/.style = {
                },
                accepted/.style = {
                    draw=red,
                    thick,
                },
                almost-accepted/.style = {
                    draw=red,
                    thick,
                    densely dotted,
                },
                fc-label/.style = {
                    align=center,
                    myParula01Blue,
                    anchor=north,
                },
                fc-arrow/.style = {
                    -Latex,
                    myParula01Blue,
                },
            }
        }

        \begin{scope}[finalitygadget,xshift=0cm]
            \coordinate (genesis) at (0,0.5);
            
            \node (b11) at (0,0) [ancestor] {};
            \node (b21) at (-1,-1) [block,accepted] {70};
            \node (b22) at (+1,-1) [block] {30};
            \node (b31) at (-1,-2) [block] {30};
            \node (b32) at (+1,-2) [block,almost-accepted,alias=attackable1] {60};
            
            \draw [link] (b11) -- (genesis);
            \draw [link] (b21) -- (b11);
            \draw [link] (b22) -- (b11);
            \draw [link] (b31) -- (b21);
            \draw [link] (b32) -- (b22);
            
            \node [fc-label] (fc) at (-1,-3) {Fork\\choice};
            \draw [fc-arrow] (fc) -- (b31);
            
        \end{scope}

        \begin{scope}[x=0.25cm,y=0.25cm,xshift=0.85cm,yshift=-0.75cm]
            \draw [draw=none,fill=black!10] (-0.5,1) -- (0,1) -- (0,2) -- (1,0) -- (0,-2) -- (0,-1) -- (-0.5,-1) -- cycle;
        \end{scope}

        \begin{scope}[finalitygadget,xshift=1.8cm]
            \coordinate (genesis) at (0,0.5);
            
            \node (b11) at (0,0) [ancestor] {};
            \node (b21) at (-1,-1) [block,accepted] {70};
            \node (b22) at (+1,-1) [block] {30};
            \node (b31) at (-1,-2) [block] {30};
            \node (b32) at (+1,-2) [block,almost-accepted] {60};
            \node (b41) at (-1,-3) [block,almost-accepted,alias=attackable2] {60};
            
            \draw [link] (b11) -- (genesis);
            \draw [link] (b21) -- (b11);
            \draw [link] (b22) -- (b11);
            \draw [link] (b31) -- (b21);
            \draw [link] (b32) -- (b22);
            \draw [link] (b41) -- (b31);
            
            \node [fc-label] (fc) at (-1,-3.75) {Fork\\choice};
            \draw [fc-arrow] (fc) -- (b41);
            
        \end{scope}

        \begin{scope}[finalitygadget,xshift=0cm]
            \node (attackable) at (1,-5) [anchor=north,align=center] {Adversary can release\\ $f=10$ votes to justify};
            \draw [-Latex] (attackable) -- (attackable1);
            \draw [-Latex] (attackable) -- (attackable2);
        \end{scope}

        \begin{scope}[x=0.25cm,y=0.25cm,xshift=2.65cm,yshift=-0.75cm]
            \draw [draw=none,fill=black!10] (-0.5,1) -- (0,1) -- (0,2) -- (1,0) -- (0,-2) -- (0,-1) -- (-0.5,-1) -- cycle;
        \end{scope}

        \begin{scope}[finalitygadget,xshift=3.6cm]
            \coordinate (genesis) at (0,0.5);
            
            \node (b11) at (0,0) [ancestor] {};
            \node (b21) at (-1,-1) [block,accepted] {70};
            \node (b22) at (+1,-1) [block] {30};
            \node (b31) at (-1,-2) [block] {30};
            \node (b32) at (+1,-2) [block,accepted] {70};
            \node (b41) at (-1,-3) [block,almost-accepted] {60};
            \node (b42) at (+1,-3) [block] {0};
            
            \draw [link] (b11) -- (genesis);
            \draw [link] (b21) -- (b11);
            \draw [link] (b22) -- (b11);
            \draw [link] (b31) -- (b21);
            \draw [link] (b32) -- (b22);
            \draw [link] (b41) -- (b31);
            \draw [link] (b42) -- (b32);
            
            \node [fc-label] (fc) at (+1,-4) {Fork\\choice};
            \draw [fc-arrow] (fc) -- (b42);
            
        \end{scope}

        \begin{scope}[x=0.25cm,y=0.25cm,xshift=4.45cm,yshift=-0.75cm]
            \draw [draw=none,fill=black!10] (-0.5,1) -- (0,1) -- (0,2) -- (1,0) -- (0,-2) -- (0,-1) -- (-0.5,-1) -- cycle;
        \end{scope}

        \begin{scope}[finalitygadget,xshift=5.4cm]
            \coordinate (genesis) at (0,0.5);
            
            \node (b11) at (0,0) [ancestor] {};
            \node (b21) at (-1,-1) [block,accepted] {70};
            \node (b22) at (+1,-1) [block] {30};
            \node (b31) at (-1,-2) [block] {30};
            \node (b32) at (+1,-2) [block,accepted] {70};
            \node (b41) at (-1,-3) [block,almost-accepted] {60};
            \node (b42) at (+1,-3) [block] {30};
            
            \draw [link] (b11) -- (genesis);
            \draw [link] (b21) -- (b11);
            \draw [link] (b22) -- (b11);
            \draw [link] (b31) -- (b21);
            \draw [link] (b32) -- (b22);
            \draw [link] (b41) -- (b31);
            \draw [link] (b42) -- (b32);
            
            \node [fc-label] (fc) at (+1,-4) {Fork\\choice};
            \draw [fc-arrow] (fc) -- (b42);
            
        \end{scope}

        \begin{scope}[x=0.25cm,y=0.25cm,xshift=6.25cm,yshift=-0.75cm]
            \draw [draw=none,fill=black!10] (-0.5,1) -- (0,1) -- (0,2) -- (1,0) -- (0,-2) -- (0,-1) -- (-0.5,-1) -- cycle;
            
            \node at (2.5,-0.5) {\large ...};
        \end{scope}

    \end{tikzpicture}%
    \caption{By releasing withheld \CasperFFG votes late, the adversary can force honest validators to adopt a competing chain, due to the modification of the fork choice rule to respect `the justified checkpoint of the greatest [depth]'.
    Over longer periods of time, the adversary forces honest validators to switch back and forth between a `left' and a `right' chain
    and thus liveness of finalizations is disrupted.}
    \label{fig:bouncing-attack-casper-visualization}
\end{figure}
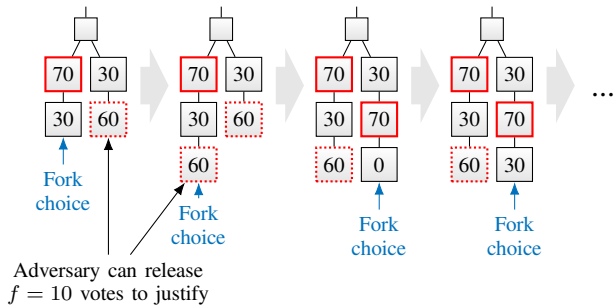

The \emph{bouncing attack} \cite{ethresearch-bouncing-attack,ethresearch-bouncing-attack-analysis}
exploits this bidirectional interaction
to attack liveness of the overall protocol
as follows
(see Figure~\ref{fig:bouncing-attack-casper-visualization}).
Suppose there are two competing chains, `left' and `right',
with checkpoints shown as squares in Figure~\ref{fig:bouncing-attack-casper-visualization}.
A square's label represents the number of votes for that checkpoint,
in a system with $n=100$ total and $f=10$ adversarial validators.
The initial setting of blocks and votes could be produced, \eg,
during a period of asynchrony in which the adversary controls message delivery in its favor. %
`Left'
has the deepest justified checkpoint and is thus
chosen by the fork choice rule of honest validators.
At the same time, `right'
has a deeper checkpoint which is not yet justified
but can be justified by the adversary whenever it casts its $f=10$ votes for the respective checkpoint depth.
Once `left'
advances to a new checkpoint depth, and accumulates enough votes so that the adversary
could again justify that new checkpoint in the future by releasing its $f=10$ votes, the adversary
releases its votes for the competing checkpoint of `right'
on the previous checkpoint depth.
The deepest justified checkpoint is now
on `right',
and honest validators %
switch to propose new blocks on `right'. %
Note that the chains are now already set up such that the adversary can bounce honest validators
back to `left' %
once `right'
advances to a new deepest checkpoint depth.

As a result, a single brief period of asynchrony
suffices to set
the consensus system up such that both chains
grow in parallel indefinitely.
No checkpoint will ever be finalized,
the protocol stalls. %
What is more, since the fork choice flip-flops between the two chains, the underlying blockchain
is rendered unsafe by the modified fork choice rule.
The bidirectional interdependency of \CasperFFG and the %
blockchain
gives the adversary
major leverage over honest nodes on the proposal layer
and thus enables this attack.

In contrast,
an isolated partially synchronous BFT-style protocol,
akin to \CasperFFG,
would have eventually recovered from the period of asynchrony
and regained liveness, while remaining safe throughout.
Similarly,
an isolated
typical
dynamically available longest-chain protocol
with intact fork choice rule
could have suffered from security violations
during and shortly after the period of asynchrony,
but would have `healed' eventually, \ie, from some point on,
no more safety violations occur and transactions
get included in the ledger.

\end{onlyonarxiv}

\end{document}